\documentclass[pra,twocolumn,nofootinbib,superscriptaddress]{revtex4-2}%
\usepackage{amsfonts}
\usepackage{amsmath}
\usepackage{amssymb}
\usepackage{diagbox}
\usepackage{multirow}
\usepackage{array}
\usepackage{graphicx}
\usepackage{mathtools}
\usepackage{forest}
\usetikzlibrary{shapes.geometric,positioning}
\tikzset{elli/.style={ellipse,draw}}
\tikzset{rect/.style={rectangle,draw}}
\usepackage[ colorlinks = true,              linkcolor = blue, urlcolor  =
blue,              citecolor = red,              anchorcolor = green,
]{hyperref}%
\providecommand{\U}[1]{\protect\rule{.1in}{.1in}}
\newcommand{\PreserveBackslash}[1]{\let\temp=\\#1\let\\=\temp}
\newcolumntype{C}[1]{>{\PreserveBackslash\centering}p{#1}}
\newcolumntype{R}[1]{>{\PreserveBackslash\raggedleft}m{#1}}
\newcommand{\outerprod}[2]{\vert #1 \rangle\!\langle #2 \vert}
\newcommand{\ket}[1]{\vert #1 \rangle}
\newcommand{\bra}[1]{\langle #1 \vert}
\newcommand{\outerproj}[1]{\vert #1 \rangle\!\langle #1 \vert}
\newcommand{\im}{\mathrm{i}}
\newcommand{\NHS}[1]{\frac{1}{\sqrt{2}}\left\Vert #1 \right\Vert_2}
\newtheorem{theorem}{Theorem}

\newtheorem{algorithm}{Algorithm}

\newtheorem{corollary}[theorem]{Corollary}

\newtheorem{lemma}[theorem]{Lemma}

\newtheorem{problem}{Problem}
\newtheorem{proposition}{Proposition}
\newtheorem{remark}[theorem]{Remark}

\newcolumntype{P}[1]{>{\raggedright\arraybackslash}p{#1}}
\newcolumntype{L}{>{\centering\arraybackslash}m{2.6cm}}
\newcolumntype{M}[1]{>{\centering\arraybackslash}m{#1}}
\newenvironment{proof}[1][Proof]{\noindent\textbf{#1.} }{\ \rule{0.5em}{0.5em}}
\begin{document}
\preprint{ }
\title{Estimating distinguishability measures on quantum computers}
\author{Soorya Rethinasamy}
\affiliation{Hearne Institute for Theoretical Physics, Department of Physics and Astronomy,
and Center for Computation and Technology, Louisiana State University, Baton
Rouge, Louisiana 70803, USA}
\affiliation{School of Applied and Engineering Physics, Cornell University, Ithaca, New York 14850, USA}

\author{Rochisha Agarwal}
\affiliation{Department of Physics, Indian Institute of Technology Roorkee, Roorkee,
Uttarakhand, India}
\affiliation{School of Electrical and Computer Engineering, Cornell University, Ithaca, New York 14850, USA}

\author{Kunal Sharma}
\affiliation{Joint Center for Quantum Information and Computer Science, University of
Maryland, College Park, Maryland 20742, USA}
\affiliation{IBM Quantum, IBM T.J. Watson Research Center, Yorktown Heights, New York 10598, USA}

\author{Mark M.~Wilde}
\affiliation{Hearne Institute for Theoretical Physics, Department of Physics and Astronomy,
and Center for Computation and Technology, Louisiana State University, Baton
Rouge, Louisiana 70803, USA}
\affiliation{School of Electrical and Computer Engineering, Cornell University, Ithaca, New York 14850, USA}

\begin{abstract}
The performance of a quantum information processing protocol is ultimately judged by distinguishability measures that quantify how distinguishable the actual result of the protocol is from the ideal case. The most prominent distinguishability measures are those based on the fidelity and trace distance, due to their physical interpretations. In this paper, we propose and review several algorithms for estimating distinguishability measures based on trace distance and fidelity. The algorithms can be used for distinguishing quantum states, channels, and strategies (the last also known in the literature as ``quantum combs''). The fidelity-based algorithms offer novel physical interpretations of these distinguishability measures in terms of the maximum probability with which a single prover (or competing provers) can convince a verifier to accept the outcome of an associated computation. We simulate many of these algorithms by using a variational approach with parameterized quantum circuits. We find that the simulations converge well in both the noiseless and noisy scenarios, for all examples considered. Furthermore, the noisy simulations exhibit a parameter noise resilience. Finally, we establish a strong relationship between various quantum computational complexity classes and distance estimation problems.

\end{abstract}
\date{\today}
\maketitle

\tableofcontents
\allowdisplaybreaks

\section{Introduction}

In quantum information processing, it is essential to quantify the performance
of protocols by using distinguishability measures. It is typically the case
that there is an ideal state to prepare or an ideal channel to simulate, but
in practice, we can only realize approximations, due to experimental error.
Two commonly employed distinguishability measures for states are the trace
distance \cite{H67,H69} and the fidelity \cite{U76}. The former has an
operational interpretation as the distinguishing advantage in the optimal
success probability when trying to distinguish two states that are chosen
uniformly at random. The latter has an operational meaning as the maximum
probability that a purification of one state could pass a test for being a
purification of the other (this is known as Uhlmann's transition probability
\cite{U76}). These distinguishability measures have generalizations to quantum
channels, in the form of the diamond distance \cite{Kit97}\ and the fidelity
of channels \cite{GLN04}, as well as to strategies (sequences of channels), in
the form of the strategy distance \cite{CDP08a,CDP09,G12} and the fidelity of
strategies \cite{Gutoski2018fidelityofquantum}. Each of these measures are
generalized by the generalized divergence of states \cite{PV10}, channels
\cite{LKDW18}, and strategies \cite{Wang2019a}. The operational
interpretations of these latter distinguishability measures are similar to the
aforementioned ones, but the corresponding protocols involve more steps that
are used in the distinguishing process.

Both the trace distance and the fidelity can be computed by means of
semi-definite programming \cite{Wat13}, so that they can be estimated
accurately with a run-time that is polynomial in the dimension of the states.
The same is true for the diamond distance \cite{Wat09}, fidelity of channels
\cite{Yuan2017,KW20}, the strategy distance \cite{CDP08a,CDP09,G12}, and the
fidelity of strategies \cite{Gutoski2018fidelityofquantum}. While this method
of estimating these quantities is reasonable for states, channels, and
strategies of small dimension, its computational complexity actually increases
exponentially with the number of qubits involved, due to the well-known fact
that Hilbert-space dimension grows exponentially with the number of qubits.

In this paper, we provide several quantum algorithms for estimating these
distinguishability measures. Some of the algorithms rely on interaction with a
quantum prover, in which case they are not necessarily efficiently computable
even on a quantum computer. In fact, the computational hardness results of
\cite{W02,RW05,W09zkqa} lend credence to the belief that estimating these
quantities reliably is not generally possible in polynomial time on a quantum
computer. However, as we show in our paper, by replacing the quantum prover
with a parameterized circuit (see \cite{CABBEFMMYCC20,bharti2021noisy}\ for
reviews of variational algorithms), it is possible in some cases to estimate
these quantities reliably. Identifying precise conditions under which a
quantum computer can estimate these quantities efficiently is an interesting
open question that we leave for future research. Already in \cite{WZCGFY21},
it was shown that estimating the fidelity of two quantum states is possible in
quantum polynomial time when one of the states is low rank, and the same is the case for estimating the trace distance under certain promises \cite{WGLZY22,WZ23}. See also
\cite{CPCC19, CSZW20, TV21}\ for variational algorithms that estimate fidelity
of states and \cite{CSZW20,LLSL21} for variational algorithms to estimate
trace distance. It is open to determine precise conditions under which
estimation is possible for channel and strategy distinguishability measures.

We perform noiseless and noisy simulations of several of the algorithms provided. We find that in the noiseless scenario, all algorithms converge, for the examples considered, to the true known value of the distinguishability measure under consideration. In the noisy simulations, the algorithms converge well, and the parameters obtained exhibit a noise resilience, as put forward in \cite{Sharma_2020}; i.e., the relevant quantity can be accurately estimated by inputting the parameters learned from the noisy simulator into the noiseless simulator.

Lastly, we discuss the computational complexity of various distance estimation algorithms. We prove that several fidelity and distance estimation algorithms are complete for well-known quantum complexity classes (see \cite{W09, VW15}\ for reviews of quantum computational complexity theory). In particular, we prove that estimating the fidelity between two pure states, a mixed state and a pure state, and estimating the Hilbert--Schmidt distance of two mixed states are BQP-complete problems. These aforementioned results follow by demonstrating that there is an efficient quantum algorithm for these tasks and by showing a reduction from an arbitrary BQP\ algorithm to one for these tasks. Thus, if we believe that there is a separation between the computational power of classical and quantum computers, then these estimation problems are those for which a quantum computer has an advantage. Several BQP-complete promise problems are known, including approximating the Jones polynomial \cite{AJL06}, estimating quadratically signed weight enumerators \cite{KL01}, estimating diagonal entries of powers of sparse matrices \cite{JW07}, a problem related to matrix inversion \cite{HHL09}, and deciding whether a pure bipartite state is entangled \cite{GHMW15}. See \cite{Zhang2012} for a 2012 review of BQP-complete promise problems. 

We then prove that the problem of estimating the fidelity between a channel with arbitrary input and a pure state is a QMA-complete promise problem. We show this by constructing an efficient quantum algorithm, augmented by a single all-powerful prover, to solve this problem, and by showing a reduction from an arbitrary QMA problem to one for this task. Lastly, we demonstrate that the problem of estimating the fidelity between a channel with separable input and a pure state is QMA(2)-complete. QMA(2) is the class of problems that can be efficiently solved when augmented by two all-powerful quantum provers who are guaranteed to be unentangled \cite{KMY01,HM10}.

In the rest of the paper, we provide details of the algorithms and results mentioned above. In particular, our paper proceeds as follows:

\begin{enumerate}
\item The various subsections of Section~\ref{sec:est-fidelity} are about estimating
the fidelity of states, channels, and strategies. We begin in
Section~\ref{sec:fid-pure-states}\ by establishing two quantum algorithms for
estimating the fidelity of pure states, one of which is based on a state
overlap test (Algorithm~\ref{alg:fid-pure-states-1}) and another that employs
Bell state preparation and measurement along with a controlled unitary
(Algorithm~\ref{alg:fid-pure-states-2}).

\item In Section~\ref{sec:fid-pure-mixed}, we generalize
Algorithm~\ref{alg:fid-pure-states-1} to estimate the fidelity of a pure state
and a mixed state (see Algorithm~\ref{alg:fid-pure-mixed}).

\item In Section~\ref{sec:fid-arbitrary-states}, we establish several quantum
algorithms for estimating the fidelity of two arbitrary states.
Algorithm~\ref{alg:fid-states} generalizes
Algorithm~\ref{alg:fid-pure-states-2}.
Algorithm~\ref{alg:mixed-state-swap-test} generalizes the well-known swap test
to the case of arbitrary states. Algorithm~\ref{alg:mixed-state-Bell-tests} is a
variational algorithm that employs Bell measurements, as a generalization of
the approach in \cite{GC13, Suba__2019} for pure states.
Algorithm~\ref{alg:mixed-state-FC-meas-min} is another variational algorithm
that attempts to simulate a fidelity-achieving measurement, such as the
Fuchs--Caves measurement \cite{FC95}, in order to estimate the fidelity.

\item In Section~\ref{sec:fid-channels}, we generalize
Algorithm~\ref{alg:fid-states} to a quantum algorithm for estimating the
fidelity of quantum channels (see Algorithm~\ref{alg:fid-channels}). This
algorithm involves interaction with competing quantum provers, and
interestingly, its acceptance probability is directly related to the
fidelity of channels, thus giving the latter an operational meaning. Later, we
replace the provers with parameterized circuits and arrive at a method for
estimating the fidelity of channels.

\item In Section~\ref{sec:strategy-fidelity}, we generalize the aforementioned
approach in order to estimate the fidelity of strategies (a strategy is a
sequence of quantum channels and thus generalizes the notion of a quantum channel).

\item In Section~\ref{sec:alt-methods-fid-channels}, we briefly discuss
alternative methods for estimating the fidelity of channels and strategies,
based on the approaches from Section~\ref{sec:fid-arbitrary-states} for
estimating the fidelity of states.

\item Section~\ref{sec:max-output-fid-channels} introduces a method for
estimating the maximum output fidelity of two quantum channels, which has an
application to generating a fixed point of a quantum channel (as discussed later on in Section~\ref{sec:fixed-points-of-channels}).

\item In Sections~\ref{sec:mult-states} and \ref{sec:mult-channels-strategies}%
, we generalize the whole development above to the case of testing similarity
of arbitrary ensembles of states, channels, or strategies. We find that the
acceptance probability of the corresponding algorithms is related to the
secrecy measure from \cite{KRS09}, which can be understood as a measure of
similarity of the states in an ensemble. We then establish generalizations of
this measure for an ensemble of channels and an ensemble of strategies and
remark how this has applications in private quantum reading \cite{BDW18,DBW20}.

\item We then move on in Section~\ref{sec:TD-based-measures} to estimating
trace-distance-based measures, for states, channels, and strategies. We stress
that these various algorithms were already known, and our goal here is to
investigate their performance using a variational approach. In
Sections~\ref{sec:est-trace-dist-states}, \ref{sec:est-diamond-dist}, and \ref{sec:est-strategy-dist},
Algorithms~\ref{alg:trace-distance-states}, \ref{alg:diamond-distance-channels}, and \ref{alg:strategy-distance-strategies} provide methods for estimating the
trace distance of states, 
the diamond distance of channels, and the strategy distance of strategies, respectively.

\item In Section~\ref{sec:est-min-TD-channels}, we provide two different but
related algorithms for estimating the minimum trace distance between two
quantum channels. The related approaches employ competing provers to do so.

\item In Section~\ref{sec:multiple-TD-based}, we generalize the whole development for trace-distance based algorithms to the case of multiple states, channels, and strategies.

\item In Section~\ref{sec:num-sims}, we discuss the results of numerical simulations of Algorithms~\ref{alg:fid-states}--\ref{alg:fid-channels}, Algorithms~\ref{alg:trace-distance-states}--\ref{alg:diamond-distance-channels}, and Algorithm~\ref{alg:mult-state-disc-qubit}. We use both noiseless and noisy quantum simulators and a variational approach with parameterized circuits. 

\item In Section~\ref{sec:compFidAlgs}, we prove that the problems of evaluating the fidelity between two pure states, a pure state and a mixed state, and evaluating the Hilbert--Schmidt distance of two mixed states are BQP-complete (Theorem~\ref{thm:BQP-comp-fid-pure},\ \ref{thm:BQP-comp-fid},\ \ref{thm:BQP-comp-HSD}). We then show that the problem of evaluating the fidelity between a channel with arbitrary input and a pure state is QMA-complete (Theorem~\ref{thm:QMA-comp-fid}). Finally, we demonstrate that  the problem of evaluating the fidelity between a channel with separable input and a pure state is QMA(2)-complete (Theorem~\ref{thm:QMA(2)-comp-fid}).

\item In Section~\ref{sec:fixed-points-of-channels}, we discuss how
Algorithm~\ref{alg:fid-channels-single-prover} can generate a fixed-point
state or an approximate fixed-point state of a quantum channel.

\end{enumerate}

\noindent We finally conclude in Section~\ref{sec:conclusion} with a summary and some open questions.

\section{Estimating fidelity}
\label{sec:est-fidelity}

\renewcommand{\arraystretch}{2.2}
\begin{table*}
\begin{center}
    \begin{tabular}{|M{3.6cm}||M{2.5cm}|M{3cm}|M{8.5cm}|}
        \hline 
         \multirow{2}{*}{Problem} & 
         \multirow{2}{*}{Algorithms} & 
         \multirow{2}{*}{Approach} &
         \multirow{2}{*}{Comparison}\\
         & & & \\
        \hline
        \multirow{2}{*}{$F(\psi, \phi)$} & Algorithm~\ref{alg:fid-pure-states-1} & 
        State Overlap & 
        \multirow{2}{=}{\centering Algorithm~\ref{alg:fid-pure-states-1} is simpler than Algorithm~\ref{alg:fid-pure-states-2}. Algorithm~\ref{alg:fid-pure-states-2} generalizes in a straightforward manner to testing fidelity of mixed states.} \\
        \cline{2-3}
        & Algorithm~\ref{alg:fid-pure-states-2} & Bell-State Overlap & \\
        \hline
        $F(\psi, \rho)$ & Algorithm~\ref{alg:fid-pure-mixed} & State Overlap & - \\
        \hline
        \multirow[c]{4}{*}[-0.2cm]{$F(\rho_0, \rho_1)$} & Algorithm~\ref{alg:fid-states} & 
        Bell-State Overlap & 
        \multirow[c]{4}{=}[-0.4cm]{\centering Algorithm~\ref{alg:fid-states} is a generalization of Algorithm~\ref{alg:fid-pure-states-2} for mixed state inputs. Algorithm~\ref{alg:mixed-state-swap-test} uses a controlled SWAP gate to generalize the SWAP Test. Requires more qubits, but no controlled unitaries to generate the states being tested. Algorithm~\ref{alg:mixed-state-Bell-tests} uses a variational unitary on the reference system of one state only. Algorithms~\ref{alg:fid-states}, \ref{alg:mixed-state-swap-test} and \ref{alg:mixed-state-Bell-tests} are based on learning the Uhlmann unitary and provides a lower bound. Algorithm~\ref{alg:mixed-state-FC-meas-min} is based on learning the optimal Fuchs--Caves measurement and provides an upper bound.} \\
        \cline{2-3}
        & \multirow{1}{*}[0.1cm]{Algorithm~\ref{alg:mixed-state-swap-test}} & Generalized SWAP Test & \\
        \cline{2-3}
        & Algorithm~\ref{alg:mixed-state-Bell-tests} & Bell Measurement & \\
        \cline{2-3}
        & \multirow{1}{*}[0.1cm]{Algorithm~\ref{alg:mixed-state-FC-meas-min}} & Fuchs--Caves Measurement & \\
        \hline
        $F(\mathcal{N}_0, \mathcal{N}_1)$ & Algorithm~\ref{alg:fid-channels} & Bell-State Overlap & - \\
        \hline
        $F(\mathcal{N}^{0, (n)}, \mathcal{N}^{1, (n)})$ & Algorithm~\ref{alg:fid-strategies} & Bell-State Overlap & - \\
        \hline
        $F_{\text{max}} (\mathcal{N}_0,\mathcal{N}_1)$ & Algorithm~\ref{alg:fid-channels-single-prover} & Bell-State Overlap & - \\
        \hline
        $p_{\text{sim}}(\left\{  p(x),\rho^{x}\right\}_{x\in\mathcal{X}})$ & Algorithm~\ref{alg:fid-multiple-states} & Bell-State Overlap & Generalization of Algorithm~\ref{alg:fid-states} to ensemble of states.\\
        \hline
        $p_{\text{sim}}(\{p(x),\mathcal{N}^{x}\}_{x\in\mathcal{X}})$ & Algorithm~\ref{alg:fid-multiple-channels} & Bell-State Overlap & Generalization of Algorithm~\ref{alg:fid-channels} to ensemble of channels. \\
        \hline
        $p_{\text{sim,max}}(\{p(x),\mathcal{N}^{x}\}_{x\in\mathcal{X}})$ & Algorithm~\ref{alg:fid-multiple-channels-single-prover} & Bell-State Overlap & 
        Generalization of Algorithm~\ref{alg:fid-channels-single-prover} to ensemble of channels. \\
        \hline
    \end{tabular}
\caption{List of fidelity problems and algorithms addressed in this work. Approach used for each algorithm and comparison within a type of fidelity problem is also presented.}
\label{tab:FidAlgs}
\end{center}
\end{table*}

In this section, we propose algorithms for several different fidelity problems. A summary of all algorithms presented in this section is available in Table~\ref{tab:FidAlgs}.

\subsection{Estimating fidelity of pure states}

\label{sec:fid-pure-states}

We begin by outlining two simple quantum algorithms for estimating fidelity
when both states are pure. A standard approach for doing so is to use the swap
test \cite{BBD+97,BCWW01} or Bell measurements \cite{GC13,Suba__2019}. The approaches
that we discuss below are different from these approaches. The first algorithm
is a special case of that proposed in \cite{W02} (see also \cite{CSZW20}), as
well as a special case of Algorithm~\ref{alg:fid-pure-mixed}\ presented later. 
The second algorithm involves a Bell-state preparation and projection, as well
as controlled interactions, and it is a special case of
Algorithm~\ref{alg:fid-states}\ presented later. We list both of these
algorithms here for completeness and because later algorithms build upon them.

Suppose that the goal is to estimate the fidelity of pure states $\psi^{0}$
and $\psi^{1}$, and we are given access to quantum circuits $U^{0}$ and
$U^{1}$ that prepare these states when acting on the all-zeros state. We now
detail a first quantum algorithm for estimating the fidelity%
\begin{equation}
F(\psi^{0},\psi^{1})\coloneqq \left\vert \langle\psi^{1}|\psi^{0}%
\rangle\right\vert ^{2}. \label{eq:fid-pure-states-def}%
\end{equation}

\begin{algorithm}
\label{alg:fid-pure-states-1}The algorithm proceeds as follows:

\begin{enumerate}
\item Act with the circuit $U^{0}$ on the all-zeros state $|0\rangle$.

\item Act with $U^{1\dag}$ and perform a measurement of all qubits in the
computational basis.

\item Accept if and only if the all-zeros outcome is observed.
\end{enumerate}
\end{algorithm}

\begin{figure}
\begin{center}
\includegraphics[
width=\linewidth
]{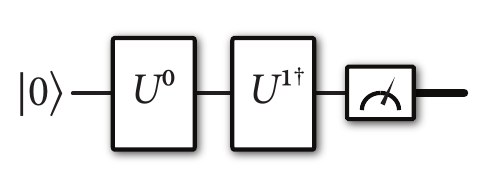}
\end{center}
\caption{This figure depicts Algorithm~\ref{alg:fid-pure-states-1} for
estimating the fidelity of pure states generated by quantum circuits $U^{0}$
and~$U^{1}$. In this, and all following figures, we use the convention that a bold line represents a classical register.}%
\label{fig:fid-pure-states}%
\end{figure}

Algorithm~\ref{alg:fid-pure-states-1} is depicted in
Figure~\ref{fig:fid-pure-states}. The acceptance probability of
Algorithm~\ref{alg:fid-pure-states-1}\ is precisely equal to $\left\vert
\langle0|U^{1\dag}U^{0}|0\rangle\right\vert ^{2}$, which by definition is
equal to the fidelity in \eqref{eq:fid-pure-states-def}. In fact,
Algorithm~\ref{alg:fid-pure-states-1}\ is a quantum computational
implementation of the well known operational interpretation of the fidelity as
the probability that the state $\psi^{0}$ passes a test for being the state
$\psi^{1}$.

Our next quantum algorithm for estimating fidelity makes use of a Bell-state
preparation and projection. Its acceptance probability is equal to%
\begin{equation}
\frac{1}{2}\left(  1+\sqrt{F}(\psi^{0},\psi^{1})\right)
\label{eq:acc-prob-simple-pure-state-alg}%
\end{equation}
and thus gives a way to estimate the fidelity through repetition. It is a
variational algorithm that optimizes over a phase $\phi$ and makes use of the
fact that%
\begin{equation}
\max_{\phi\in\lbrack0,2\pi]}\operatorname{Re}[e^{\im\phi}\langle\psi^{0}%
|\psi^{1}\rangle]=|\langle\psi^{0}|\psi^{1}\rangle|. \label{eq:optimize-phase}%
\end{equation}
This can be seen from the fact that the optimal phase $\phi$ picked is such that
\begin{equation}
e^{\im\phi} = \frac{\langle\psi^{1}|\psi^{0}\rangle}{|\langle\psi^{1}|\psi^{0}\rangle|}.
\end{equation}

Let $S$ denote the quantum system in which the states $\psi^{0}$ and $\psi
^{1}$ are prepared.

\begin{algorithm}
\label{alg:fid-pure-states-2}The algorithm proceeds as follows:

\begin{enumerate}
\item Prepare a Bell state%
\begin{equation}
|\Phi\rangle_{T^{\prime}T}\coloneqq\frac{1}{\sqrt{2}}(|00\rangle_{T^{\prime}%
T}+|11\rangle_{T^{\prime}T})
\end{equation}
on registers $T^{\prime}$ and $T$ and prepare system $S$ in the all-zeros
state $|0\rangle_{S}$.

\item Using the circuits $U_{S}^{0}$ and $U_{S}^{1}$, perform the following
controlled unitary:%
\begin{equation}
\sum_{i\in\left\{  0,1\right\}  }|i\rangle\!\langle i|_{T}\otimes U_{S}^{i}.
\label{eq:alg-pure-bell-c-Ui}%
\end{equation}

\item Act with the following unitary on system $T'$:%
\begin{equation}%
\begin{bmatrix}
1 & 0\\
0 & e^{\im\phi}%
\end{bmatrix}
.
\end{equation}

\item Perform a Bell measurement%
\begin{equation}
\{\Phi_{T^{\prime}T},I_{T^{\prime}T}-\Phi_{T^{\prime}T}\}
\end{equation}
on systems $T^{\prime}$ and $T$. Accept if and only if the outcome
$\Phi_{T^{\prime}T}$ occurs.
\end{enumerate}
\end{algorithm}

\begin{figure}
\begin{center}
\includegraphics[
width=\linewidth
]{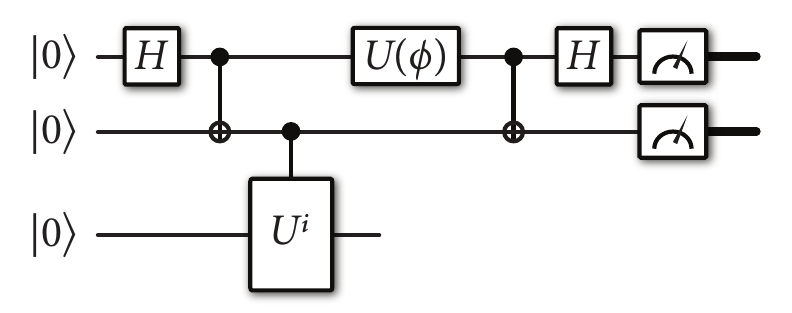}
\end{center}
\caption{This figure depicts Algorithm~\ref{alg:fid-pure-states-2} for
estimating the fidelity of pure states generated by quantum circuits $U^{0}$
and~$U^{1}$. The third gate with $U^{i}$ in the box is defined
in~\eqref{eq:alg-pure-bell-c-Ui}.}%
\label{fig:fid-pure-states-2}%
\end{figure}

Figure~\ref{fig:fid-pure-states-2} depicts
Algorithm~\ref{alg:fid-pure-states-2}. After Step~3 of
Algorithm~\ref{alg:fid-pure-states-2}, the overall state is as follows:%
\begin{equation}
\frac{1}{\sqrt{2}}\sum_{j\in\left\{  0,1\right\}  }|jj\rangle_{T^{\prime}%
T}e^{\im j\phi}|\psi^{j}\rangle_{S},
\end{equation}
and the acceptance probability is equal to%
\begin{align}
&  \left\Vert \langle\Phi|_{T^{\prime}T}\left(  \frac{1}{\sqrt{2}}\sum
_{j\in\left\{  0,1\right\}  }|jj\rangle_{T^{\prime}T}e^{\im j\phi}|\psi
^{j}\rangle_{S}\right)  \right\Vert _{2}^{2}\nonumber\\
&  =\frac{1}{4}\left\Vert \sum_{j,k\in\left\{  0,1\right\}  }\langle
kk|jj\rangle_{T^{\prime}T}e^{\im j\phi}|\psi^{j}\rangle_{S}\right\Vert _{2}^{2}\\
&  =\frac{1}{4}\left\Vert \sum_{j\in\left\{  0,1\right\}  }e^{\im j\phi}|\psi
^{j}\rangle_{S}\right\Vert _{2}^{2}\\
&  =\frac{1}{4}\left(  2+2\operatorname{Re}[e^{\im\phi}\langle\psi^{0}|\psi
^{1}\rangle]\right)  .
\end{align}
By choosing the optimal phase $\phi$ in \eqref{eq:optimize-phase}, we find
that the acceptance probability is equal to the expression
in~\eqref{eq:acc-prob-simple-pure-state-alg}. Note that, through repetition,
we can execute Algorithm~\ref{alg:fid-pure-states-2}\ in a variational way to
learn the optimal value of $\phi$. 

Later on, in Section~\ref{sec:compFidAlgs}, we prove that a promise version of
the problem of estimating the fidelity between two pure states is a BQP-complete promise problem.

\subsection{Estimating fidelity when one state is pure and the other is mixed}

\label{sec:fid-pure-mixed}

In this section, we outline a simple quantum
algorithm that estimates the fidelity between a mixed state $\rho_{S}$ and a
pure state $\psi_{S}$. It is a straightforward generalization of
Algorithm~\ref{alg:fid-pure-states-1}. 

Let $U_{RS}^{\rho}$ be a quantum circuit that generates a purification
$\varphi_{RS}$ of $\rho_{S}$ when acting on the all-zeros state of systems
$RS$, and let $U_{S}^{\psi}$ be a circuit that generates $\psi_{S}$ when
acting on the all-zeros state.

\begin{algorithm}
\label{alg:fid-pure-mixed}The algorithm proceeds as follows:

\begin{enumerate}
\item Act on the all-zeros state$~|0\rangle_{RS}$ with the circuit~$U_{RS}^{\rho}$.

\item Act with $U_{S}^{\psi\dag}$ on system $S$ and perform a measurement of
all qubits of system $S$ in the computational basis.

\item Accept if and only if the all-zeros outcome is observed.
\end{enumerate}
\end{algorithm}

\begin{figure}
\begin{center}
\includegraphics[
width=\linewidth
]{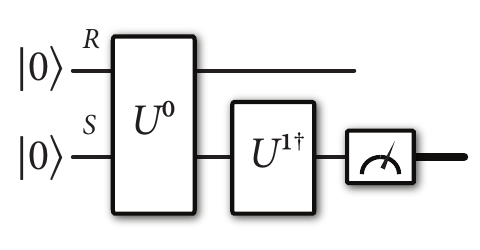}
\end{center}
\caption{This figure depicts Algorithm~\ref{alg:fid-pure-mixed} for estimating
the fidelity of a mixed state generated by a quantum circuit $U^{0}$ and a
pure state generated by~$U^{1}$.}%
\label{fig:fid-pure-mixed}%
\end{figure}

Figure~\ref{fig:fid-pure-mixed} depicts Algorithm~\ref{alg:fid-pure-mixed}.
The acceptance probability of Algorithm~\ref{alg:fid-pure-mixed} is equal to
the fidelity $F(\psi,\rho)=\langle\psi|\rho|\psi\rangle$, which follows
because%
\begin{align}
&  \left\Vert \langle0|_{S}U_{S}^{\psi\dag}U_{RS}^{\rho}|0\rangle
_{RS}\right\Vert _{2}^{2}\nonumber\\
&  =\operatorname{Tr}[\left(  I_{R}\otimes|\psi\rangle\!\langle\psi
|_{S}\right)  |\varphi\rangle\!\langle\varphi|_{RS}]\\
&  =\operatorname{Tr}[|\psi\rangle\!\langle\psi|_{S}\rho_{S}]\\
&  =\langle\psi|\rho|\psi\rangle.
\end{align}

We note here that it is not strictly necessary to have access to the reference
system~$R$ of $|\varphi\rangle_{RS}$ in order to execute
Algorithm~\ref{alg:fid-pure-mixed}. It is only necessary to have some method
of generating the reduced state $\rho_{S}$.

Later on, in Section~\ref{sec:compFidAlgs}, we prove that a promise version of
the problem of estimating the fidelity of a pure state and a mixed state is a BQP-complete promise problem.

\subsection{Estimating fidelity of arbitrary\ states}

\label{sec:fid-arbitrary-states}In this section, we outline several quantum
algorithms for estimating the fidelity of arbitrary states on a quantum
computer, some of which involve an interaction with a quantum prover (more
precisely, the algorithms involving interaction with a prover are QSZK
algorithms, where QSZK stands for ``quantum statistical zero knowledge''
\cite{W02,W09zkqa}). The algorithms are different from the algorithm proposed
in \cite{W02} (as also considered in \cite{CSZW20}), which is based on
Uhlmann's formula for fidelity \cite{U76}.

Suppose that the goal is to estimate the fidelity of states $\rho_{S}^{0}$ and
$\rho_{S}^{1}$, defined as \cite{U76}%
\begin{equation}
F(\rho_{S}^{0},\rho_{S}^{1})\coloneqq \left\Vert \sqrt{\rho_{S}^{0}}\sqrt
{\rho_{S}^{1}}\right\Vert _{1}^{2},
\end{equation}
where the trace norm of an operator $A$ is defined as $\left\Vert A
\right\Vert _{1} \coloneqq \operatorname{Tr}[\sqrt{A^{\dag}A}]$. Suppose also
that we are given access to quantum circuits $U_{RS}^{0}$ and $U_{RS}^{1}$
that prepare purifications $\psi_{RS}^{0}$ and $\psi_{RS}^{1}$ of $\rho
_{S}^{0}$ and $\rho_{S}^{1}$, respectively, when acting on the all-zeros state
$|0\rangle_{RS}$. Let us recall Uhlmann's formula for fidelity \cite{U76}:%
\begin{equation}
F(\rho_{S}^{0},\rho_{S}^{1})=\max_{|\psi^{0}\rangle_{RS},|\psi^{1}\rangle
_{RS}}\left\vert \langle\psi^{1}|\psi^{0}\rangle_{RS}\right\vert ^{2},
\label{eq:uhlmann-thm}%
\end{equation}
where the optimization is over all purifications $\psi_{RS}^{0}$ and
$\psi_{RS}^{1}$ of $\rho_{S}^{0}$ and $\rho_{S}^{1}$, respectively. We note
here that the fidelity can be computed by means of a semi-definite program
\cite{Wat13}. Also, the promise version of this problem, involving
descriptions of quantum circuits as input, is a QSZK-complete promise problem
\cite{W02}, where QSZK\ stands for quantum statistical zero knowledge (see
\cite{W02,W09zkqa} for details of this complexity class). Thus, it is unlikely
that anyone will find a general-purpose efficient quantum algorithm for
estimating fidelity (i.e., one that does not involve interaction with an
all-powerful prover).

We note that the algorithms in this subsection need the purification of the state of interest to be provided. In scenarios where the purification of a state is not available, there exist variational algorithms to learn the purification \cite{osti_1961893, CSZW20}.

\subsubsection{Controlled unitary and Bell state overlap}

\label{sec:contr-U-bell-overlap-alg}We now detail a QSZK algorithm for
estimating the following quantity:%
\begin{equation}
\frac{1}{2}\left(  1+\sqrt{F}(\rho_{S}^{0},\rho_{S}^{1})\right)  .
\end{equation}
It is a QSZK\ algorithm because, in the case that the fidelity $\sqrt{F}%
(\rho_{S}^{0},\rho_{S}^{1})\approx1$, the verifier does not learn anything by
interacting with the prover (i.e., the verifier only learns that the algorithm
accepts with high probability). This  algorithm is somewhat similar to
the quantum algorithm proposed in \cite{PhysRevA.94.022310}, which was used
for estimating a quantity known as fidelity of recovery \cite{SW14}. It is
also similar to the algorithm described in Figure~3 of \cite{KW00}. It can be
understood as a generalization of Algorithm~\ref{alg:fid-pure-states-2}\ from
pure states to arbitrary states.

\begin{algorithm}
\label{alg:fid-states} The algorithm proceeds as follows:

\begin{enumerate}
\item The verifier prepares a Bell state%
\begin{equation}
|\Phi\rangle_{T^{\prime}T}\coloneqq\frac{1}{\sqrt{2}}(|00\rangle_{T^{\prime}%
T}+|11\rangle_{T^{\prime}T})
\end{equation}
on registers $T^{\prime}$ and $T$ and prepares systems $RS$ in the all-zeros
state $|0\rangle_{RS}$.

\item Using the circuits $U_{RS}^{0}$ and $U_{RS}^{1}$, the verifier performs
the following controlled unitary:%
\begin{equation}
\sum_{i\in\left\{  0,1\right\}  }|i\rangle\!\langle i|_{T}\otimes U_{RS}^{i}.
\label{eq:controlled-U-fid-alg-bell-overlap}
\end{equation}

\item The verifier transmits systems $T^{\prime}$ and $R$ to the prover.

\item The prover prepares a system $F$ in the $|0\rangle_{F}$ state and acts
on systems $T^{\prime}$, $R$, and $F$ with a unitary $P_{T^{\prime
}RF\rightarrow T^{\prime\prime}F^{\prime}}$ to produce the output systems
$T^{\prime\prime}$ and $F^{\prime}$, where $T^{\prime\prime}$ is a qubit system.

\item The prover sends system $T^{\prime\prime}$ to the verifier, who then
performs a Bell measurement%
\begin{equation}
\{\Phi_{T^{\prime\prime}T},I_{T^{\prime\prime}T}-\Phi_{T^{\prime\prime}T}\}
\end{equation}
on systems $T^{\prime\prime}$ and $T$. The verifier accepts if and only if the
outcome $\Phi_{T^{\prime\prime}T}$ occurs.
\end{enumerate}
\end{algorithm}

\begin{figure*}
\begin{center}
\includegraphics[
width=\linewidth
]{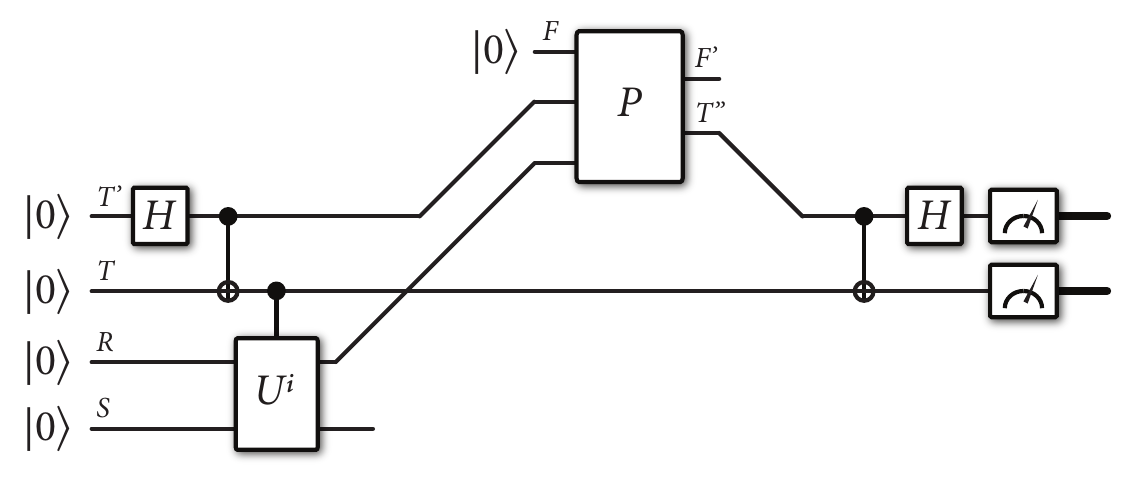}
\end{center}
\caption{This figure depicts Algorithm~\ref{alg:fid-states} for estimating the
fidelity of mixed states generated by quantum circuits $U^{0}_{RS}$ and~$U^{1}_{RS}$.}%
\label{fig:fid-alg-mixed-states-1}%
\end{figure*}

Figure~\ref{fig:fid-alg-mixed-states-1} depicts Algorithm~\ref{alg:fid-states}.

\begin{theorem}
\label{thm:acc-prob-fid-states}The acceptance probability of
Algorithm~\ref{alg:fid-states} is equal to%
\begin{equation}
\frac{1}{2}\left(  1+\sqrt{F}(\rho_{S}^{0},\rho_{S}^{1})\right)  .
\end{equation}

\end{theorem}

\begin{proof}
The proof can be found in Appendix~\ref{app:ProofAlg4}.
\end{proof}

\subsubsection{Generalized swap test}

\label{sec:mixed-state-swap-test-alg}We now detail another quantum algorithm
for estimating the fidelity of arbitrary states, which is a generalization of the
well known swap test from \cite{BBD+97,BCWW01}. We note that this algorithm was used in \cite[Figure~3]{KW00} as part of their proof that QIP = QIP(3). A key difference between Algorithm~\ref{alg:mixed-state-swap-test} and \cite[Figure~3]{KW00} is that Algorithm~\ref{alg:mixed-state-swap-test} accepts if and only if both qubits at the end are measured to be in the all-zeros state, whereas it is written in \cite[Figure~3]{KW00} that their algorithm accepts if and only if the first qubit is measured to be in the zero state.

\begin{algorithm}
\label{alg:mixed-state-swap-test} The algorithm proceeds as follows:

\begin{enumerate}
\item The verifier prepares a Bell state%
\begin{equation}
|\Phi\rangle_{T^{\prime}T}\coloneqq\frac{1}{\sqrt{2}}(|00\rangle_{T^{\prime}%
T}+|11\rangle_{T^{\prime}T})
\end{equation}
on registers $T^{\prime}$ and $T$ and prepares systems $R_{1}S_{1}R_{2}S_{2}$
in the all-zeros state $|0\rangle_{R_{1}S_{1}R_{2}S_{2}}$.

\item Using the circuits $U_{RS}^{0}$ and $U_{RS}^{1}$, the verifier acts on
$R_{1}S_{1}R_{2}S_{2}$ to prepare the two pure states $|\psi^{\rho^{0}}%
\rangle_{R_{1}S_{1}}$ and $|\psi^{\rho^{1}}\rangle_{R_{2}S_{2}}$.

\item The verifier performs a controlled SWAP\ from qubit$~T$ to
systems$~S_{1}$ and$~S_{2}$, which applies the identity if the control qubit is
$|0\rangle$ and swaps $S_{1}$ with $S_{2}$ if the control qubit is $|1\rangle$.

\item The verifier transmits systems $T^{\prime}$, $R_{1}$, and $R_{2}$ to the prover.

\item The prover prepares a system$~F$ in the $|0\rangle_{F}$ state and acts
on systems $T^{\prime}$, $R_{1}$, $R_{2}$, and $F$ with a unitary
$P_{T^{\prime}R_{1}R_{2}F\rightarrow T^{\prime\prime}F^{\prime}}$ to produce
the output systems$~T^{\prime\prime}$ and$~F^{\prime}$, where~$T^{\prime
\prime}$ is a qubit system.

\item The prover sends system~$T^{\prime\prime}$ to the verifier, who then
performs a Bell measurement%
\begin{equation}
\{\Phi_{T^{\prime\prime}T},I_{T^{\prime\prime}T}-\Phi_{T^{\prime\prime}T}\}
\end{equation}
on systems$~T^{\prime\prime}$ and~$T$. The verifier accepts if and only if the
outcome $\Phi_{T^{\prime\prime}T}$ occurs.
\end{enumerate}
\end{algorithm}

Figure~\ref{fig:fid-alg-mixed-states-swap-test} depicts
Algorithm~\ref{alg:mixed-state-swap-test}.

\begin{figure}[ptb]
\begin{center}
\includegraphics[
width=\linewidth
]{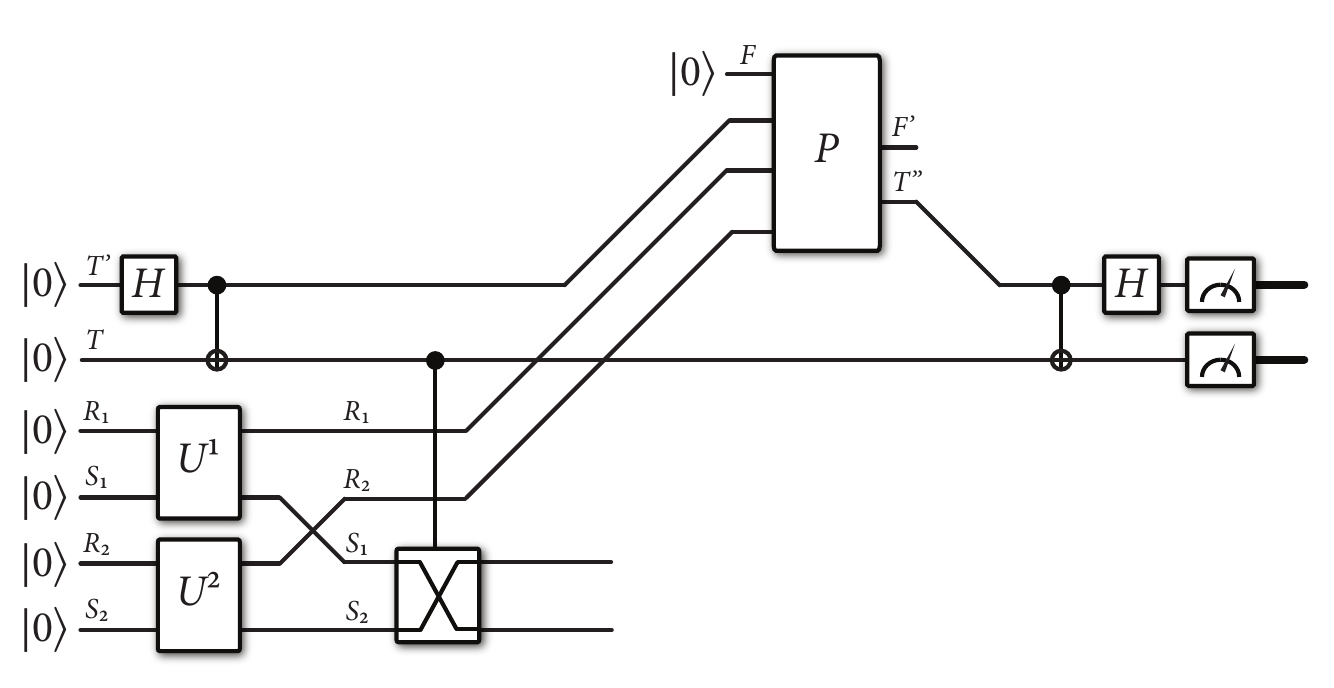}
\end{center}
\caption{This figure depicts Algorithm~\ref{alg:mixed-state-swap-test} for
estimating the fidelity of mixed states generated by quantum circuits
$U^{0}_{RS}$ and $U^{1}_{RS}$. Algorithm~\ref{alg:mixed-state-swap-test}
represents a generalization of the well known swap test for estimating the
fidelity of pure states.}%
\label{fig:fid-alg-mixed-states-swap-test}%
\end{figure}

\begin{theorem}
\label{thm:acc-prob-mixed-state-swap-test}
The acceptance probability of Algorithm~\ref{alg:mixed-state-swap-test}\ is
equal to%
\begin{equation}
\frac{1}{2}\left(  1+F(\rho_{S}^{0},\rho_{S}^{1})\right)  .
\label{eq:accept-prob-mixed-swap-test}
\end{equation}

\end{theorem}

\begin{proof}
The proof can be found in Appendix~\ref{app:ProofAlg5}.
\end{proof}

\subsubsection{Variational algorithm with Bell measurements}

\label{sec:mixed-state-Bell-tests-alg}A third method for estimating the
fidelity of arbitrary multi-qubit states is a variational algorithm that is
based on a generalization of the approach outlined in \cite{GC13,Suba__2019}.
The approach from \cite{GC13,Suba__2019} employs Bell measurements to estimate
the expectation of the SWAP\ observable, which in turn allows for estimating
the fidelity of multi-qubit pure states. See also \cite{Brun04}.

We begin in this section by recalling the basic idea from
\cite{GC13,Suba__2019} for estimating fidelity of pure states. Let $\psi_{S}$
and $\varphi_{S}$ be $m$-qubit pure states of a system$~S$ (so that
$S=S_{1}\cdots S_{m}$, where each $S_{i}$ is a qubit system, for $i\in\left\{
1,\ldots,m\right\}  $). Let $F_{S\tilde{S}}$ denote the unitary swap operator
that swaps systems $S$ and $\tilde{S}$, and recall that%
\begin{equation}
\operatorname{Tr}[F_{S\tilde{S}}(\psi_{S}\otimes\varphi_{\tilde{S}%
})]=\left\vert \langle\psi|\varphi\rangle\right\vert ^{2}=F(\psi_{S}%
,\varphi_{S}).
\end{equation}
Consider that%
\begin{equation}
F_{S\tilde{S}}=F_{S_{1}\tilde{S}_{1}}\otimes F_{S_{2}\tilde{S}_{2}}%
\otimes\cdots\otimes F_{S_{m}\tilde{S}_{m}}.
\end{equation}
Now observe that%
\begin{equation}
F_{S_{i}\tilde{S}_{i}}=\sum_{x,z\in\left\{  0,1\right\}  }\left(  -1\right)
^{x\cdot z}\Phi_{S_{i}\tilde{S}_{i}}^{x,z},
\end{equation}
where the Bell states are defined as%
\begin{align}
|\Phi^{0,0}\rangle &  \coloneqq\frac{1}{\sqrt{2}}\left(  |00\rangle
+|11\rangle\right)  ,\\
|\Phi^{0,1}\rangle &  \coloneqq\frac{1}{\sqrt{2}}\left(  |00\rangle
-|11\rangle\right)  ,\\
|\Phi^{1,0}\rangle &  \coloneqq\frac{1}{\sqrt{2}}\left(  |01\rangle
+|10\rangle\right)  ,\\
|\Phi^{1,1}\rangle &  \coloneqq\frac{1}{\sqrt{2}}\left(  |01\rangle
-|10\rangle\right)  .
\end{align}
We then conclude that%
\begin{align}
&  F(\psi_{S},\varphi_{S})\nonumber\\
&  =\operatorname{Tr}\!\left[  \left(  \bigotimes\limits_{i=1}^{m}%
F_{S_{i}\tilde{S}_{i}}\right)  \left(  \psi_{S}\otimes\varphi_{\tilde{S}%
}\right)  \right] \\
&  =\operatorname{Tr}\!\left[  \left(  \bigotimes\limits_{i=1}^{m}\sum
_{x_{i},z_{i}\in\left\{  0,1\right\}  }\left(  -1\right)  ^{x_{i}\cdot z_{i}%
}\Phi_{S_{i}\tilde{S}_{i}}^{x_{i},z_{i}}\right)  \left(  \psi_{S}%
\otimes\varphi_{\tilde{S}}\right)  \right] \\
&  =\sum_{\substack{x_{1},z_{1},\ldots,\\x_{m},z_{m}\in\left\{  0,1\right\}
}}\left(  -1\right)  ^{\overrightarrow{x}\cdot\overrightarrow{z}%
}\operatorname{Tr}\!\left[  \left(  \bigotimes\limits_{i=1}^{m}\Phi
_{S_{i}\tilde{S}_{i}}^{x_{i},z_{i}}\right)  \left(  \psi_{S}\otimes
\varphi_{\tilde{S}}\right)  \right]  ,
\end{align}
where%
\begin{equation}
\overrightarrow{x}\cdot\overrightarrow{z}\equiv\sum_{i=1}^{m}x_{i}\cdot z_{i}.
\end{equation}
Thus, the approach of \cite{GC13,Suba__2019} is to estimate $F(\psi
_{S},\varphi_{S})$ by repeatedly performing Bell measurements on corresponding
qubits of $\psi_{S}$ and $\varphi_{\tilde{S}}$ followed by classical
postprocessing of the outcomes. In particular, for $j\in\left\{
1,\ldots,n\right\}  $, set $Y_{j}=\left(  -1\right)  ^{\sum_{i=1}^{m}%
x_{i}\cdot z_{i}}$, where $x_{1},z_{1},\ldots,x_{m},z_{m}\in\left\{
0,1\right\}  $ are the outcomes of the Bell measurements on the $j$th
iteration. Then set $\overline{Y^{n}}\coloneqq\frac{1}{n}\sum_{j=1}^{n}Y_{j}$.
By the Hoeffding inequality \cite{H63}, for accuracy $\varepsilon\in(0,1)$ and
failure probability $\delta\in(0,1)$, we are guaranteed that
\begin{equation}
\Pr[\left\vert \overline{Y^{n}}-F(\psi_{S},\varphi_{S})\right\vert
\leq\varepsilon]\geq1-\delta,
\end{equation}
as long as $n\geq\frac{2}{\varepsilon^{2}}\ln\!\left(  \frac{2}{\delta
}\right)  $. Thus, the algorithm is polynomial in the inverse accuracy and
logarithmic in the inverse failure probability.

We now form a simple generalization of this algorithm to estimate the fidelity
of arbitrary states $\rho_{S}^{0}$ and $\rho_{S}^{1}$, in which we perform a
variational optimization over unitaries that act on the reference system of
one of the states. For $i\in\left\{  0,1\right\}  $, let $U_{RS}^{i}$ be an
$m$-qubit unitary that acts on $|0\rangle_{RS}$ to generate the $m$-qubit
state $|\psi^{\rho^{i}}\rangle_{RS}$; i.e.,%
\begin{equation}
|\psi^{\rho^{i}}\rangle_{RS}=U_{RS}^{i}|0\rangle_{RS},
\end{equation}
such that%
\begin{equation}
\rho_{S}^{i}=\operatorname{Tr}_{R}[|\psi^{\rho^{i}}\rangle\!\langle\psi
^{\rho^{i}}|_{RS}].
\end{equation}

\begin{algorithm}
\label{alg:mixed-state-Bell-tests} Set the error tolerance $\varepsilon>0$.
Set $\eta,\delta\in(0,1)$. The algorithm proceeds as follows:

\begin{enumerate}
\item Prepare systems $R_{1}S_{1}R_{2}S_{2}$ in the all-zeros state
$|0\rangle_{R_{1}S_{1}R_{2}S_{2}}$.

\item Act with the circuits $U_{RS}^{0}$ and $U_{RS}^{1}$ on systems
$R_{1}S_{1}R_{2}S_{2}$ to prepare the two pure states $|\psi^{\rho^{0}}%
\rangle_{R_{1}S_{1}}$ and $|\psi^{\rho^{1}}\rangle_{R_{2}S_{2}}$.

\item Perform a unitary $V_{R_{1}}(\mathbf{\theta})$ on system $R_{1}$.

\item For $j\in\left\{  1,\ldots,n\right\}  $, where $n\geq\frac{2}{\eta^{2}%
}\ln\!\left(  \frac{2}{\delta}\right)  $, for $i\in\left\{  1,\ldots
,m\right\}  $, perform a Bell measurement on qubit $i$ of system $R_{1}$ and
qubit $i$ of system $R_{2}$, with outcomes $x_{R}^{i}$ and $z_{R}^{i}$, and
perform a Bell measurement on qubit $i$ of system $S_{1}$ and qubit~$i$ of
system $S_{2}$, with outcomes $x_{S}^{i}$ and $z_{S}^{i}$. Set $Y_{j}%
(\mathbf{\theta})=\left(  -1\right)  ^{\sum_{i=1}^{m}x_{R}^{i}\cdot z_{R}%
^{i}+x_{S}^{i}\cdot z_{S}^{i}}$.

\item Set
\begin{equation}
\overline{Y^{n}}(\mathbf{\theta})\coloneqq\frac{1}{n}\sum_{j=1}^{n}%
Y_{j}(\mathbf{\theta}),
\end{equation}
as an estimate of
\begin{equation}
F_{\mathbf{\theta}}\equiv\left\vert \langle\psi^{\rho^{1}}|_{RS}%
V_{R}(\mathbf{\theta})\otimes I_{S}|\psi^{\rho^{0}}\rangle_{RS}\right\vert
^{2},
\end{equation}
so that%
\begin{equation}
\Pr\!\left[  \left\vert \overline{Y^{n}}(\mathbf{\theta})-F_{\mathbf{\theta}%
}\right\vert \leq\eta\right]  \geq1-\delta.
\end{equation}

\item Perform a maximization of the reward function $\overline{Y^{n}%
}(\mathbf{\theta})$ and update the parameters in $\mathbf{\theta}$.

\item Repeat 1-6 until the reward function $\overline{Y^{n}}(\mathbf{\theta})$
converges with tolerance $\varepsilon$, so that $\left\vert \Delta
\overline{Y^{n}}(\mathbf{\theta})\right\vert \leq\varepsilon$, or until some
maximum number of iterations is reached. (Here $\Delta\overline{Y^{n}%
}(\mathbf{\theta})$ represents the difference in $\overline{Y^{n}%
}(\mathbf{\theta})$ from the previous and current iteration.)

\item Output the final $\overline{Y^{n}}(\mathbf{\theta})$ as an estimate of
the fidelity $F(\rho_{S}^{0},\rho_{S}^{1})$.
\end{enumerate}
\end{algorithm}

\begin{figure}
\begin{center}
\includegraphics[
width=\linewidth
]{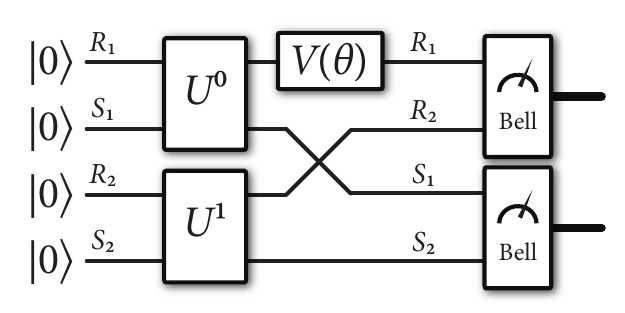}
\end{center}
\caption{This figure depicts Algorithm~\ref{alg:mixed-state-Bell-tests} for
estimating the fidelity of quantum states generated by quantum circuits
$U^{0}_{ RS}$ and $U^{1}_{RS}$.}%
\label{fig:fid-alg-mixed-states-Bell-test}%
\end{figure}

Figure~\ref{fig:fid-alg-mixed-states-Bell-test} depicts
Algorithm~\ref{alg:mixed-state-Bell-tests}. Since this is a variational
algorithm, it is not guaranteed to converge or have a specified runtime, other than running for a maximum number of iterations.
However, it is clearly a generalization of the algorithm from
\cite{GC13,Suba__2019}, in which we estimate the fidelity%
\begin{multline}
\left\vert \langle\psi^{\rho^{1}}|_{RS}V_{R}(\mathbf{\theta})\otimes
I_{S}|\psi^{\rho^{0}}\rangle_{RS}\right\vert ^{2}\\
=F(\psi_{RS}^{\rho^{1}},V_{R}(\mathbf{\theta})\psi_{RS}^{\rho^{0}}%
V_{R}(\mathbf{\theta})^{\dag})
\end{multline}
at each iteration of the algorithm. If we could actually optimize over all
possible unitaries acting on the reference system~$R$, then the algorithm
would indeed estimate the fidelity, as a consequence of Uhlmann's theorem
\cite{U76}:%
\begin{equation}
F(\rho_{S}^{0},\rho_{S}^{1})=\sup_{V_{R}}F(\psi_{RS}^{\rho^{1}},V_{R}\psi
_{RS}^{\rho^{0}}V_{R}^{\dag}).
\end{equation}
However, by optimizing over only a subset of all unitaries,
Algorithm~\ref{alg:mixed-state-Bell-tests}\ estimates a lower bound on the
fidelity $F(\rho_{S}^{0},\rho_{S}^{1})$.

\subsubsection{Variational algorithm for Fuchs--Caves measurement}

Algorithm~\ref{alg:fid-states}\ from
Section~\ref{sec:contr-U-bell-overlap-alg} is based on Uhlmann's formula for
fidelity in \eqref{eq:uhlmann-thm}, and the same is true for
Algorithm~\ref{alg:mixed-state-swap-test}\ from
Section~\ref{sec:mixed-state-swap-test-alg} and
Algorithm~\ref{alg:mixed-state-Bell-tests}\ from
Section~\ref{sec:mixed-state-Bell-tests-alg}. An alternate optimization
formula for the fidelity of states $\rho_{S}^{0}$ and $\rho_{S}^{1}$ is as
follows \cite{FC95}:%
\begin{equation}
F(\rho_{S}^{0},\rho_{S}^{1})=\left[  \min_{\left\{  \Lambda_{S}^{x}\right\}
_{x}}\sum_{x}\sqrt{\operatorname{Tr}[\Lambda_{S}^{x}\rho_{S}^{0}%
]\operatorname{Tr}[\Lambda_{S}^{x}\rho_{S}^{1}]}\right]  ^{2},
\label{eq:fid-min-opt-FC-meas}%
\end{equation}
where the minimization is over every positive operator-valued measure
$\left\{  \Lambda_{S}^{x}\right\}  _{x}$ (i.e., the operators satisfy
$\Lambda_{S}^{x}\geq0$ for all $x$ and $\sum_{x}\Lambda_{S}^{x}=I_{S}$). A
measurement achieving the optimal value of the fidelity is known as the
Fuchs--Caves measurement \cite{FC95} and has the form $\{|\varphi_{x}%
\rangle\!\langle\varphi_{x}|\}_{x}$, where $|\varphi_{x}\rangle$ is an
eigenvector, with eigenvalue $\lambda_{x}$, of the following operator
geometric mean of $\rho^{0}$ and $(\rho^{1})^{-1}$ (also called
\textquotedblleft quantum likelihood ratio\textquotedblright\ operator in
\cite{F96}):%
\begin{equation}
M\coloneqq \left(  \rho^{1}\right)  ^{-1/2}\sqrt{\left(  \rho^{1}\right)
^{1/2}\rho^{0}\left(  \rho^{1}\right)  ^{1/2}}\left(  \rho^{1}\right)
^{-1/2},
\end{equation}
so that%
\begin{equation}
M=\sum_{x}\lambda_{x}|\varphi_{x}\rangle\!\langle\varphi_{x}|.
\end{equation}
That is, it is known from \cite{FC95,F96} that%
\begin{equation}
F(\rho_{S}^{0},\rho_{S}^{1})=\left[  \sum_{x}\sqrt{\operatorname{Tr}%
[|\varphi_{x}\rangle\!\langle\varphi_{x}|\rho_{S}^{0}]\operatorname{Tr}%
[|\varphi_{x}\rangle\!\langle\varphi_{x}|\rho_{S}^{1}]}\right]  ^{2}.
\end{equation}

Thus, we can build a variational algorithm around this formulation of
fidelity, with the idea being to optimize over parameterized measurements in
an attempt to optimize the fidelity, while at the same time learn the
Fuchs--Caves measurement (or a different fidelity-achieving measurement). In contrast to the other variational algorithms
presented in previous sections, this alternate approach leads to an upper
bound on the fidelity.

Before detailing the algorithm, recall the Naimark extension theorem
\cite{N40} (see also \cite{Wbook17,Wat18,KW20book}), which states that a
general POVM\ $\{\Lambda_{S}^{x}\}_{x}$\ with $m$ outcomes, acting on a
quantum state $\rho$ of a $d$-dimensional system $S$, can be realized as a
unitary interaction $U_{SP}$ of the system $S$ with an $m$-dimensional probe
system $P$, followed by a projective measurement $\{|x\rangle\!\langle
x|_{P}\}_{x}$ acting on the probe system. That is,%
\begin{equation}
\operatorname{Tr}[\Lambda_{S}^{x}\rho_{S}]=\operatorname{Tr}[(I_{S}%
\otimes|x\rangle\!\langle x|_{P})U_{SP}(\rho_{S}\otimes|0\rangle
\!\langle0|_{P})U_{SP}^{\dag}].
\end{equation}
It suffices to choose $U_{SP}$ so that%
\begin{equation}
U_{SP}|\psi\rangle_{S}|0\rangle_{P}=\sum_{x}\sqrt{\Lambda_{S}^{x}}|\psi
\rangle_{S}|x\rangle_{P}.
\end{equation}
Thus, we can express the optimization problem in
\eqref{eq:fid-min-opt-FC-meas} as follows:%
\begin{multline}
\sqrt{F}(\rho_{S}^{0},\rho_{S}^{1})=\label{eq:fid-min-all-unitaries}\\
  \min_{U_{SP}}\sum_{x}\sqrt{%
\begin{array}
[c]{c}%
\operatorname{Tr}[(I_{S}\otimes|x\rangle\!\langle x|_{P})U_{SP}(\rho_{S}%
^{0}\otimes|0\rangle\!\langle0|_{P})U_{SP}^{\dag}]\times\\
\operatorname{Tr}[(I_{S}\otimes|x\rangle\!\langle x|_{P})U_{SP}(\rho_{S}%
^{1}\otimes|0\rangle\!\langle0|_{P})U_{SP}^{\dag}]
\end{array}
}.
\end{multline}
By replacing the optimization in \eqref{eq:fid-min-all-unitaries}\ over all
unitaries with an optimization over parameterized ones, we arrive at a
variational algorithm for estimating fidelity:

\begin{algorithm}
\label{alg:mixed-state-FC-meas-min} Set $n \in\mathbb{N}$ and the error
tolerance $\varepsilon>0$. The algorithm proceeds as follows:

\begin{enumerate}
\item For $j\in\left\{  1,\ldots,n\right\}  $, prepare system $S_{1}$ in the
state $\rho_{S_{1}}^{0}$ and system $S_{2}$ in the state $\rho_{S_{2}}^{1}$,
and prepare systems $P_{1}$ and $P_{2}$ in the all-zeros state $|0\rangle
_{P_{1}}\otimes|0\rangle_{P_{2}}$.

\item Act with the circuit $U_{S_{1}P_{1}}(\mathbf{\theta})$ on systems
$S_{1}P_{1}$ and act with the same circuit $U_{S_{2}P_{2}}(\mathbf{\theta})$
on systems $S_{2}P_{2}$.

\item Measure system $P_{1}$ in the computational basis and record the outcome
as $y_{j}$, and measure system $P_{2}$ in the computational basis and record
the outcome as $z_{j}$.

\item Using the measurement data $\left\{  y_{j}\right\}  _{j=1}^{n}$ and
$\left\{  z_{j}\right\}  _{j=1}^{n}$, calculate the empirical distributions
$\tilde{p}_{\mathbf{\theta}}(x)$ and $\tilde{q}_{\mathbf{\theta}}(x)$, where
$\tilde{p}_{\mathbf{\theta}}(x)$ is the empirical distribution resulting from%
\begin{multline}
p_{\mathbf{\theta}}(x)\coloneqq\\
\operatorname{Tr}[(I_{S}\otimes|x\rangle\!\langle x|_{P})U_{SP}(\mathbf{\theta
})(\rho_{S}^{0}\otimes|0\rangle\!\langle0|_{P})U_{SP}^{\dag}(\mathbf{\theta
})],
\end{multline}
and $\tilde{q}_{\mathbf{\theta}}(x)$ is the empirical distribution resulting
from%
\begin{multline}
q_{\mathbf{\theta}}(x)\coloneqq\\
\operatorname{Tr}[(I_{S}\otimes|x\rangle\!\langle x|_{P})U_{SP}(\mathbf{\theta
})(\rho_{S}^{1}\otimes|0\rangle\!\langle0|_{P})U_{SP}^{\dag}(\mathbf{\theta
})].
\end{multline}

\item Output
\begin{equation}
F(\tilde{p}_{\mathbf{\theta}},\tilde{q}_{\mathbf{\theta}})\coloneqq \left[
\sum_{x}\sqrt{\tilde{p}_{\mathbf{\theta}}(x)\tilde{q}_{\mathbf{\theta}}%
(x)}\right]  ^{2} \label{eq:fid-estimator-FC-alg}%
\end{equation}
as an estimate of $F(p_{\mathbf{\theta}},q_{\mathbf{\theta}})$.

\item Perform a minimization of the cost function $F(\tilde{p}_{\mathbf{\theta}%
},\tilde{q}_{\mathbf{\theta}})$ and update the parameters in $\mathbf{\theta}$.

\item Repeat 1-6 until the cost function $F(\tilde{p}_{\mathbf{\theta}}%
,\tilde{q}_{\mathbf{\theta}})$ converges with tolerance $\varepsilon$, so that
$\left\vert \Delta F(\tilde{p}_{\mathbf{\theta}},\tilde{q}_{\mathbf{\theta}%
})\right\vert \leq\varepsilon$, or until some maximum number of iterations is
reached. (Here $\Delta F(\tilde{p}_{\mathbf{\theta}},\tilde{q}_{\mathbf{\theta
}})$ represents the difference in $F(\tilde{p}_{\mathbf{\theta}},\tilde
{q}_{\mathbf{\theta}})$ from the previous and current iteration.)

\item Output the final value of $F(\tilde{p}_{\mathbf{\theta}},\tilde
{q}_{\mathbf{\theta}})$ as an estimate of the fidelity $F(\rho_{S}^{0}%
,\rho_{S}^{1})$.
\end{enumerate}
\end{algorithm}

\begin{figure}
\begin{center}
\includegraphics[
width=\linewidth
]{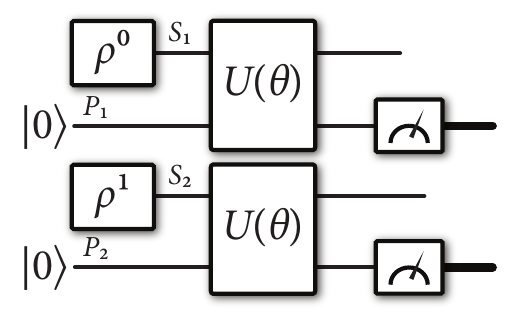}
\end{center}
\caption{This figure depicts Algorithm~\ref{alg:mixed-state-FC-meas-min} for
estimating the fidelity of quantum states $\rho^{0} _{ S}$ and $\rho^{1}_{S}$.
The boxes enclosing $\rho^{0}$ and $\rho^{1}$ indicate that these are some
mechanisms by which these states are prepared.}%
\label{fig:fid-alg-mixed-states-FC-meas}%
\end{figure}

Figure~\ref{fig:fid-alg-mixed-states-FC-meas} depicts
Algorithm~\ref{alg:mixed-state-FC-meas-min}. As before, since this is a
variational algorithm, it is not guaranteed to converge or have a specified
runtime, other than running for a maximum number of iterations. One advantage of this algorithm is that it does not require
purifications of the states $\rho_{S}^{0}$ and $\rho_{S}^{1}$. All it requires
is a circuit or method to prepare these states, and then it performs
measurements on these states, in an attempt to learn an optimal measurement
with respect to the cost function $F(\tilde{p}_{\mathbf{\theta}},\tilde
{q}_{\mathbf{\theta}})$.

In Algorithm~\ref{alg:mixed-state-FC-meas-min}, we did not specify how large
$n$ should be in order to get a desired accuracy of the estimator in
\eqref{eq:fid-estimator-FC-alg} for the classical fidelity
$F(p_{\mathbf{\theta}},q_{\mathbf{\theta}})$. This estimator is called a
``plug-in estimator'' in the literature on this topic, and it is a biased
estimator, which however converges to $F(p_{\mathbf{\theta}},q_{\mathbf{\theta
}})$ in the asymptotic limit $n \to\infty$. As a consequence of the estimator
in \eqref{eq:fid-estimator-FC-alg} being biased, the Hoeffding inequality does
not readily apply in this case. As far as we can tell, it is an open question
to determine the rate of convergence of this estimator to $F(p_{\mathbf{\theta
}},q_{\mathbf{\theta}})$. Related work on this topic has been considered in
\cite{JVHW15,AOST17}.

\subsection{Estimating fidelity of channels}

\label{sec:fid-channels}In this section, we outline a method for estimating
the fidelity of channels on a quantum computer, by means of an interaction
with competing quantum provers \cite{GW05,G05,GW06,G09,GW13}. The goal of one
prover is to maximize the acceptance probability, while the goal of the other
prover is to minimize the acceptance probability. We refer to the first prover
as the max-prover and the second as the min-prover. The specific setting that
we deal with is called a double quantum interactive proof (DQIP) \cite{GW13},
due to the fact that the min-prover goes first and then the max-prover goes
last. The class of promise problems that can be solved in this model is
equivalent to PSPACE \cite{GW13}, which is the class of problems that can be
decided on a classical computer with polynomial memory.

Let us recall that the fidelity of channels $\mathcal{N}_{A\rightarrow B}^{0}$
and $\mathcal{N}_{A\rightarrow B}^{1}$ is defined as follows \cite{GLN04}:%
\begin{equation}
F(\mathcal{N}_{A\rightarrow B}^{0},\mathcal{N}_{A\rightarrow B}^{1}%
)\coloneqq\inf_{\rho_{RA}}F(\mathcal{N}_{A\rightarrow B}^{0}(\rho
_{RA}),\mathcal{N}_{A\rightarrow B}^{1}(\rho_{RA})),
\label{eq:def-fid-channels}%
\end{equation}
where the infimum is over every state $\rho_{RA}$, with the reference system
$R$ arbitrarily large. It is known that the infimum is achieved by a pure
state $\psi_{RA}$ with the reference system $R$ isomorphic to the channel
input system~$A$, so that%
\begin{equation}
F(\mathcal{N}_{A\rightarrow B}^{0},\mathcal{N}_{A\rightarrow B}^{1}%
)\coloneqq\min_{\psi_{RA}}F(\mathcal{N}_{A\rightarrow B}^{0}(\psi
_{RA}),\mathcal{N}_{A\rightarrow B}^{1}(\psi_{RA})).
\label{eq:def-fid-channels-simplify}%
\end{equation}
It is also known that it is possible to calculate the fidelity of channels by
means of a semi-definite program \cite{Yuan2017,KW20}, which provides a way to
verify the output of our proposed algorithm for sufficiently small examples.

Suppose that the goal is to estimate the fidelity of channels $\mathcal{N}%
_{A\rightarrow B}^{0}$ and $\mathcal{N}_{A\rightarrow B}^{1}$, and we are
given access to quantum circuits $U_{AE^{\prime}\rightarrow BE}^{0}$ and
$U_{AE^{\prime}\rightarrow BE}^{1}$ that realize isometric extensions of the
channels $\mathcal{N}_{A\rightarrow B}^{0}$ and $\mathcal{N}_{A\rightarrow
B}^{1}$, respectively, in the sense that%
\begin{multline}
\mathcal{N}_{A\rightarrow B}^{i}(\omega_{A})=\\
\operatorname{Tr}_{E}[U_{AE^{\prime}\rightarrow BE}^{i}(\omega_{A}%
\otimes|0\rangle\!\langle0|_{E^{\prime}})(U_{AE^{\prime}\rightarrow BE}%
^{i})^{\dag}],
\end{multline}
for $i\in\left\{  0,1\right\}  $.

We now provide a DQIP algorithm for estimating the following quantity:%
\begin{equation}
\frac{1}{2}\left(  1+\sqrt{F}(\mathcal{N}_{A\rightarrow B}^{0},\mathcal{N}%
_{A\rightarrow B}^{1})\right)  ,
\end{equation}
which is based in part on Algorithm~\ref{alg:fid-states}\ but instead features
an optimization over input states of the min-prover.

\begin{algorithm}
\label{alg:fid-channels} The algorithm proceeds as follows:

\begin{enumerate}
\item The verifier prepares a Bell state
\begin{equation}
|\Phi\rangle_{T^{\prime}T}\coloneqq\frac{1}{\sqrt{2}}(|00\rangle_{T^{\prime}%
T}+|11\rangle_{T^{\prime}T})
\end{equation}
on registers $T^{\prime}$ and $T$ and prepares system $E^{\prime}$ in the
all-zeros state $|0\rangle_{E^{\prime}}$.

\item The min-prover transmits the system $A$ of the state $|\psi\rangle_{RA}$
to the verifier.

\item Using the circuits $U_{AE^{\prime}\rightarrow BE}^{0}$ and
$U_{AE^{\prime}\rightarrow BE}^{1}$, the verifier performs the following
controlled unitary:%
\begin{equation}
\label{eq:controlled-U-fid-channels}
\sum_{i\in\left\{  0,1\right\}  }|i\rangle\!\langle i|_{T}\otimes
U_{AE^{\prime}\rightarrow BE}^{i}.
\end{equation}

\item The verifier transmits systems $T^{\prime}$ and $E$ to the max-prover.

\item The max-prover prepares a system $F$ in the $|0\rangle_{F}$ state and
acts on systems $T^{\prime}$, $E$, and $F$ with a unitary $P_{T^{\prime
}EF\rightarrow T^{\prime\prime}F^{\prime}}$ to produce the output systems
$T^{\prime\prime}$ and $F^{\prime}$, where $T^{\prime\prime}$ is a qubit system.

\item The max-prover sends system $T^{\prime\prime}$ to the verifier, who then
performs a Bell measurement%
\begin{equation}
\{\Phi_{T^{\prime\prime}T},I_{T^{\prime\prime}T}-\Phi_{T^{\prime\prime}T}\}
\end{equation}
on systems $T^{\prime\prime}$ and $T$. The verifier accepts if and only if the
outcome $\Phi_{T^{\prime\prime}T}$ occurs.
\end{enumerate}
\end{algorithm}

Figure~\ref{fig:fid-alg-channels-min-max-1} depicts
Algorithm~\ref{alg:fid-channels}.

\begin{figure*}
\begin{center}
\includegraphics[
width=\linewidth
]{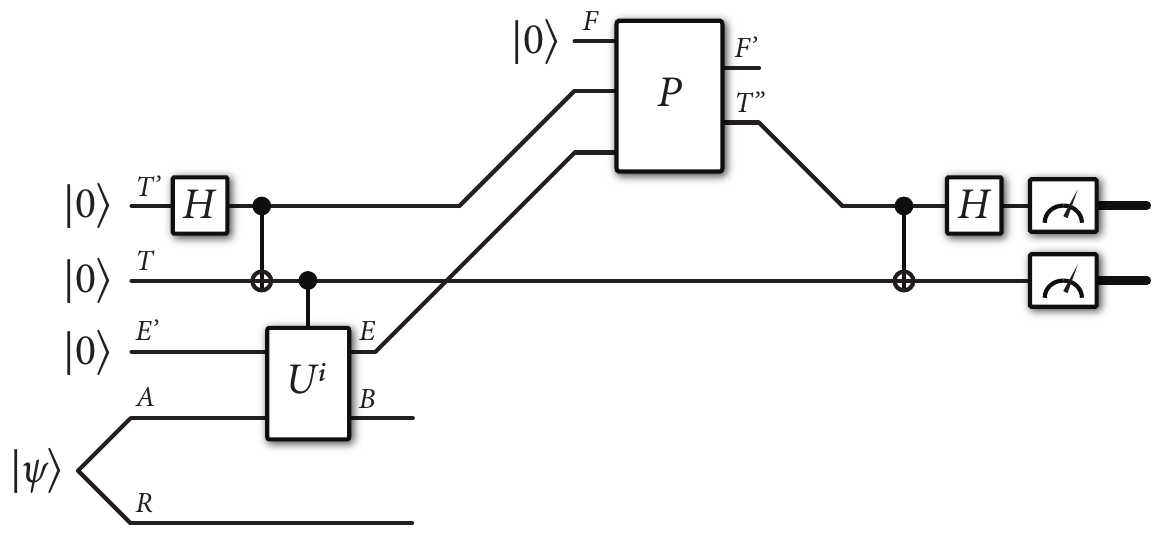}
\end{center}
\caption{This figure depicts Algorithm~\ref{alg:fid-channels} for estimating
the fidelity of quantum channels generated by quantum circuits $U_{AE^{\prime
}\rightarrow BE}^{0}$ and $U_{AE^{\prime}\rightarrow BE}^{1}$. The min-prover
prepares the state $|\psi\rangle_{RA}$ and the max-prover acts with the
unitary $P_{T^{\prime}EF\rightarrow T^{\prime\prime}F^{\prime}}$.}%
\label{fig:fid-alg-channels-min-max-1}%
\end{figure*}

\begin{theorem}
\label{thm:acc-prob-fid-channels}
The acceptance probability of Algorithm~\ref{alg:fid-channels} is equal to%
\begin{equation}
\frac{1}{2}\left(  1+\sqrt{F}(\mathcal{N}_{A\rightarrow B}^{0},\mathcal{N}%
_{A\rightarrow B}^{1})\right)  . \label{eq:acc-prob-fid-channels}%
\end{equation}

\end{theorem}

\begin{proof}
The proof can be found in Appendix~\ref{app:ProofAlg8}.
\end{proof}

\begin{proposition}
\label{prop:minimax-acc-prob-fid-channels}An alternative expression for the
acceptance probability of Algorithm~\ref{alg:fid-channels} is%
\begin{multline}
\min_{\rho_{RA}}\max_{\mathcal{P}_{T^{\prime}E\rightarrow T^{\prime\prime}}%
}\operatorname{Tr}[\Phi_{T^{\prime\prime}T}\mathcal{P}_{T^{\prime}E\rightarrow
T^{\prime\prime}}(\mathcal{M}_{A\rightarrow T^{\prime}TBE}(\rho_{RA}%
))]\label{eq:min-max-acc-prob-fid-channels}\\
=\max_{\mathcal{P}_{T^{\prime}E\rightarrow T^{\prime\prime}}}\min_{\rho_{RA}%
}\operatorname{Tr}[\Phi_{T^{\prime\prime}T}\mathcal{P}_{T^{\prime}E\rightarrow
T^{\prime\prime}}(\mathcal{M}_{A\rightarrow T^{\prime}TBE}(\rho_{RA}))],
\end{multline}
where $\rho_{RA}$ is a quantum state, $\mathcal{P}_{T^{\prime}E\rightarrow
T^{\prime\prime}}$ is a quantum channel, and $\mathcal{M}_{A\rightarrow
T^{\prime}TBE}$ is a quantum channel defined as%
\begin{multline}
\mathcal{M}_{A\rightarrow T^{\prime}TBE}(\rho_{RA}%
)\coloneqq\label{eq:M-qchannel-def}\\
\frac{1}{2}\sum_{i,j\in\left\{  0,1\right\}  }|ii\rangle\!\langle
jj|_{T^{\prime}T}\otimes U^{i}(\rho_{RA}\otimes|0\rangle\!\langle
0|_{E^{\prime}})(U^{j})^{\dag},
\end{multline}
with $U^{i}\equiv U_{AE^{\prime}\rightarrow BE}^{i}$.
\end{proposition}

\begin{proof}
In Step~2 of Algorithm~\ref{alg:fid-channels}, the min-prover could send a
mixed quantum state $\rho_{RA}$ instead of sending a pure state. The
acceptance probability does not change under this modification due to the
argument around
\eqref{eq:def-fid-channels}--\eqref{eq:def-fid-channels-simplify}.
Furthermore, due to the Stinespring dilation theorem \cite{S55}, the actions of
tensoring in $|0\rangle_{F}$, performing the unitary $P_{T^{\prime
}EF\rightarrow T^{\prime\prime}F^{\prime}}$, and tracing over system
$F^{\prime}$ are equivalent to performing a quantum channel $\mathcal{P}%
_{T^{\prime}E\rightarrow T^{\prime\prime}}$. Under these observations,
consider that the acceptance probability is then equal to%
\begin{equation}
\operatorname{Tr}[\Phi_{T^{\prime\prime}T}\mathcal{P}_{T^{\prime}E\rightarrow
T^{\prime\prime}}(\mathcal{M}_{A\rightarrow T^{\prime}TBE}(\rho_{RA}))],
\label{eq:obj-func-fid-channels}%
\end{equation}
where the quantum channel $\mathcal{M}_{A\rightarrow T^{\prime}TBE}$ is
defined in \eqref{eq:M-qchannel-def}. Performing the optimizations $\min
_{\rho_{RA}}\max_{\mathcal{P}_{T^{\prime}E\rightarrow T^{\prime\prime}}}$ then
leads to the first expression in \eqref{eq:min-max-acc-prob-fid-channels}.
Considering that the set of channels is convex and the set of states is
convex, and the objective function in \eqref{eq:obj-func-fid-channels} is
linear in $\rho_{RA}$ for fixed $\mathcal{P}_{T^{\prime}E\rightarrow
T^{\prime\prime}}$ and linear in $\mathcal{P}_{T^{\prime}E\rightarrow
T^{\prime\prime}}$ for fixed $\rho_{RA}$, the minimax theorem \cite{Sion1958}%
\ applies and we can exchange the optimizations.
\end{proof}

Proposition~\ref{prop:minimax-acc-prob-fid-channels} indicates that if the
provers involved can optimize over all possible states and channels, then
indeed the order of optimization can be exchanged. However, in a variational
algorithm, the optimization is generally dependent upon the order in which it
is conducted because we are not optimizing over all possible states and
channels, but instead optimizing over parameterized circuits. In this latter
case, the state space is no longer convex and the objective function no longer
linear in these parameters. However, we can still attempt the following
``see-saw'' strategy in a variational algorithm: first minimize the objective
function with respect to the input state $\psi_{RA}$ while keeping the unitary
$P_{T^{\prime}EF\rightarrow T^{\prime\prime}F^{\prime}}$ fixed. Then maximize
the objective function with respect to the unitary $P_{T^{\prime}EF\rightarrow
T^{\prime\prime}F^{\prime}}$ while keeping the state $\psi_{RA}$ fixed. Then
repeat this process some number of times. We consider this approach in Section~\ref{sec:numerics-fid-chs}. 

\subsection{Estimating fidelity of strategies}

\label{sec:strategy-fidelity}In this section, we extend
Algorithm~\ref{alg:fid-channels}\ beyond estimating the fidelity of channels
to estimating the fidelity of general strategies
\cite{Gutoski2018fidelityofquantum}, by conducting several rounds of
interaction with the min-prover followed by a single interaction with the
max-prover at the end.

We now develop this idea in detail. Let us first recall the definition of a
quantum strategy from
\cite{GW06,CDP08a,CDP09,G09,G12,Gutoski2018fidelityofquantum}. An $n$-turn
quantum strategy $\mathcal{N}^{(n)}$, with $n\geq1$, input systems $A_{1}$,
\ldots, $A_{n}$, and output systems $B_{1}$, \ldots, $B_{n}$ consists of the following:

\begin{enumerate}
\item memory systems $M_{1}$, \ldots, $M_{n-1}$, and

\item quantum channels $\mathcal{N}_{A_{1}\rightarrow M_{1}B_{1}}^{1}$,
$\mathcal{N}_{M_{1}A_{2}\rightarrow M_{2}B_{2}}^{2}$, \ldots, $\mathcal{N}%
_{M_{n-2}A_{n-1}\rightarrow M_{n-1}B_{n-1}}^{n-1}$, and $\mathcal{N}%
_{M_{n-1}A_{n}\rightarrow B_{n}}^{n}$.
\end{enumerate}

\noindent It is implicit that any of the systems involved can be trivial
systems, which means that state preparation and measurements are included as
special cases.

A co-strategy interacts with a strategy; co-strategies are in fact strategies
also, but it is useful conceptually to provide an explicit means by which an
agent can interact with a strategy. An $(n-1)$-turn co-strategy
$\mathcal{S}^{(n-1)}$, with input systems $B_{1}$, \ldots, $B_{n}$ and output
systems $A_{1}$, \ldots, $A_{n}$ consists of the following:

\begin{enumerate}
\item memory systems $R_{1}$, \ldots, $R_{n}$,

\item a quantum state $\rho_{R_{1}A_{1}}$, and

\item quantum channels $\mathcal{S}_{R_{1}B_{1}\rightarrow R_{2}A_{2}}^{1}$,
$\mathcal{S}_{R_{2}B_{2}\rightarrow R_{3}A_{3}}^{2}$, \ldots, and
$\mathcal{S}_{R_{n-1}B_{n-1}\rightarrow R_{n}A_{n}}^{n-1}$.
\end{enumerate}

\noindent The result of the interaction of the strategy $\mathcal{N}^{(n)}$
with the co-strategy $\mathcal{S}^{(n-1)}$ is a quantum state on systems
$R_{n}B_{n}$, and we employ the shorthand%
\begin{equation}
\mathcal{N}^{(n)}\circ\mathcal{S}^{(n-1)}
\label{eq:state-after-strategy-co-strategy}%
\end{equation}
to denote this quantum state. Figure~\ref{fig:strategy-co-strategy} depicts a
three-turn strategy interacting with a two-turn co-strategy.

\begin{figure}
\begin{center}
\includegraphics[
width=\linewidth
]{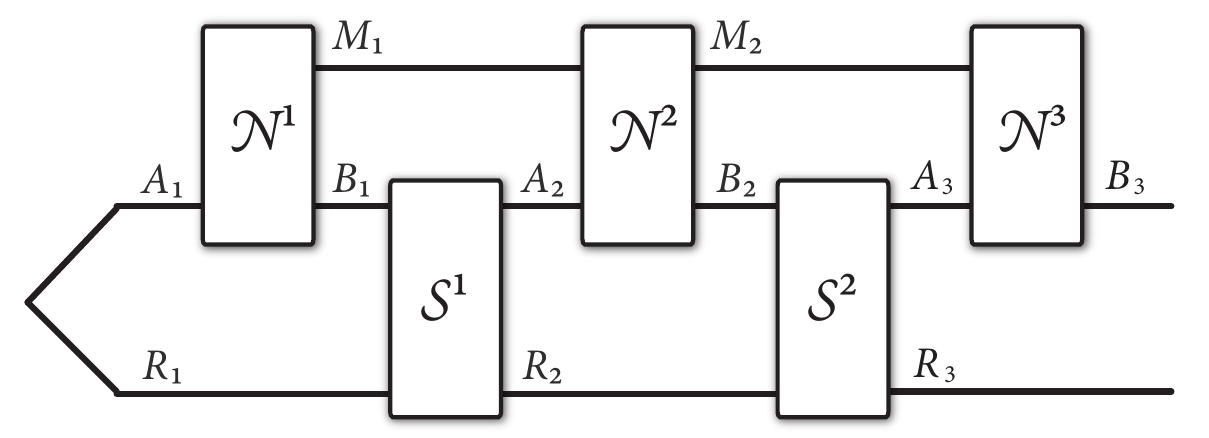}
\end{center}
\caption{Interaction of a three-turn strategy $\mathcal{N}^{(3)}$ with a
two-turn co-strategy $\mathcal{S}^{(2)}$.}%
\label{fig:strategy-co-strategy}%
\end{figure}

Let $\mathcal{N}^{0,(n)}$ and $\mathcal{N}^{1,(n)}$ denote two compatible, $n$-turn
quantum strategies, meaning that all systems involved in these strategies are
the same but the channels that make up the strategies are possibly different.
The fidelity of the strategies $\mathcal{N}^{0,(n)}$ and $\mathcal{N}^{1,(n)}$
is defined as \cite{Gutoski2018fidelityofquantum}%
\begin{multline}
F(\mathcal{N}^{0,(n)},\mathcal{N}^{1,(n)})\coloneqq\\
\inf_{\mathcal{S}^{(n-1)}}F(\mathcal{N}^{0,(n)}\circ\mathcal{S}^{(n-1)}%
,\mathcal{N}^{1,(n)}\circ\mathcal{S}^{(n-1)}), \label{eq:def-fid-strategies}%
\end{multline}
where the optimization is over every co-strategy $\mathcal{S}^{(n-1)}$. One
can interpret the strategy fidelity in \eqref{eq:def-fid-strategies} as a
generalization of the fidelity of channels in \eqref{eq:def-fid-channels}, in
which the idea is to optimize the fidelity measure over all possible
co-strategies that can be used to distinguish the strategies $\mathcal{N}%
^{0,(n)}$ and $\mathcal{N}^{1,(n)}$. It follows from a standard
data-processing argument that it suffices to perform the optimization in
\eqref{eq:def-fid-strategies} over co-strategies involving an initial pure
state $\rho_{R_{1}A_{1}}$ and channels $\mathcal{S}_{R_{1}B_{1}\rightarrow
R_{2}A_{2}}^{1}$, $\mathcal{S}_{R_{2}B_{2}\rightarrow R_{3}A_{3}}^{2}$,
\ldots, and $\mathcal{S}_{R_{n-1}B_{n-1}\rightarrow R_{n}A_{n}}^{n-1}$ that
are each isometric channels (these are called pure co-strategies in
\cite{Gutoski2018fidelityofquantum}). We also note here that the measure in
\eqref{eq:def-fid-strategies} is generalized by the generalized strategy
divergence of~\cite{Wang2019a}.

The goal of this section is to delineate a DQIP algorithm for estimating the
fidelity of strategies $\mathcal{N}^{0,(n)}$ and $\mathcal{N}^{1,(n)}$. To do
so, we suppose that the verifier has access to unitary circuits that realize
isometric extensions of all channels involved in the strategies. That is, for
$i\in\left\{  0,1\right\}  $, there exists a unitary channel $\mathcal{U}%
_{A_{1}E_{1}^{\prime}\rightarrow M_{1}B_{1}E_{1}}^{i,1}$ such that%
\begin{multline}
\mathcal{N}_{A_{1}\rightarrow M_{1}B_{1}}^{i,1}(\rho_{A_{1}})=\\
\operatorname{Tr}_{E_{1}}[\mathcal{U}_{A_{1}E_{1}^{\prime}\rightarrow
M_{1}B_{1}E_{1}}^{i,1}(\rho_{A_{1}}\otimes|0\rangle\!\langle0|_{E_{1}^{\prime
}})]
\end{multline}
for every input state $\rho_{A_{1}}$; for $j\in\left\{  2,\ldots,n-1\right\}
$, there exists a unitary channel $\mathcal{U}_{M_{j-1}A_{j}E_{j}^{\prime
}\rightarrow M_{j}B_{j}E_{j}}^{i,j}$\ such that%
\begin{multline}
\mathcal{N}_{M_{j-1}A_{j}\rightarrow M_{j}B_{j}}^{i,j}(\rho_{A_{j}})=\\
\operatorname{Tr}_{E_{j}}[\mathcal{U}_{M_{j-1}A_{j}E_{j}^{\prime}\rightarrow
M_{j}B_{j}E_{j}}^{i,j}(\rho_{A_{j}}\otimes|0\rangle\!\langle0|_{E_{j}^{\prime
}})],
\end{multline}
for every input state $\rho_{A_{j}}$; and there exists a unitary channel
$\mathcal{U}_{M_{n-1}A_{n}E_{n}^{\prime}\rightarrow B_{n}E_{n}}^{i,n}$ such
that%
\begin{multline}
\mathcal{N}_{M_{n-1}A_{n}\rightarrow B_{n}}^{i,n}(\rho_{A_{n}})=\\
\operatorname{Tr}_{E_{n}}[\mathcal{U}_{M_{n-1}A_{n}E_{n}^{\prime}\rightarrow
B_{n}E_{n}}^{i,n}(\rho_{A_{n}}\otimes|0\rangle\!\langle0|_{E_{n}^{\prime}})],
\end{multline}
for every input state $\rho_{A_{n}}$. We use the notation $U_{A_{1}%
E_{1}^{\prime}\rightarrow M_{1}B_{1}E_{1}}^{i,1}$, $U_{M_{j-1}A_{j}%
E_{j}^{\prime}\rightarrow M_{j}B_{j}E_{j}}^{i,j}$, and $U_{M_{n-1}A_{n}%
E_{n}^{\prime}\rightarrow B_{n}E_{n}}^{i,n}$ to refer to the unitary circuits.

We now provide a DQIP algorithm for estimating the following quantity:%
\begin{equation}
\frac{1}{2}\left(  1+\sqrt{F}(\mathcal{N}^{0,(n)},\mathcal{N}^{1,(n)})\right)
,
\end{equation}
which is based in part on Algorithm~\ref{alg:fid-channels}\ but instead
features an optimization over all co-strategies of the min-prover.

\begin{algorithm}
\label{alg:fid-strategies} The algorithm proceeds as follows:

\begin{enumerate}
\item The verifier prepares a Bell state
\begin{equation}
|\Phi\rangle_{T^{\prime}T}\coloneqq\frac{1}{\sqrt{2}}(|00\rangle_{T^{\prime}%
T}+|11\rangle_{T^{\prime}T})
\end{equation}
on registers $T^{\prime}$ and $T$ and prepares systems $E_{1}^{\prime}\cdots
E_{n}^{\prime}$ in the all-zeros state $|0\rangle_{E_{1}^{\prime}\cdots
E_{n}^{\prime}}$.

\item The min-prover transmits the system $A$ of the state $|\psi\rangle_{RA}$
to the verifier.

\item Using the circuits $U_{A_{1}E_{1}^{\prime}\rightarrow M_{1}B_{1}E_{1}%
}^{0,1}$ and $U_{A_{1}E_{1}^{\prime}\rightarrow M_{1}B_{1}E_{1}}^{1,1}$, the
verifier performs the following controlled unitary:%
\begin{equation}
\sum_{i\in\left\{  0,1\right\}  }|i\rangle\!\langle i|_{T}\otimes
U_{A_{1}E_{1}^{\prime}\rightarrow M_{1}B_{1}E_{1}}^{i,1}.
\end{equation}

\item The verifier transmits system $B_{1}$ to the min-prover, who
subsequently acts with the isometric quantum channel $\mathcal{S}_{R_{1}%
B_{1}\rightarrow R_{2}A_{2}}^{1}$ and then sends system $A_{2}$ to the verifier.

\item For $j\in\left\{  2,\ldots,n-1\right\}  $, using the circuits
$U_{M_{j-1}A_{j}E_{j}^{\prime}\rightarrow M_{j}B_{j}E_{j}}^{0,j}$ and
$U_{M_{j-1}A_{j}E_{j}^{\prime}\rightarrow M_{j}B_{j}E_{j}}^{1,j}$, the
verifier performs the following controlled unitary:%
\begin{equation}
\sum_{i\in\left\{  0,1\right\}  }|i\rangle\!\langle i|_{T}\otimes
U_{M_{j-1}A_{j}E_{j}^{\prime}\rightarrow M_{j}B_{j}E_{j}}^{i,j}.
\end{equation}
The verifier transmits system $B_{j}$ to the min-prover, who subsequently acts
with the isometric quantum channel $\mathcal{S}_{R_{j}B_{j}\rightarrow
R_{j+1}A_{j+1}}^{j}$ and then sends system $A_{j+1}$ to the verifier.

\item Using the circuits $U_{M_{n-1}A_{n}E_{n}^{\prime}\rightarrow B_{n}E_{n}%
}^{0,n}$ and $U_{M_{n-1}A_{n}E_{n}^{\prime}\rightarrow B_{n}E_{n}}^{1,n}$, the
verifier performs the following controlled unitary:%
\begin{equation}
\sum_{i\in\left\{  0,1\right\}  }|i\rangle\!\langle i|_{T}\otimes
U_{M_{n-1}A_{n}E_{n}^{\prime}\rightarrow B_{n}E_{n}}^{i,n}.
\end{equation}

\item The verifier transmits systems $T^{\prime}$, $E_{1}$, \ldots, $E_{n}$ to
the max-prover.

\item The max-prover prepares a system $F$ in the $|0\rangle_{F}$ state and
acts on systems $T^{\prime}$, $E_{1}$, \ldots, $E_{n}$, and $F$ with a unitary
$P_{T^{\prime}E_{1}\cdots E_{n}F\rightarrow T^{\prime\prime}F^{\prime}}$ to
produce the output systems $T^{\prime\prime}$ and $F^{\prime}$, where
$T^{\prime\prime}$ is a qubit system.

\item The max-prover sends system $T^{\prime\prime}$ to the verifier, who then
performs a Bell measurement%
\begin{equation}
\{\Phi_{T^{\prime\prime}T},I_{T^{\prime\prime}T}-\Phi_{T^{\prime\prime}T}\}
\end{equation}
on systems $T^{\prime\prime}$ and $T$. The verifier accepts if and only if the
outcome $\Phi_{T^{\prime\prime}T}$ occurs.
\end{enumerate}
\end{algorithm}

Figure~\ref{fig:fid-alg-strategies-min-max-1} depicts
Algorithm~\ref{alg:fid-strategies}.

\begin{figure*}
\begin{center}
\includegraphics[
width=\linewidth
]{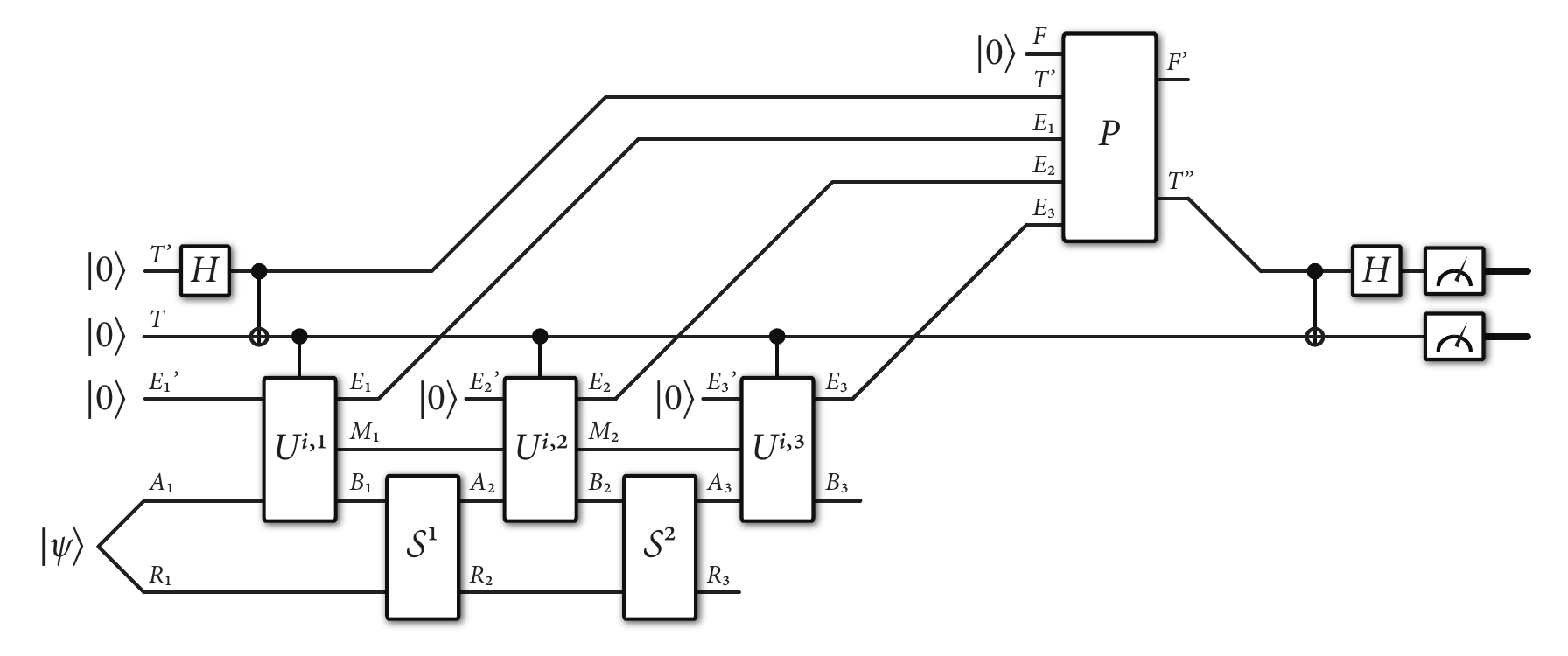}
\end{center}
\caption{This figure depicts Algorithm~\ref{alg:fid-strategies} for estimating
the fidelity of quantum strategies $\mathcal{N}^{0,(n)}$ and $\mathcal{N}%
^{1,(n)}$ generated by quantum circuits $U_{A_{1}E_{1}^{\prime}\rightarrow
M_{1}B_{1}E_{1}}^{i,1}$, $\,\{U_{M_{j-1}A_{j}E_{j}^{\prime}\rightarrow
M_{j}B_{j}E_{j}}^{i,j}\}_{j=2}^{n-1}$, and $U_{M_{n-1}A_{n}E_{n}^{\prime
}\rightarrow B_{n}E_{n}}^{i,n}$ for $i\in\{0,1\}$ and $n=3$. The min-prover
prepares the state $|\psi\rangle_{RA}$ and acts with a co-strategy, while the
max-prover acts with the unitary $P_{T^{\prime}E_{1}\cdots E_{n}F\rightarrow
T^{\prime\prime}F^{\prime}}$.}%
\label{fig:fid-alg-strategies-min-max-1}%
\end{figure*}

\begin{theorem}
\label{thm:acc-prob-fid-strategies}
The acceptance probability of Algorithm~\ref{alg:fid-strategies} is equal to%
\begin{equation}
\frac{1}{2}\left(  1+\sqrt{F}(\mathcal{N}^{0,(n)},\mathcal{N}^{1,(n)})\right)
, \label{eq:acc-prob-fid-strategies}%
\end{equation}
where $\sqrt{F}(\mathcal{N}^{0,(n)},\mathcal{N}^{1,(n)})$ is the strategy
fidelity defined in \eqref{eq:def-fid-strategies}.
\end{theorem}

\begin{proof}
The proof can be found in Appendix~\ref{app:ProofAlg9}.
\end{proof}

\subsection{Alternate methods of estimating the fidelity of channels and
strategies}

\label{sec:alt-methods-fid-channels}

We note briefly here that other methods for estimating fidelity of channels
can be based on Algorithms~\ref{alg:mixed-state-swap-test},
\ref{alg:mixed-state-Bell-tests}, and~\ref{alg:mixed-state-FC-meas-min}. It is
not clear how to phrase them in the language of quantum interactive proofs, in
such a way that the acceptance probability is a simple function of the channel
fidelity. However, we can employ variational algorithms in which we repeat the
circuit for determining an optimal input state $\psi_{RA}$ for the channel
fidelity. Then these variational algorithms employ an extra minimization step
in order to approximate an optimal input state for the channel fidelity.

Similarly, we can estimate the fidelity of strategies by employing a sequence
of parameterized circuits to function as a co-strategy and then minimize
over them, in conjunction with any of the previous methods for estimating
fidelity of states.

\subsection{Estimating maximum output fidelity of channels}

\label{sec:max-output-fid-channels}

In this section, we show how a simple variation of
Algorithm~\ref{alg:fid-channels}, in which we combine the actions of the
min-prover and max-prover into a single max-prover, leads to a QIP\ algorithm
for estimating the following fidelity function of two quantum channels
$\mathcal{N}_{A\rightarrow B}^{0}$ and $\mathcal{N}_{A\rightarrow B}^{1}$:%
\begin{equation}
F_{\text{max}}(\mathcal{N}^0, \mathcal{N}^1) := \sup_{\rho_{A}}F(\mathcal{N}_{A\rightarrow B}^{0}(\rho_{A}),\mathcal{N}%
_{A\rightarrow B}^{1}(\rho_{A})), \label{eq:max-fid-channels-fixed-point}%
\end{equation}
where the optimization is over every input state $\rho_{A}$. This algorithm is
based in part on Algorithm~\ref{alg:fid-states}\ but instead features an
optimization over input states of the prover.

\begin{algorithm}
\label{alg:fid-channels-single-prover} The algorithm proceeds as follows:

\begin{enumerate}
\item The verifier prepares a Bell state
\begin{equation}
|\Phi\rangle_{T^{\prime}T}\coloneqq\frac{1}{\sqrt{2}}(|00\rangle_{T^{\prime}%
T}+|11\rangle_{T^{\prime}T})
\end{equation}
on registers $T^{\prime}$ and $T$ and prepares system $E^{\prime}$ in the
all-zeros state $|0\rangle_{E^{\prime}}$.

\item The prover transmits the system $A$ of the state $|\psi\rangle_{RA}$ to
the verifier.

\item Using the circuits $U_{AE^{\prime}\rightarrow BE}^{0}$ and
$U_{AE^{\prime}\rightarrow BE}^{1}$, the verifier performs the following
controlled unitary:%
\begin{equation}
\sum_{i\in\left\{  0,1\right\}  }|i\rangle\!\langle i|_{T}\otimes
U_{AE^{\prime}\rightarrow BE}^{i}.
\end{equation}

\item The verifier transmits systems $T^{\prime}$ and $E$ to the prover.

\item The prover prepares a system $F$ in the $|0\rangle_{F}$ state and acts
on systems $T^{\prime}$, $E$, and $F$ with a unitary $P_{T^{\prime
}EF\rightarrow T^{\prime\prime}F^{\prime}}$ to produce the output systems
$T^{\prime\prime}$ and $F^{\prime}$, where $T^{\prime\prime}$ is a qubit system.

\item The prover sends system $T^{\prime\prime}$ to the verifier, who then
performs a Bell measurement%
\begin{equation}
\{\Phi_{T^{\prime\prime}T},I_{T^{\prime\prime}T}-\Phi_{T^{\prime\prime}T}\}
\end{equation}
on systems $T^{\prime\prime}$ and $T$. The verifier accepts if and only if the
outcome $\Phi_{T^{\prime\prime}T}$ occurs.
\end{enumerate}
\end{algorithm}

Figure~\ref{fig:fid-alg-fixed-point-channels} depicts
Algorithm~\ref{alg:fid-channels-single-prover}.

\begin{figure*}
\begin{center}
\includegraphics[
width=\linewidth
]{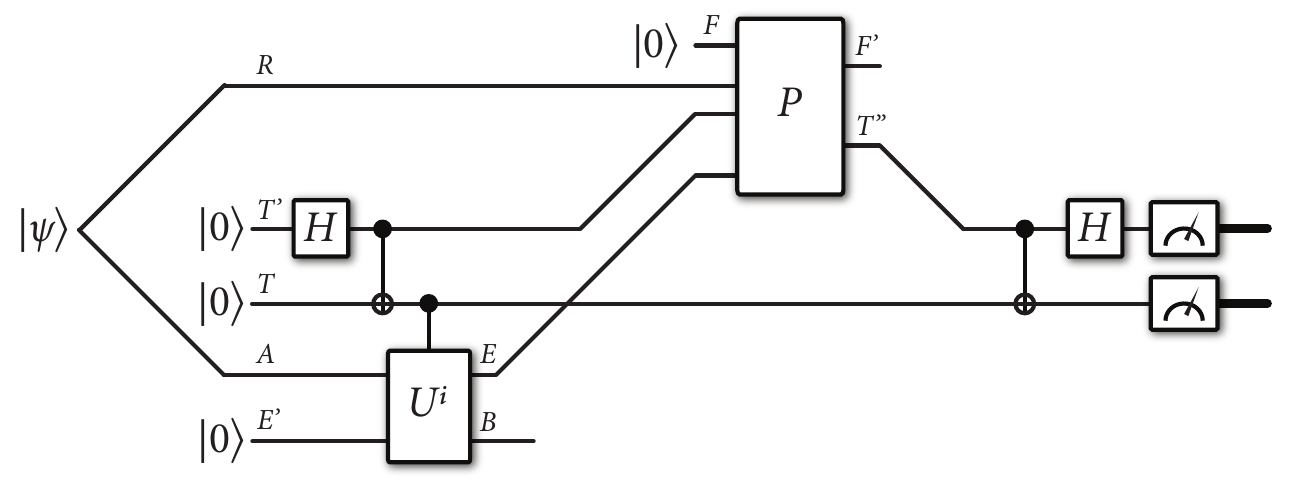}
\end{center}
\caption{This figure depicts Algorithm~\ref{alg:fid-channels-single-prover}
for generating a state $\rho_{A}$ that maximizes the fidelity of quantum
channels generated by quantum circuits $U_{AE^{\prime}\rightarrow BE}^{0}$ and
$U_{AE^{\prime}\rightarrow BE}^{1}$.}%
\label{fig:fid-alg-fixed-point-channels}%
\end{figure*}

\begin{theorem}
\label{thm:max-fid-channels}
The acceptance probability of Algorithm~\ref{alg:fid-channels-single-prover}%
\ is equal to%
\begin{equation}
\frac{1}{2}\left( 1 + \sqrt{F_{\operatorname{max}}}(\mathcal{N}_{A\rightarrow B}%
^{0},\mathcal{N}_{A\rightarrow B}^{1})\right)  .
\label{eq:max-fid-channels}%
\end{equation}

\end{theorem}

\begin{proof}
The proof can be found in Appendix~\ref{app:ProofAlg10}.
\end{proof}

\subsection{Generalization to multiple states}

\label{sec:mult-states}

In this section, we generalize Algorithm~\ref{alg:fid-states}\ to multiple
states, by devising a quantum algorithm that tests how similar all the states
of an ensemble are to each other.

Suppose that we are given an ensemble $\left\{  p(x),\rho_{S}^{x}\right\}
_{x\in\mathcal{X}}$ of states of system $S$, with $d=\left\vert \mathcal{X}%
\right\vert $, and we would like to know how similar they are to each other.
Then we can perform a test like that given in Algorithm~\ref{alg:fid-states},
but it is a multiple-state similarity test. The main difference is that the
verifier prepares an initial entangled state that encodes the prior
probabilities $\left\{  p(x)\right\}  _{x\in\mathcal{X}}$ and the algorithm
employs $d$-dimensional control systems throughout, instead of qubit control
systems.\ We suppose that, for all $x\in\mathcal{X}$, there is a circuit
$U_{RS}^{x}$ that generates a purification $|\psi^{x}\rangle_{RS}$ as follows:%
\begin{align}
|\psi^{x}\rangle_{RS}  &  \coloneqq U_{RS}^{x}|0\rangle_{RS},\\
\rho_{S}^{x}  &  =\operatorname{Tr}_{R}[|\psi^{x}\rangle\!\langle\psi
^{x}|_{RS}].
\end{align}

\begin{algorithm}
\label{alg:fid-multiple-states} The algorithm proceeds as follows:

\begin{enumerate}
\item The verifier prepares a state%
\begin{equation}
|\Phi^{p}\rangle_{T^{\prime}T}\coloneqq\sum_{x\in\mathcal{X}}\sqrt
{p(x)}|xx\rangle_{T^{\prime}T}%
\end{equation}
on registers $T^{\prime}$ and $T$ and prepares systems $RS$ in the all-zeros
state $|0\rangle_{RS}$.

\item Using the circuits in the set $\{U_{RS}^{x}\}_{x\in\mathcal{X}}$, the
verifier performs the following controlled unitary:%
\begin{equation}
\sum_{x\in\mathcal{X}}|x\rangle\!\langle x|_{T}\otimes U_{RS}^{x}.
\end{equation}

\item The verifier transmits systems $T^{\prime}$ and $R$ to the prover.

\item The prover prepares a system $F$ in the $|0\rangle_{F}$ state and acts
on systems $T^{\prime}$, $R$, and $F$ with a unitary $P_{T^{\prime
}RF\rightarrow T^{\prime\prime}F^{\prime}}$ to produce the output systems
$T^{\prime\prime}$ and $F^{\prime}$, where $T^{\prime\prime}$ is a qudit system.

\item The prover sends system $T^{\prime\prime}$ to the verifier, who then
performs a qudit Bell measurement%
\begin{equation}
\{\Phi_{T^{\prime\prime}T},I_{T^{\prime\prime}T}-\Phi_{T^{\prime\prime}T}\}
\end{equation}
on systems $T^{\prime\prime}$ and $T$, where%
\begin{align}
\Phi_{T^{\prime\prime}T}  &  =|\Phi\rangle\!\langle\Phi|_{T^{\prime\prime}%
T},\label{eq:qudit-bell-state}\\
|\Phi\rangle_{T^{\prime\prime}T}  &  \coloneqq\frac{1}{\sqrt{d}}\sum
_{x\in\mathcal{X}}|xx\rangle_{T^{\prime\prime}T}.
\end{align}
The verifier accepts if and only if the outcome $\Phi_{T^{\prime\prime}T}$ occurs.
\end{enumerate}
\end{algorithm}

\begin{theorem}
\label{thm:fid-multiple-states}
The acceptance probability of
Algorithm~\ref{alg:fid-multiple-states}\ is equal to%
\begin{equation}
 p_{\operatorname{sim}}(\left\{  p(x),\rho_{S}^{x}\right\}_{x\in\mathcal{X}}) \coloneqq \frac{1}{d}\left[  \sup_{\sigma_{S}}\sum_{x\in\mathcal{X}}\sqrt{p(x)}\sqrt
{F}(\rho_{S}^{x},\sigma_{S})\right]  ^{2}, \label{eq:acc-prob-mult-states}%
\end{equation}
where the optimization is over every density operator $\sigma_{S}$. This
acceptance probability is bounded from above by%
\begin{equation}
\frac{1}{d}+\frac{2}{d}\sum_{x,y\in\mathcal{X}:x<y}\sqrt{p(x)p(y)}\sqrt
{F}(\rho_{S}^{x},\rho_{S}^{y}). \label{eq:upper-bound-fid-multiple-states}%
\end{equation}
When $d=2$, this upper bound is tight.
\end{theorem}

\begin{proof}
The proof can be found in Appendix~\ref{app:ProofAlg11}.
\end{proof}

\begin{corollary}
The fact that the upper bound is achieved in
Theorem~\ref{thm:fid-multiple-states}\ for $d=2$ leads to the following
identity for states $\rho_{S}^{0}$ and $\rho_{S}^{1}$ and probability
$p\in\left[  0,1\right]  $:%
\begin{multline}
\left[  \sup_{\sigma_{S}}\sqrt{p}\sqrt{F}(\rho_{S}^{0},\sigma_{S})+\sqrt
{1-p}\sqrt{F}(\rho_{S}^{1},\sigma_{S})\right]  ^{2}\\
=1+2\sqrt{p\left(  1-p\right)  }\sqrt{F}(\rho_{S}^{0},\rho_{S}^{1}),
\end{multline}
where the optimization is over every density operator $\sigma_{S}$.
\end{corollary}

The acceptance probability in \eqref{eq:acc-prob-mult-states} is proportional
to the secrecy measure discussed in \cite[Eq.~(19)]{KRS09}, which is the same
as the max-conditional entropy of the following classical--quantum state:%
\begin{equation}
\sum_{x\in\mathcal{X}}p(x)|x\rangle\!\langle x|_{T}\otimes\rho_{S}^{x}.
\label{eq:cq-state-multiple}%
\end{equation}
Indeed, it is a measure of secrecy because if an eavesdropper has access to
system $S$ and if $\rho_{S}^{x}\approx\sigma$ for all $x\in\mathcal{X}$ and if
$p(x)\approx1/d$, then it is difficult for the eavesdropper to guess the
classical message in system $T$ (also, the fidelity is close to one). According to \cite[Remark~2.7]{SDGWW21} and
the expression in \eqref{eq:sym-dist-meas} of Appendix~\ref{app:ProofAlg11}, the acceptance probability in
\eqref{eq:acc-prob-mult-states} is also a measure of the symmetric
distinguishability of the classical--quantum state in
\eqref{eq:cq-state-multiple}, and thus gives this measure an operational meaning.

The upper bound in \eqref{eq:upper-bound-fid-multiple-states}\ on the
acceptance probability has some conceptual similarity with known upper bounds
on the success probability in state discrimination \cite{Mont08,Qiu08}, in the
sense that we employ the fidelity of pairs of states in the upper bound.
Finally, we note some similarities between the problem outlined here and
coherent channel discrimination considered recently in \cite{Wilde20}.
However, these two problems are ultimately different in their objectives.

\subsection{Generalization to multiple channels and strategies}

\label{sec:mult-channels-strategies}

We now generalize Algorithms~\ref{alg:fid-channels}\ and
\ref{alg:fid-multiple-states} to the case of testing the similarity of an
ensemble of channels. The resulting algorithm thus has applications in the
context of private quantum reading \cite{BDW18,DBW20}, in which one goal of
such a protocol is to encode a classical message into a channel selected
randomly from an ensemble of channels such that it is indecipherable by an
eavesdropper who has access to the output of the channel. We also remark at
the end of this section about a generalization of
Algorithms~\ref{alg:fid-strategies} and \ref{alg:fid-multiple-channels} to the
case of an ensemble of $n$-turn quantum strategies.

Let us first consider the case of channels. In more detail, let
$\{p(x),\mathcal{N}_{A\rightarrow B}^{x}\}_{x\in\mathcal{X}}$ be an ensemble
of quantum channels. Set $d=\left\vert \mathcal{X}\right\vert $. We suppose that, for all
$x\in\mathcal{X}$, there is a circuit $U_{AE^{\prime}\rightarrow BE}^{x}$ that
generates an isometric extension of the channel $\mathcal{N}_{A\rightarrow
B}^{x}$, in the following sense:%
\begin{multline}
\mathcal{N}_{A\rightarrow B}^{x}(\omega_{A})= \label{eq:iso-ext-channels-mult}%
\\
\operatorname{Tr}_{E}[U_{AE^{\prime}\rightarrow BE}^{x}(\omega_{A}%
\otimes|0\rangle\!\langle0|_{E^{\prime}})(U_{AE^{\prime}\rightarrow BE}%
^{x})^{\dag}].
\end{multline}
The following
algorithm employs competing provers, similar to how
Algorithm~\ref{alg:fid-channels}\ does. 

\begin{algorithm}
\label{alg:fid-multiple-channels} The algorithm proceeds as follows:

\begin{enumerate}
\item The verifier prepares a state%
\begin{equation}
|\Phi^{p}\rangle_{T^{\prime}T}\coloneqq\sum_{x\in\mathcal{X}}\sqrt
{p(x)}|xx\rangle_{T^{\prime}T}%
\end{equation}
on registers $T^{\prime}$ and $T$ and prepares system $E^{\prime}$ in the
all-zeros state $|0\rangle_{RS}$.

\item The min-prover transmits the system $A$ of the state $|\psi\rangle_{RA}$
to the verifier.

\item Using the circuits in the set $\{U_{AE^{\prime}\rightarrow BE}%
^{x}\}_{x\in\mathcal{X}}$, the verifier performs the following controlled
unitary:%
\begin{equation}
\sum_{x\in\mathcal{X}}|x\rangle\!\langle x|_{T}\otimes U_{AE^{\prime
}\rightarrow BE}^{x}.
\end{equation}

\item The verifier transmits systems $T^{\prime}$ and $E$ to the max-prover.

\item The max-prover prepares a system $F$ in the $|0\rangle_{F}$ state and
acts on systems $T^{\prime}$, $R$, and $F$ with a unitary $P_{T^{\prime
}EF\rightarrow T^{\prime\prime}F^{\prime}}$ to produce the output systems
$T^{\prime\prime}$ and $F^{\prime}$, where $T^{\prime\prime}$ is a qudit system.

\item The max-prover sends system $T^{\prime\prime}$ to the verifier, who then
performs a qudit Bell measurement%
\begin{equation}
\{\Phi_{T^{\prime\prime}T},I_{T^{\prime\prime}T}-\Phi_{T^{\prime\prime}T}\}
\end{equation}
on systems $T^{\prime\prime}$ and $T$, where $\Phi_{T^{\prime\prime}T}$ is
defined in~\eqref{eq:qudit-bell-state}. The verifier accepts if and only if
the outcome $\Phi_{T^{\prime\prime}T}$ occurs.
\end{enumerate}
\end{algorithm}

\begin{theorem}
\label{thm:fid-multiple-channels}
The acceptance probability of
Algorithm~\ref{alg:fid-multiple-channels}\ is equal to%
\begin{multline}
p_{\operatorname{sim}}(\{p(x),\mathcal{N}^{x}\}_{x\in\mathcal{X}}) = \\
\frac{1}{d}\left[  \inf_{\psi_{RA}}\sup_{\sigma_{RB}}\sum_{x\in\mathcal{X}%
}\sqrt{p(x)}\sqrt{F}(\mathcal{N}_{A\rightarrow B}^{x}(\psi_{RA}),\sigma
_{RB})\right]  ^{2}. \label{eq:accept-prob-mult-channels}%
\end{multline}
This acceptance probability is bounded from above by%
\begin{multline}
\frac{1}{d}+\frac{2}{d}\times\label{eq:upper-bound-accept-prob-mult-ch}\\
\inf_{\psi_{RA}}\sum_{\substack{x,y\in\mathcal{X}:\\x<y}}\sqrt{p(x)p(y)}%
\sqrt{F}(\mathcal{N}_{A\rightarrow B}^{x}(\psi_{RA}),\mathcal{N}_{A\rightarrow
B}^{y}(\psi_{RA})).
\end{multline}
When $d=2$, this upper bound is tight.
\end{theorem}

\begin{proof}
The proof can be found in Appendix~\ref{app:ProofAlg12}.
\end{proof}

\begin{corollary}
The following identity holds in the special case of two channels
$\mathcal{N}_{A\rightarrow B}^{0}$ and $\mathcal{N}_{A\rightarrow B}^{1}$ and
probability $p\in\left[  0,1\right]  $:%
\begin{multline}
\left[  \inf_{\psi_{RA}}\sup_{\sigma_{RB}}\left(
\begin{array}
[c]{c}%
\sqrt{p}\sqrt{F}(\mathcal{N}_{A\rightarrow B}^{0}(\psi_{RA}),\sigma_{RB})\\
+\sqrt{1-p}\sqrt{F}(\mathcal{N}_{A\rightarrow B}^{1}(\psi_{RA}),\sigma_{RB})
\end{array}
\right)  \right]  ^{2}\\
=1+2\sqrt{p\left(  1-p\right)  }\inf_{\psi_{RA}}\sqrt{F}(\mathcal{N}%
_{A\rightarrow B}^{0}(\psi_{RA}),\mathcal{N}_{A\rightarrow B}^{1}(\psi_{RA})),
\end{multline}
where the supremum is with respect to every density operator $\sigma_{RB}$.
\end{corollary}

\begin{remark}
We note here that we can generalize the developments in this section and the
previous one to the case of quantum strategies, in order to test how similar
strategies in a set are to each other. Let $\{p(x),\mathcal{N}^{x,(n)}%
\}_{x\in\mathcal{X}}$ be an ensemble of quantum strategies, each of which has
$n$ turns. Then the acceptance probability of an algorithm that is the obvious
generalization of Algorithms~\ref{alg:fid-strategies} and
\ref{alg:fid-multiple-channels}\ is given by%
\begin{equation}
\frac{1}{d}\left[  \inf_{\mathcal{S}^{(n-1)}}\sup_{\sigma}\sum_{x\in
\mathcal{X}}\sqrt{p(x)}\sqrt{F}(\mathcal{N}^{x,(n)}\circ\mathcal{S}%
^{(n-1)},\sigma_{R_{n}B_{n}})\right]  ^{2},
\label{eq:similarity-multiple-strategies}%
\end{equation}
where the infimum is with respect to every $(n-1)$-turn pure co-strategy that
leads to a quantum state $\mathcal{N}^{x,(n)}\circ\mathcal{S}^{(n-1)}$ (as
discussed around \eqref{eq:state-after-strategy-co-strategy}) and the supremum
is with respect to every state $\sigma_{R_{n}B_{n}}$. The expression in
\eqref{eq:similarity-multiple-strategies} is a similarity measure for the
strategies in the ensemble $\{p(x),\mathcal{N}^{x,(n)}\}_{x\in\mathcal{X}}$.
\end{remark}

We can also generalize Algorithm~\ref{alg:fid-channels-single-prover}\ from
Section~\ref{sec:max-output-fid-channels}, to estimate the following
similarity measure for an ensemble $\{p(x),\mathcal{N}_{A\rightarrow B}%
^{x}\}_{x\in\mathcal{X}}$\ of channels:%
\begin{equation}
\frac{1}{d}\left[  \sup_{\rho_{A},\sigma_{B}}\sum_{x\in\mathcal{X}}\sqrt
{p(x)}\sqrt{F}(\mathcal{N}_{A\rightarrow B}^{x}(\rho_{A}),\sigma_{B})\right]
^{2},
\end{equation}
where the optimization is over all density operators $\rho_{A}$ and
$\sigma_{B}$. As is the case with
Algorithm~\ref{alg:fid-channels-single-prover}, there is a single prover who
is trying to make all of the channel outputs look like the same state. Again
we suppose that there is a circuit $U_{AE^{\prime}\rightarrow BE}^{x}$ that
generates an isometric extension of the channel $\mathcal{N}_{A\rightarrow
B}^{x}$, in the sense of \eqref{eq:iso-ext-channels-mult}.

\begin{algorithm}
\label{alg:fid-multiple-channels-single-prover} The algorithm proceeds as follows:

\begin{enumerate}
\item The verifier prepares a state%
\begin{equation}
|\Phi^{p}\rangle_{T^{\prime}T}\coloneqq\sum_{x\in\mathcal{X}}\sqrt
{p(x)}|xx\rangle_{T^{\prime}T}%
\end{equation}
on registers $T^{\prime}$ and $T$ and prepares system $E^{\prime}$ in the
all-zeros state $|0\rangle_{E^{\prime}}$.

\item The prover transmits the system $A$ of the state $|\psi\rangle_{RA}$ to
the verifier.

\item Using the circuits in the set $\{U_{AE^{\prime}\rightarrow BE}%
^{x}\}_{x\in\mathcal{X}}$, the verifier performs the following controlled
unitary:%
\begin{equation}
\sum_{x\in\mathcal{X}}|x\rangle\!\langle x|_{T}\otimes U_{AE^{\prime
}\rightarrow BE}^{x}.
\end{equation}

\item The verifier transmits systems $T^{\prime}$ and $E$ to the max-prover.

\item The prover prepares a system $F$ in the $|0\rangle_{F}$ state and acts
on systems $T^{\prime}$, $R$, and $F$ with a unitary $P_{T^{\prime
}EF\rightarrow T^{\prime\prime}F^{\prime}}$ to produce the output systems
$T^{\prime\prime}$ and $F^{\prime}$, where $T^{\prime\prime}$ is a qudit system.

\item The prover sends system $T^{\prime\prime}$ to the verifier, who then
performs a qudit Bell measurement%
\begin{equation}
\{\Phi_{T^{\prime\prime}T},I_{T^{\prime\prime}T}-\Phi_{T^{\prime\prime}T}\}
\end{equation}
on systems $T^{\prime\prime}$ and $T$, where $\Phi_{T^{\prime\prime}T}$ is
defined in \eqref{eq:qudit-bell-state}. The verifier accepts if and only if
the outcome $\Phi_{T^{\prime\prime}T}$ occurs.
\end{enumerate}
\end{algorithm}

\begin{theorem}
\label{thm:fid-multiple-channels-single-prover}The acceptance probability of
Algorithm~\ref{alg:fid-multiple-channels-single-prover}\ is equal to%
\begin{multline}
p_{\operatorname{sim,max}}(\{p(x),\mathcal{N}^{x}\}_{x\in\mathcal{X}}) = \\
\frac{1}{d}\left[  \sup_{\rho_{A},\sigma_{B}}\sum_{x\in\mathcal{X}}\sqrt
{p(x)}\sqrt{F}(\mathcal{N}_{A\rightarrow B}^{x}(\rho_{A}),\sigma_{B})\right]
^{2}.
\end{multline}
This acceptance probability is bounded from above by%
\begin{multline}
\frac{1}{d}+\frac{2}{d}\times\\
\sup_{\rho_{A}}\sum_{x,y\in\mathcal{X}:x<y}\sqrt{p(x)p(y)}\sqrt{F}%
(\mathcal{N}_{A\rightarrow B}^{x}(\rho_{A}),\mathcal{N}_{A\rightarrow B}%
^{y}(\rho_{A})).
\end{multline}
When $d=2$, this upper bound is tight.
\end{theorem}

\begin{proof}
For a fixed state $\psi_{RA}$ of the prover, the problem is equivalent to that
specified by Algorithm~\ref{alg:fid-multiple-states}, for the ensemble
$\{p(x),F(\mathcal{N}_{A\rightarrow B}^{x}(\rho_{A})\}_{x\in\mathcal{X}}$,
where $\rho_{A}=\operatorname{Tr}_{A}[\psi_{RA}]$. Thus, all of the statements
from Theorem~\ref{thm:fid-multiple-states}\ apply for this fixed state. We
arrive at the statement of the theorem after optimizing over all input states.
\end{proof}

\section{Estimating trace distance, diamond distance, and strategy distance}

\label{sec:TD-based-measures}

We now review several well known algorithms for estimating trace distance \cite{W02}, diamond distance \cite{RW05}, and strategy distance \cite{GW06,G09,G12}\ by interacting with quantum provers. Later on, we replace the provers with parameterized circuits to see how well this approach can perform in estimating these distinguishability measures. A summary of the algorithms is presented in Table~\ref{tab:TdAlgs}.

\renewcommand{\arraystretch}{2.3}
\begin{table*}
\begin{center}
    \begin{tabular}{|M{3.5cm}||M{2.5cm}|M{8.5cm}|}
        \hline 
         \multirow{2}{*}{Problem} & 
         \multirow{2}{*}{Algorithms} & 
         \multirow{2}{*}{Comparison}\\
         & & \\
        \hline
        \multirow{1}{*}[0.15cm]{$\left\Vert \rho_0 - \rho_1 \right\Vert_1$} & \multirow{1}{*}[0.1cm]{Algorithm~\ref{alg:trace-distance-states}} & 
        Algorithm~\ref{alg:trace-distance-states} does not require the purifying system, unlike fidelity algorithms.  \\
        \hline
        \multirow{1}{*}[-0.05cm]{$\left\Vert \mathcal{N}_0 - \mathcal{N}_1 \right\Vert_\diamond $} & \multirow{1}{*}[-0.05cm]{Algorithm~\ref{alg:diamond-distance-channels}} & - \\
        \hline
        \multirow{1}{*}[-0.05cm]{$\left\Vert \mathcal{N}^{0, (n)} - \mathcal{N}^{1, (n)} \right\Vert_{\diamond n} $} & \multirow{1}{*}[-0.05cm]{Algorithm~\ref{alg:strategy-distance-strategies}} & - \\
        \hline
        
        \multirow{2}{*}[-0.1cm]{$\left\Vert \mathcal{N}_0 - \mathcal{N}_1 \right\Vert_{\diamond, \text{min}} $} & Algorithm~\ref{alg:min-trace-distance-channels} & \multirow[c]{2}{8.5cm}{\centering Algorithm~\ref{alg:min-trace-distance-channels-swap-min-max} swaps the role of the max-prover and min-prover from Algorithm~\ref{alg:min-trace-distance-channels}.} \\
        \cline{2-2}
        & Algorithm~\ref{alg:min-trace-distance-channels-swap-min-max} & \\
        \hline
        \multirow{1}{*}[-0.05cm]{$ p_g(\{p(x), \rho^x\}_{x \in \mathcal{X}})  $} & Algorithm~\ref{alg:mult-state-disc-qubit} & 
        \multirow{1}{*}[-0.05cm]{Generalizes Algorithm~\ref{alg:trace-distance-states} to ensemble of states. }\\
        \hline
        
    \end{tabular}
\caption{List of trace distance problems and algorithms addressed in this work. Approach used for each algorithm and comparison within a type of trace distance problem is also presented.}
\label{tab:TdAlgs}
\end{center}
\end{table*}

\subsection{Estimating trace distance}

\label{sec:est-trace-dist-states}

The trace distance between quantum states $\rho_{S}^{0}$ and $\rho_{S}^{1}$ is
defined as $\left\Vert \rho_{S}^{0}-\rho_{S}^{1}\right\Vert _{1}$, where
$\left\Vert A\right\Vert _{1}=\operatorname{Tr}[\sqrt{A^{\dag}A}]$. It is a
well known and operationally motivated measure of distinguishability for
quantum states.

We suppose, as is the case in Section~\ref{sec:fid-arbitrary-states}, that
quantum circuits $U_{RS}^{0}$\ and $U_{RS}^{1}$ are available for generating
purifications of the states $\rho_{S}^{0}$ and $\rho_{S}^{1}$. That is, for
$i\in\left\{  0,1\right\}  $,%
\begin{equation}
\rho_{S}^{i}=\operatorname{Tr}_{R}[U_{RS}^{i}|0\rangle\!\langle0|_{RS}%
(U_{RS}^{i})^{\dag}].
\end{equation}
However, the purifying systems are not strictly necessary in the operation of
the algorithm given below, which is an advantage over some of the algorithms from
Section~\ref{sec:fid-arbitrary-states}.

The following QSZK\ algorithm allows for estimating the trace distance
\cite{W02}, in the sense that its acceptance probability is a simple function
of the trace distance:

\begin{algorithm}
[\cite{H67,H69,Hol72,W02}]
\label{alg:trace-distance-states}The algorithm
proceeds as follows:

\begin{enumerate}
\item The verifier picks a classical bit $i\in\left\{  0,1\right\}  $
uniformly at random, prepares the state $\rho_{S}^{i}$, and sends system~$S$
to the prover.

\item The prover prepares a system $F$ in the $|0\rangle_{F}$ state and acts
on systems $S$ and $F$ with a unitary $P_{SF\rightarrow TF^{\prime}}$ to
produce the output systems $T$ and $F^{\prime}$, where $T$ is a qubit system.

\item The prover sends system $T$ to the verifier, who then performs a
measurement on system $T$, with outcome $j \in \{0,1\}$. The verifier accepts if and only
if $i=j$.
\end{enumerate}
\end{algorithm}

This algorithm has been well known for some time \cite{H67,H69,Hol72,W02} and
its maximum acceptance probability is equal to%
\begin{multline}
\max_{\Lambda:0\leq\Lambda\leq I}\frac{1}{2}\operatorname{Tr}[\Lambda\rho
_{S}^{0}]+\frac{1}{2}\operatorname{Tr}[(I-\Lambda)\rho_{S}^{1}]\\
=\frac{1}{2}\left(  1+\frac{1}{2}\left\Vert \rho_{S}^{0}-\rho_{S}%
^{1}\right\Vert _{1}\right)  .
\end{multline}
This follows because the acceptance probability can be written as follows, for
a fixed unitary $P\equiv P_{SF\rightarrow TF^{\prime}}$ of the prover:%
\begin{align}
&  \frac{1}{2}\sum_{i\in\left\{  0,1\right\}  }\operatorname{Tr}%
[(|i\rangle\!\langle i|_{T}\otimes I_{F^{\prime}})P(\rho_{S}^{i}%
\otimes|0\rangle\!\langle0|_{F})P^{\dag}]\nonumber\\
&  =\frac{1}{2}\sum_{i\in\left\{  0,1\right\}  }\operatorname{Tr}%
[\langle0|_{F}P^{\dag}(|i\rangle\!\langle i|_{T}\otimes I_{F^{\prime}%
})P|0\rangle_{F}\rho_{S}^{i}]\\
&  =\frac{1}{2}\sum_{i\in\left\{  0,1\right\}  }\operatorname{Tr}[\Lambda
_{S}^{i}\rho_{S}^{i}],
\end{align}
where we have defined the measurement operator $\Lambda_{S}^{i}$, for
$i\in\left\{  0,1\right\}  $, as%
\begin{equation}
\Lambda_{S}^{i}\coloneqq \langle0|_{F}(P_{SF\rightarrow TF^{\prime}})^{\dag
}(|i\rangle\!\langle i|_{T}\otimes I_{F^{\prime}})P_{SF\rightarrow TF^{\prime
}}|0\rangle_{F},
\end{equation}
and it is clear that $\sum_{i\in\left\{  0,1\right\}  }\Lambda_{S}^{i}=I_{S}$.
By the Naimark extension theorem \cite{N40} (see also \cite{KW20book}), every
measurement can be realized in this way, so that%
\begin{multline}
\max_{P}\frac{1}{2}\sum_{i\in\left\{  0,1\right\}  }\operatorname{Tr}%
[(|i\rangle\!\langle i|_{T}\otimes I_{F^{\prime}})P(\rho_{S}^{i}%
\otimes|0\rangle\!\langle0|_{F})P^{\dag}]\\
=\max_{\Lambda:0\leq\Lambda\leq I}\frac{1}{2}\operatorname{Tr}[\Lambda\rho
_{S}^{0}]+\frac{1}{2}\operatorname{Tr}[(I-\Lambda)\rho_{S}^{1}].
\end{multline}

Thus, by replacing the actions of the prover with a parameterized circuit and
repeating the algorithm, we can use a quantum computer to estimate a lower
bound on the trace distance of the states $\rho_{S}^{0}$ and $\rho_{S}^{1}$.
An approach similar to this has been adopted in \cite{CSZW20}.

We note here that the following identity holds also \cite{H67,H69,Hol72} (see
also \cite[Theorem~3.13]{KW20book}):%
\begin{multline}
\min_{\Lambda:0\leq\Lambda\leq I}\frac{1}{2}\operatorname{Tr}[\Lambda\rho
_{S}^{0}]+\frac{1}{2}\operatorname{Tr}[(I-\Lambda)\rho_{S}^{1}%
]\label{eq:min-formula-trace-distance}\\
=\frac{1}{2}\left(  1-\frac{1}{2}\left\Vert \rho_{S}^{0}-\rho_{S}%
^{1}\right\Vert _{1}\right)  .
\end{multline}

\subsection{Estimating diamond distance}

\label{sec:est-diamond-dist}

The diamond distance between quantum channels $\mathcal{N}_{A\rightarrow
B}^{0}$ and $\mathcal{N}_{A\rightarrow B}^{1}$ is defined as \cite{Kit97}%
\begin{multline}
\left\Vert \mathcal{N}_{A\rightarrow B}^{0}-\mathcal{N}_{A\rightarrow B}%
^{1}\right\Vert _{\diamond}\coloneqq\\
\sup_{\rho_{RA}}\left\Vert \mathcal{N}_{A\rightarrow B}^{0}(\rho
_{RA})-\mathcal{N}_{A\rightarrow B}^{1}(\rho_{RA})\right\Vert _{1},
\end{multline}
where the optimization is over every bipartite state $\rho_{RA}$ and the
system $R$ can be arbitrarily large. By a well known data processing argument,
the following equality holds%
\begin{multline}
\left\Vert \mathcal{N}_{A\rightarrow B}^{0}-\mathcal{N}_{A\rightarrow B}%
^{1}\right\Vert _{\diamond}\coloneqq\\
\max_{\psi_{RA}}\left\Vert \mathcal{N}_{A\rightarrow B}^{0}(\psi
_{RA})-\mathcal{N}_{A\rightarrow B}^{1}(\psi_{RA})\right\Vert _{1},
\end{multline}
where the optimization is over every pure bipartite state $\psi_{RA}$ and the
system $R$ is isomorphic to the channel input system $A$. The diamond distance
is a well known and operationally motivated measure of distinguishability for
quantum channels \cite{RW05,GLN04}.

We suppose, as is the case in Section~\ref{sec:fid-channels}, that quantum
circuits $U_{AE^{\prime}\rightarrow BE}^{0}$\ and $U_{AE^{\prime}\rightarrow
BE}^{1}$ are available for generating isometric extensions of the channels
$\mathcal{N}_{A\rightarrow B}^{0}$ and $\mathcal{N}_{A\rightarrow B}^{1}$.
That is, for $i\in\left\{  0,1\right\}  $,%
\begin{equation}
\mathcal{N}_{A\rightarrow B}^{i}(\cdot)=\operatorname{Tr}_{E}[U_{AE^{\prime
}\rightarrow BE}^{i}((\cdot)\otimes|0\rangle\!\langle0|_{E^{\prime}%
})(U_{AE^{\prime}\rightarrow BE}^{i})^{\dag}].
\end{equation}
However, the environment systems are not strictly necessary in the operation
of the algorithm given below, which is an advantage over some of  the algorithms from
Section~\ref{sec:fid-channels}.

The following QIP\ algorithm allows for estimating the diamond distance
\cite{RW05}, in the sense that its acceptance probability is a simple function
of the diamond distance:

\begin{algorithm}
[\cite{RW05}]\label{alg:diamond-distance-channels}The algorithm proceeds as follows:

\begin{enumerate}
\item The prover prepares a pure state $\psi_{RA}$ and sends system $A$ to the verifier.

\item The verifier picks a classical bit $i\in\left\{  0,1\right\}  $
uniformly at random, applies the channel $\mathcal{N}_{A\rightarrow B}^{i}$,
and sends system $B$ to the prover.

\item The prover prepares a system $F$ in the $|0\rangle_{F}$ state and acts
on systems $R$, $B$, and $F$ with a unitary $P_{RBF\rightarrow TF^{\prime}}$
to produce the output systems $T$ and $F^{\prime}$, where $T$ is a qubit system.

\item The prover sends system $T$ to the verifier, who then performs a
measurement on system $T$, with outcome $j \in \{0,1\}$. The verifier accepts if and only
if $i=j$.
\end{enumerate}
\end{algorithm}

This algorithm has been well known for some time \cite{RW05} and its maximum
acceptance probability is equal to%
\begin{equation}
\frac{1}{2}\left(  1+\frac{1}{2}\left\Vert \mathcal{N}_{A\rightarrow B}%
^{0}-\mathcal{N}_{A\rightarrow B}^{1}\right\Vert _{\diamond}\right)  .
\end{equation}
Thus, by replacing the actions of the prover with a parameterized circuit and
repeating the algorithm, we can use a quantum computer to estimate a lower
bound on the diamond distance of the channels $\mathcal{N}_{A\rightarrow
B}^{0}$ and $\mathcal{N}_{A\rightarrow B}^{1}$.

\subsection{Estimating strategy distance}

\label{sec:est-strategy-dist}

We already provided the definition of a quantum strategy in
Section~\ref{sec:strategy-fidelity}, and therein, we discussed the strategy
fidelity (see Eq.~\eqref{eq:def-fid-strategies}). The strategy distance
\cite{GW06,CDP08a,G12}\ is conceptually similar, but it is defined with the
trace distance as the underlying metric:%
\begin{multline}
\left\Vert \mathcal{N}^{0,(n)}-\mathcal{N}^{1,(n)}\right\Vert _{\diamond
n}\coloneqq\\
\sup_{\mathcal{S}^{(n-1)}}\left\Vert \mathcal{N}^{0,(n)}\circ\mathcal{S}%
^{(n-1)}-\mathcal{N}^{1,(n)}\circ\mathcal{S}^{(n-1)}\right\Vert _{1},
\end{multline}
where the supremum is with respect to every co-strategy $\mathcal{S}^{(n-1)}$
that leads to the quantum states $\mathcal{N}^{0,(n)}\circ\mathcal{S}^{(n-1)}$
and $\mathcal{N}^{1,(n)}\circ\mathcal{S}^{(n-1)}$ (here we have employed the same notation used in \eqref{eq:state-after-strategy-co-strategy}). The strategy distance is an
operationally motivated measure of distinguishability for quantum strategies.

The following QIP algorithm allows for estimating the strategy distance
\cite{GW06}, in the sense that its acceptance probability is a simple function
of the strategy distance:

\begin{algorithm}
[\cite{GW06}]\label{alg:strategy-distance-strategies}The algorithm proceeds as follows:

\begin{enumerate}
\item The prover prepares a pure state $\psi_{RA}$ and sends system $A$ to the verifier.

\item The verifier picks a classical bit $i\in\left\{  0,1\right\}  $
uniformly at random, applies the channel $\mathcal{N}_{A_{1}\rightarrow
M_{1}B_{1}}^{i,1}$, and sends system $B_{1}$ to the prover.

\item The prover acts with the isometric channel $\mathcal{S}_{R_{1}%
B_{1}\rightarrow R_{2}A_{2}}^{1}$ and then sends system $A_{2}$ to the verifier.

\item For $k\in\left\{  2,\ldots,n-1\right\}  $, the verifier applies the
channel $\mathcal{N}_{M_{k-1}A_{k}\rightarrow M_{k}B_{k}}^{i,k}$ and transmits
system $B_{k}$ to the prover, who subsequently acts with the isometric channel
$\mathcal{S}_{R_{k}B_{k}\rightarrow R_{k+1}A_{k+1}}^{k}$ and then sends system
$A_{k+1}$ to the verifier.

\item The verifier applies the channel $\mathcal{N}_{M_{n-1}A_{n}\rightarrow
B_{n}}^{i,n}$ and sends system $B_{n}$ to the prover.

\item The prover prepares a system $F$ in the $|0\rangle_{F}$ state and acts
on systems $R_{n}$, $B_{n}$, and $F$ with a unitary $P_{R_{n}B_{n}F\rightarrow
TF^{\prime}}$ to produce the output systems $T$ and $F^{\prime}$, where $T$ is
a qubit system.

\item The prover sends system $T$ to the verifier, who then performs a
measurement on system $T$, with outcome $j \in \{0,1\}$. The verifier accepts if and only
if $i=j$.
\end{enumerate}
\end{algorithm}

This algorithm has been well known since \cite{GW06} and its maximum
acceptance probability is equal to%
\begin{equation}
\frac{1}{2}\left(  1+\frac{1}{2}\left\Vert \mathcal{N}^{0,(n)}-\mathcal{N}%
^{1,(n)}\right\Vert _{\diamond n}\right)  .
\end{equation}
Thus, by replacing the actions of the prover with a parameterized circuit and
repeating the algorithm, we can use a quantum computer to estimate a lower
bound on the strategy distance of the strategies $\mathcal{N}^{0,(n)}$ and
$\mathcal{N}^{1,(n)}$. See \cite{G12,KW20adv}\ for semi-definite programs for
evaluating the strategy distance of two strategies.

\subsection{Estimating minimum trace distance of channels}

\label{sec:est-min-TD-channels}

In this section, we show how to estimate the following trace distance function
of channels $\mathcal{N}_{A\rightarrow B}^{0}$ and $\mathcal{N}_{A\rightarrow
B}^{1}$ by means of a short quantum game (SQG) algorithm:%
\begin{equation}
\inf_{\rho_{A}}\left\Vert \mathcal{N}_{A\rightarrow B}^{0}(\rho_{A}%
)-\mathcal{N}_{A\rightarrow B}^{1}(\rho_{A})\right\Vert _{1},
\label{eq:min-trace-distance-channels}%
\end{equation}
where the optimization is over every input state $\rho_{A}$. The algorithm
features a min-prover and a max-prover. Short quantum games were defined and
studied in \cite{GW05,G05}.

\begin{algorithm}
\label{alg:min-trace-distance-channels}The algorithm proceeds as follows:

\begin{enumerate}
\item The min-prover prepares a state $\psi_{RA}$ and sends system $A$ to the verifier.

\item The verifier picks a classical bit $i\in\left\{  0,1\right\}  $
uniformly at random, applies the channel $\mathcal{N}_{A\rightarrow B}^{i}$,
and sends system $B$ to the max-prover.

\item The max-prover prepares a system $F$ in the $|0\rangle_{F}$ state and
acts on systems $R$, $B$, and $F$ with a unitary $P_{RBF\rightarrow
TF^{\prime}}$ to produce the output systems $T$ and $F^{\prime}$, where $T$ is
a qubit system.

\item The max-prover sends system $T$ to the verifier, who then performs a
measurement on system $T$, with outcome $j \in \{0,1\}$. The verifier accepts if and only
if $i=j$.
\end{enumerate}
\end{algorithm}

For a fixed state $\psi_{RA}$ of the min-prover, it follows from
Algorithm~\ref{alg:trace-distance-states}\ that the acceptance probability is
equal to%
\begin{equation}
\frac{1}{2}\left(  1+\frac{1}{2}\left\Vert \mathcal{N}_{A\rightarrow B}%
^{0}(\rho_{A})-\mathcal{N}_{A\rightarrow B}^{1}(\rho_{A})\right\Vert
_{1}\right)  ,
\end{equation}
where $\rho_{A}=\operatorname{Tr}_{R}[\psi_{RA}]$. Since the min-prover plays
first and his goal is to minimize the acceptance probability, it follows that
the acceptance probability of Algorithm~\ref{alg:min-trace-distance-channels}
is given by%
\begin{equation}
\frac{1}{2}\left( 1+\Vert \mathcal{N}_0 - \mathcal{N}_1 \Vert_{\diamond, \text{min}}\right),
\end{equation}
where 
\begin{equation}
    \Vert \mathcal{N}_0 - \mathcal{N}_1 \Vert_{\diamond, \text{min}} \coloneqq \frac{1}{2}\inf_{\rho_{A}}\left\Vert \mathcal{N}_{A\rightarrow B}^{0}(\rho_{A})-\mathcal{N}_{A\rightarrow B}^{1}(\rho_{A})\right\Vert _{1}.
\end{equation}

Another way to estimate the minimum trace distance of channels in
\eqref{eq:min-trace-distance-channels} is to swap the roles of the max-prover
and min-prover in Algorithm~\ref{alg:min-trace-distance-channels}:

\begin{algorithm}
\label{alg:min-trace-distance-channels-swap-min-max}The algorithm proceeds as follows:

\begin{enumerate}
\item The max-prover prepares a state $\psi_{RA}$ and sends system $A$ to the verifier.

\item The verifier picks a classical bit $i\in\left\{  0,1\right\}  $
uniformly at random, applies the channel $\mathcal{N}_{A\rightarrow B}^{i}$,
and sends system $B$ to the min-prover.

\item The min-prover prepares a system $F$ in the $|0\rangle_{F}$ state and
acts on systems $R$, $B$, and $F$ with a unitary $P_{RBF\rightarrow
TF^{\prime}}$ to produce the output systems $T$ and $F^{\prime}$, where $T$ is
a qubit system.

\item The min-prover sends system $T$ to the verifier, who then performs a
measurement on system $T$, with outcome $j \in \{0,1\}$. The verifier accepts if and only
if $i=j$.
\end{enumerate}
\end{algorithm}

For a fixed state $\psi_{RA}$ of the max-prover, it follows from
\eqref{eq:min-formula-trace-distance}\ that the acceptance probability is
equal to%
\begin{equation}
\frac{1}{2}\left(  1-\frac{1}{2}\left\Vert \mathcal{N}_{A\rightarrow B}%
^{0}(\rho_{A})-\mathcal{N}_{A\rightarrow B}^{1}(\rho_{A})\right\Vert
_{1}\right)  ,
\end{equation}
where $\rho_{A}=\operatorname{Tr}_{R}[\psi_{RA}]$. Since the max-prover plays
first and his goal is to maximize the acceptance probability, it follows that
the acceptance probability of Algorithm~\ref{alg:min-trace-distance-channels}
is given by%
\begin{equation}
\frac{1}{2}\left(  1-\frac{1}{2}\inf_{\rho_{A}}\left\Vert \mathcal{N}%
_{A\rightarrow B}^{0}(\rho_{A})-\mathcal{N}_{A\rightarrow B}^{1}(\rho
_{A})\right\Vert _{1}\right)  .
\end{equation}

Although the quantities estimated by
Algorithms~\ref{alg:fid-channels-single-prover} and \ref{alg:min-trace-distance-channels} or 
\ref{alg:min-trace-distance-channels-swap-min-max} are similar (and related to
each other by standard inequalities relating trace distance and fidelity
\cite{FvG99}), the algorithms are very different in that the channel output is
available at the end of Algorithm~\ref{alg:fid-channels-single-prover},
whereas it is not at the end of
Algorithms~\ref{alg:min-trace-distance-channels} and \ref{alg:min-trace-distance-channels-swap-min-max}. This has
implications for applications in which it is helpful to have access to the
channel output, for example, when one is trying to find the fixed point of a
quantum channel.

\subsection{Generalization to multiple states, channels, and strategies}

\label{sec:multiple-TD-based}

Each of the algorithms from the previous subsections has a generalization to
multiple states, channels, and strategies. We go through them briefly here.
The main idea is that, rather than randomly picking from a set of two resources,
the verifier picks randomly from a set of multiple resources and then a prover
has to guess which one was chosen. The main difference with the binary case is
that there is not a closed-form expression for the acceptance probability in
terms of a metric like the trace distance or derived metrics, but rather the
optimization is phrased as a semi-definite program that can be solved
numerically or used in some cases to obtain analytical solutions (for
example, if there is sufficient symmetry).

Suppose that we are given an ensemble $\{p(x),\rho_{S}^{x}\}_{x\in\mathcal{X}}$ of quantum states.  The verifier picks $x$
randomly according to $p(x)$, prepares $\rho_{S}^{x}$, and the prover has to
guess which state was prepared. The  acceptance probability is given by
\begin{equation}
\label{eq:mult-state-discr}
p_g(\{p(x),\rho^x\}_{x\in \mathcal{X}}) \coloneqq 
\sup_{\left\{  \Lambda_{S}^{x}\right\}  _{x\in\mathcal{X}}}\sum_{x\in
\mathcal{X}}p(x)\operatorname{Tr}[\Lambda_{S}^{x}\rho_{S}^{x}],
\end{equation}
where the optimization is over every POVM\ $\left\{ \Lambda_{S}^{x}\right\}
_{x\in\mathcal{X}}$.  In the case that $\left\vert \mathcal{X}\right\vert =2$,
this acceptance probability has the explicit form%
\begin{equation}
\frac{1}{2}\left(  1+\left\Vert p\rho_{S}^{0}-\left(  1-p\right)  \rho_{S}%
^{1}\right\Vert _{1}\right)  .
\end{equation}

To account for multiple states, we modify Algorithm~\ref{alg:trace-distance-states} as follows: the verifier's variable $i\in\left\{  0, \dotsc, \vert \mathcal{X} \vert-1\right\} $ is randomly selected and the prover's guess $j $ is chosen from the same set. System $T$ therein is generalized to be a $\lceil \log_2 \vert \mathcal{X} \vert\rceil$-qubit system.
When $\vert \mathcal{X} \vert$ is a power of two, there is a perfect match between the number $\vert \mathcal{X} \vert$ of measurement outcomes and the dimension of system $T$. The verifier accepts if the outcome $j$ equals the state $i$ that was picked. If $\vert \mathcal{X} \vert$ is not a power of two, the following algorithm handles this case by coarse graining some of the measurement outcomes together. This is relevant because most quantum computers are qubit-based.

\begin{algorithm}
\label{alg:mult-state-disc-qubit}The~algorithm
proceeds as follows:

\begin{enumerate}
\item The verifier selects an  integer $i\in\left\{  0, \dotsc, \vert \mathcal{X} \vert - 1 \right\}  $
 at random according to $p(i)$, prepares the state $\rho_{S}^{i}$, and sends system~$S$
to the prover. 


\item The prover prepares a system $F$ composed of $\lceil \log_2\vert \mathcal{X} \vert \rceil$ qubits in the $|0\rangle_{F}$ state. The prover then acts
on systems $S$ and $F$ with a unitary $P_{SF\rightarrow TF^{\prime}}$, 
producing the output systems $F^{\prime}$ and $T$, where $T$ is a system of $\lceil \log_2\vert \mathcal{X} \vert \rceil$ qubits.

\item The prover sends system $T$ to the verifier, who then performs a computational basis measurement on system $T$, with outcome $j \in \{0,\ldots, 2^{\lceil \log_2\vert \mathcal{X} \vert \rceil}-1\}$. 

\item The verifier accepts under two conditions.
\begin{itemize}
    \item $j\leq \vert \mathcal{X} \vert-1$ and $i=j$.
    \item $j>\vert \mathcal{X} \vert-1$ and $i=0$.
\end{itemize}
\end{enumerate}
\end{algorithm}

This algorithm is a direct generalization of Algorithm~\ref{alg:trace-distance-states}. To understand its connection to \eqref{eq:mult-state-discr}, consider that,
for a fixed unitary $P_{SF\rightarrow TF^{\prime}}$, its acceptance
probability is given by%
\begin{align}
& \sum_{i\in\left\{  0,\ldots,\left\vert \mathcal{X}\right\vert -1\right\}
}p(i)\operatorname{Tr}[(|i\rangle\!\langle i|_{T}\otimes I_{F^{\prime}}%
)P(\rho_{S}^{i}\otimes|0\rangle\!\langle0|_{F})P^{\dag}]\nonumber\\
& \qquad+p(0)\sum_{j=\left\vert \mathcal{X}\right\vert }^{2^{\left\lceil
\log_{2}\left\vert \mathcal{X}\right\vert \right\rceil }}\operatorname{Tr}%
[(|j\rangle\!\langle j|_{T}\otimes I_{F^{\prime}})P(\rho_{S}^{i}\otimes
|0\rangle\!\langle0|_{F})P^{\dag}]\label{eq:mult-states-disc-unitary}\\
& =\sum_{i\in\left\{  0,\ldots,\left\vert \mathcal{X}\right\vert -1\right\}
}p(i)\operatorname{Tr}[\langle0|_{F}P^{\dag}(|i\rangle\!\langle i|_{T}\otimes
I_{F^{\prime}})P|0\rangle_{F}\rho_{S}^{i}]\nonumber\\
& \qquad+p(0)\sum_{j=\left\vert \mathcal{X}\right\vert }^{2^{\left\lceil
\log_{2}\left\vert \mathcal{X}\right\vert \right\rceil }}\operatorname{Tr}%
[\langle0|_{F}P^{\dag}(|j\rangle\!\langle j|_{T}\otimes I_{F^{\prime}%
})P|0\rangle_{F}\rho_{S}^{i}]\\
& =\sum_{i\in\left\{  0,\ldots,\left\vert \mathcal{X}\right\vert -1\right\}
}p(i)\operatorname{Tr}[\Lambda_{S}^{i}\rho_{S}^{i}],
\end{align}
where we have defined the following measurement operators:%
\begin{multline}
\Lambda_{S}^{0}\coloneqq \langle0|_{F}P^{\dag}(|0\rangle\!\langle0|_{T}\otimes
I_{F^{\prime}})P|0\rangle_{F}\\
+\sum_{j=\left\vert \mathcal{X}\right\vert }^{2^{\left\lceil \log
_{2}\left\vert \mathcal{X}\right\vert \right\rceil }}
\langle0|_{F}P^{\dag}(|j\rangle\!\langle j|_{T}\otimes I_{F^{\prime}%
})P|0\rangle_{F},
\end{multline}
and for all $i\in\left\{  1,\ldots,\left\vert \mathcal{X}\right\vert -1\right\}  $:%
\begin{equation}
\Lambda_{S}^{i}\coloneqq \langle0|_{F}P^{\dag}(|i\rangle\!\langle i|_{T}\otimes
I_{F^{\prime}})P|0\rangle_{F}.
\end{equation}
As such, we coarse grain all measurement outcomes in $\{0,\left\vert
\mathcal{X}\right\vert ,\left\vert \mathcal{X}\right\vert +1,\ldots
,2^{\left\lceil \log_{2}\left\vert \mathcal{X}\right\vert \right\rceil }\}$
into a single measurement outcome.
By the Naimark extension theorem, every measurement with $\vert \mathcal{X} \vert$ outcomes can be realized in this
way, so that maximizing the expression in \eqref{eq:mult-states-disc-unitary}
over every unitary $P$ gives a value equal to that in \eqref{eq:mult-state-discr}.

On the one hand, if $|\mathcal{X}|$ is a power of two, then it follows that $|\mathcal{X}| = 2^{\left\lceil \log_{2}\left\vert \mathcal{X}\right\vert \right\rceil}$ and the outcome $j > |\mathcal{X}| - 1$ never occurs. On the other hand, if $|\mathcal{X}|$ is not a power of two, then  $|\mathcal{X}| < 2^{\left\lceil \log_{2}\left\vert \mathcal{X}\right\vert \right\rceil}$ and the outcome $j > |\mathcal{X}| - 1$ does occur.


Now suppose that we are given an ensemble $\{p(x),\mathcal{N}_{A\rightarrow B}%
^{x}\}_{x\in\mathcal{X}}$\ of quantum channels. Then a similar modification of
Algorithm~\ref{alg:diamond-distance-channels} has acceptance probability%
\begin{equation}
\sup_{\psi_{RA},\left\{  \Lambda_{RB}^{x}\right\}  _{x\in\mathcal{X}}}%
\sum_{x\in\mathcal{X}}p(x)\operatorname{Tr}[\Lambda_{RB}^{x}\mathcal{N}%
_{A\rightarrow B}^{x}(\psi_{RA})],
\end{equation}
where the optimization is over every state $\psi_{RA}$ and POVM\ $\left\{
\Lambda_{RB}^{x}\right\}  _{x\in\mathcal{X}}$. In the case that $\left\vert
\mathcal{X}\right\vert =2$, this acceptance probability has the explicit form%
\begin{equation}
\frac{1}{2}\left(  1+\left\Vert p\mathcal{N}_{A\rightarrow B}^{0}-\left(
1-p\right)  \mathcal{N}_{A\rightarrow B}^{1}\right\Vert _{\diamond}\right)  .
\end{equation}

Suppose we are given an ensemble $\{p(x),\mathcal{N}^{x,(n)}\}_{x\in
\mathcal{X}}$ of $n$-turn quantum strategies. A similar modification of
Algorithm~\ref{alg:strategy-distance-strategies} has acceptance probability%
\begin{equation}
\sup_{\substack{\mathcal{S}^{(n-1)},\\\left\{  \Lambda_{R^{n}B^{n}}%
^{x}\right\}  _{x\in\mathcal{X}}}}\sum_{x\in\mathcal{X}}p(x)\operatorname{Tr}%
[\Lambda_{R^{n}B^{n}}^{x}(\mathcal{N}^{x,(n)}\circ\mathcal{S}^{(n-1)})],
\end{equation}
where the optimization is over every $(n-1)$-turn pure co-strategy
$\mathcal{S}^{(n-1)}$ and POVM\ $\left\{  \Lambda_{RB}^{x}\right\}
_{x\in\mathcal{X}}$ (recall \eqref{eq:state-after-strategy-co-strategy}\ in
this context). In the case that $\left\vert \mathcal{X}\right\vert =2$, this
acceptance probability has the explicit form%
\begin{equation}
\frac{1}{2}\left(  1+\left\Vert p\mathcal{N}^{0,(n)}-\left(  1-p\right)
\mathcal{N}^{1,(n)}\right\Vert _{\diamond n}\right)  ,
\end{equation}
where this is the strategy norm.

Finally, we can generalize Algorithms~\ref{alg:min-trace-distance-channels}%
\ and \ref{alg:min-trace-distance-channels-swap-min-max}, with the acceptance
probabilities respectively given by%
\begin{align}
&  \inf_{\rho_{A}}\sup_{\left\{  \Lambda_{B}^{x}\right\}  _{x\in\mathcal{X}}%
}\sum_{x\in\mathcal{X}}p(x)\operatorname{Tr}[\Lambda_{B}^{x}\mathcal{N}%
_{A\rightarrow B}^{x}(\rho_{A})],\\
&  \sup_{\rho_{A}}\inf_{\left\{  \Lambda_{B}^{x}\right\}  _{x\in\mathcal{X}}%
}\sum_{x\in\mathcal{X}}p(x)\operatorname{Tr}[\Lambda_{B}^{x}\mathcal{N}%
_{A\rightarrow B}^{x}(\rho_{A})].
\end{align}
In the case that $\left\vert \mathcal{X}\right\vert =2$, these acceptance
probabilities become%
\begin{align}
&  \frac{1}{2}\left(  1+\inf_{\rho_{A}}\left\Vert p\mathcal{N}_{A\rightarrow
B}^{0}(\rho_{A})-\left(  1-p\right)  \mathcal{N}_{A\rightarrow B}^{1}(\rho
_{A})\right\Vert _{1}\right)  ,\\
&  \frac{1}{2}\left(  1-\inf_{\rho_{A}}\left\Vert p\mathcal{N}_{A\rightarrow
B}^{0}(\rho_{A})-\left(  1-p\right)  \mathcal{N}_{A\rightarrow B}^{1}(\rho
_{A})\right\Vert _{1}\right)  .
\end{align}

\section{Performance evaluation of algorithms using a noiseless and noisy quantum simulator} 

\label{sec:num-sims}

In this section, we present results obtained from numerically simulating Algorithms~\ref{alg:fid-states}--\ref{alg:mixed-state-FC-meas-min} and Algorithm~\ref{alg:trace-distance-states} on a noiseless quantum simulator and Algorithms \ref{alg:fid-channels},  \ref{alg:diamond-distance-channels}, and \ref{alg:mult-state-disc-qubit} on both a noiseless and noisy quantum simulator. In the first subsection, we introduce and discuss the circuit ansatz employed in these numerical experiments. In the next subsection, we discuss the form of the states and channels used for the numerical simulations. In the following subsections, we present the details of our numerical simulations of Algorithms~\ref{alg:fid-states}--\ref{alg:mixed-state-FC-meas-min} for fidelity of states, Algorithm~\ref{alg:fid-channels} for the fidelity of channels, Algorithm~\ref{alg:trace-distance-states} for trace distance of states, Algorithm~\ref{alg:diamond-distance-channels} for diamond distance of channels, and Algorithm~\ref{alg:mult-state-disc-qubit} for multiple state discrimination.

In the simulations below, we use a maximum number of iterations to be the stopping condition. We noted that some algorithms - in particular, ones with multiple provers - were more prone to get stuck in local minima and optimization loops. We found that, in these scenarios, using convergence as the stopping condition could lead to an unbounded number of iterations. In these cases, we found that using a maximum number of iterations was sufficient and effective.

All the program code for Algorithms~\ref{alg:fid-states}, \ref{alg:mixed-state-swap-test}, \ref{alg:mixed-state-Bell-tests}, \ref{alg:mixed-state-FC-meas-min}, \ref{alg:fid-channels}, \ref{alg:trace-distance-states}, \ref{alg:diamond-distance-channels}, \ref{alg:mult-state-disc-qubit}, and corresponding SDPs can be found as arXiv ancillary files with the arXiv posting of this paper.

\subsection{Ansatz}

\label{sec:ansatz}

To estimate the relevant quantities in this work, we employ the hardware-efficient ansatz (HEA) \cite{KMTTBCG17}. The HEA is a problem-agnostic ansatz that depends on the architecture and the connectivity of the given hardware. In this work, we consider a fixed structure of the HEA.
Let $X$, $Y$, and $Z$ denote the Pauli matrices. We define one layer of the HEA to consist of the single-qubit rotations $e^{-\im \theta/2 Y}e^{-\im \delta/2 X}$,  each of which acts on a single qubit and is parameterized by $\theta$ and $\delta$, followed by CNOTs between neighboring qubits. A CNOT between the control qubit $k$ and the target qubit $\ell$ is given by
\begin{multline}
e^{-\im \pi/2 (\vert 1 \rangle \!\langle 1\vert_k \otimes (X_\ell - I_\ell))} = \\
\vert 0\rangle \!\langle 0 |_k \otimes I_\ell +\vert 1\rangle \!\langle 1 |_k \otimes X_\ell.    
\end{multline}

For our numerical experiments, we consider a sufficiently large number of layers of the HEA. In principle, both the circuit structure and the number of layers of the HEA can be made random and this randomness can lead to better performance of variational algorithms \cite{BCVCC21}. We leave the study of such ansatze for future work. 

The HEA is used both to create the states and channels, as well as to create a parameterized unitary that replaces the provers. In the former two cases, the rotation angles are fixed, but in the prover scenario, the angles are parameters that are optimized.

\subsection{Test states and channels}

\label{subsec:testStatesChannels}

To study the performance of our algorithms, we randomly select states and channels as follows. For $n$-qubit states, we apply $m$ layers of the HEA with randomly selected angles for rotation around the $x$- and $y$-axes on $n+k$ qubits initialized to the state $\vert 0 \rangle\!\langle 0 |$. This procedure prepares a pure state on $n+k$ qubits and hence, a mixed state on $n$ qubits of rank $\leq 2^k$.

To realize an $n$-qubit channel $\mathcal{N}_{A \rightarrow B}$, we generate a unitary $U_{AE' \rightarrow BE}$ on $n+k$ qubits such that
\begin{multline}
    \mathcal{N}_{A \rightarrow B}(\omega_A) \coloneqq \\
    \operatorname{Tr}_E \!\left[ U_{AE' \rightarrow BE} (\omega_A \otimes \vert 0 \rangle \! \langle 0 \vert_{E'}) (U_{AE' \rightarrow BE})^\dagger \right],
\end{multline}
where systems $E'$ and $E$ each consist of $k$ qubits. Due the Stinespring dilation theorem \cite{S55}, this is a general approach by which arbitrary channels can be realized.

For our experiments, we set $U$ to consist of $m$ layers of the HEA itself, with randomly selected angles for rotation around the $x$- and $y$-axes on $n+1$ qubits. Tracing out one of the qubits gives a channel on $n$ qubits, as required.

Several algorithms in our paper (see \eqref{eq:alg-pure-bell-c-Ui}, \eqref{eq:controlled-U-fid-alg-bell-overlap}, \eqref{eq:controlled-U-fid-channels}) depend on having access to unitaries of the form 
\begin{equation}
\sum_{i\in\left\{  0,1\right\}  }|i\rangle\!\langle i|_{T}\otimes U_{S}^{i} = \vert 0 \rangle\!\langle 0 \vert \otimes U^0_S + \vert 1 \rangle\!\langle 1 \vert \otimes U^1_S.
\end{equation}
These can be split into the sequential application of the following two controlled unitaries:
\begin{align}
    \label{eq:controlled-un}
    |0\rangle\!\langle 0| \otimes I & + |1\rangle\!\langle 1| \otimes U^1_S, \nonumber \\
    |1\rangle\!\langle 1| \otimes I & + |0\rangle\!\langle 0| \otimes U^0_S,
\end{align}
of which our algorithms make use.

\subsection{Fidelity of states}

In this section, we discuss the performance of Algorithms \ref{alg:fid-states}--\ref{alg:mixed-state-FC-meas-min} in the noiseless scenario to estimate the fidelity between two three-qubit mixed states. Algorithms~\ref{alg:fid-states}--\ref{alg:mixed-state-FC-meas-min} require different numbers of qubits for estimating the fidelity between $\rho$ and $\sigma$. In particular, for this case, Algorithm~\ref{alg:fid-states} requires eight qubits, along with access to controlled unitaries, as defined in \eqref{eq:controlled-un}. Algorithms~\ref{alg:mixed-state-swap-test},
\ref{alg:mixed-state-Bell-tests}, and~\ref{alg:mixed-state-FC-meas-min} require 13, 10, and 8 qubits, respectively. We recall that Algorithms~\ref{alg:fid-states}--\ref{alg:mixed-state-Bell-tests} require purifications of both $\rho$ and $\sigma$, while Algorithm~\ref{alg:mixed-state-FC-meas-min} relies only on access to $\rho$ and~$\sigma$ directly. Moreover, Algorithms~\ref{alg:fid-states} and \ref{alg:mixed-state-swap-test} require measurements on two qubits, and Algorithm~\ref{alg:mixed-state-Bell-tests} requires Bell measurements on ten qubits. Finally, Algorithm~\ref{alg:mixed-state-FC-meas-min} requires two single-qubit measurements. 

We now summarize the HEA employed. For Algorithm~\ref{alg:fid-states}, the prover unitary is created using five layers of the HEA, which acts on four qubits. Similarly, in Algorithm~\ref{alg:mixed-state-swap-test}, we employ eight layers of the HEA that acts on six qubits. In Algorithm~\ref{alg:mixed-state-Bell-tests}, the ansatz acts on two qubits, and we consider four layers of it. In Algorithm~\ref{alg:mixed-state-FC-meas-min}, the ansatz acts on four qubits, and we apply eight layers of it. For our implementations, we picked these circuit depths so that the cost function is minimized. A more general framework allows for the ansatz structure to be unfixed and instead variable, but we leave the detailed study of this, for our algorithms, to future work \cite{BCVCC21}. 

\begin{figure}[t]
   \centering
  \includegraphics[width=\columnwidth]{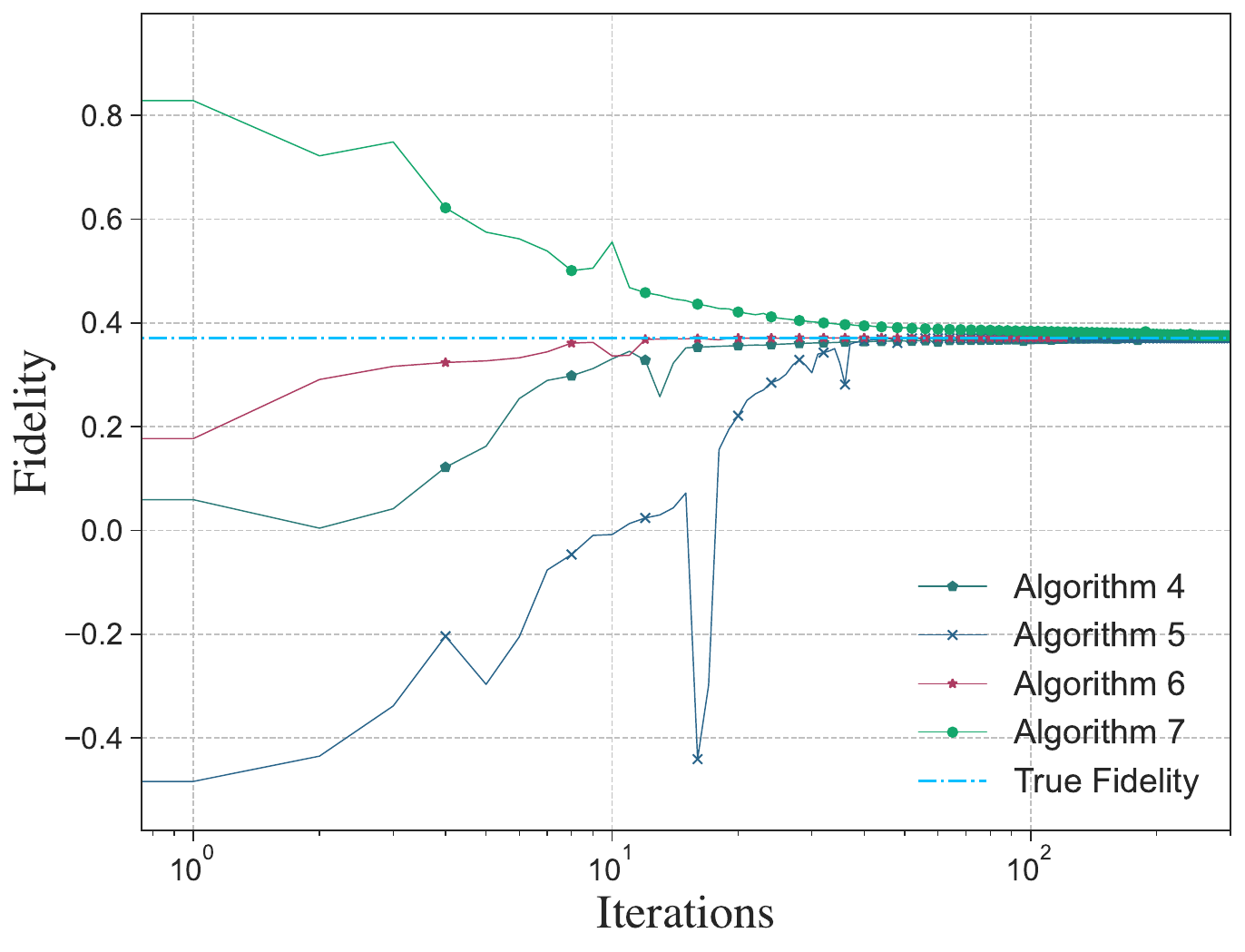}
  \caption{Estimation of the fidelity between quantum states versus the number of iterations. We implement Algorithms~\ref{alg:fid-states}--\ref{alg:mixed-state-FC-meas-min} on a noiseless simulator to estimate the fidelity between two three-qubit mixed states, each of rank $\leq 4$. For each variational algorithm, we employ the HEA, as defined in Section~\ref{sec:ansatz}. In particular, we start with a random parameter vector $\vec{\theta}$ and then update it according to a gradient-based optimization procedure. The dashed-dotted curve represents the true fidelity between two randomly chosen quantum states. In each case, the optimization procedure converges to the true fidelity with high accuracy. Algorithms~\ref{alg:fid-states}--\ref{alg:mixed-state-FC-meas-min} achieve an absolute error in fidelity estimation of order $10^{-5}, 10^{-4}, 10^{-9}$, and $10^{-3}$, respectively. }
  \label{fig:fid-algs}
\end{figure}

We begin the training with a random set of variational parameters. We evaluate the cost using a state vector simulator (noiseless simulator) \cite{Qiskit}. We then employ the gradient-descent algorithm to obtain a new set of parameters. We note that in general, the true fidelity between states $\rho$ and $\sigma$ is not known. Thus the stopping criterion for these algorithms is a maximum number of iterations. For our numerical experiments, we set the total number of iterations to be 300. For each algorithm, we run ten instances of the algorithm and pick the best run for generating Figure~\ref{fig:fid-algs}.  

In Figure~\ref{fig:fid-algs}, we plot the results of the numerical simulations. The dashed-dotted line represents the true fidelity between two random three-qubit quantum states $\rho$ and~$\sigma$, as described above. Each algorithm converges to the true fidelity with high accuracy within a finite number of iterations. As discussed above, for each algorithm, the HEA is of a different size. Thus, it is not straightforward to compare these different algorithms. In terms of the convergence rate, we find that Algorithm~\ref{alg:mixed-state-Bell-tests} converges to the true fidelity faster than all other algorithms. Algorithms~\ref{alg:fid-states}--\ref{alg:mixed-state-FC-meas-min} achieve an absolute error in fidelity estimation of order $10^{-5}$, $10^{-4}$, $10^{-9}$, and $10^{-3}$, respectively. 

\subsection{Trace distance of states}

Using Algorithm~\ref{alg:trace-distance-states}, we estimate the normalized trace distance $\frac{1}{2} \left \Vert \rho - \sigma \right\Vert_1$ between two three-qubit states $\rho$ and~$\sigma $, each having rank $\leq 4$, as defined above in Section~\ref{subsec:testStatesChannels}. For our numerical experiments, we use a noiseless simulator. Algorithm~\ref{alg:trace-distance-states} requires eight qubits in total and two single-qubit measurements. We employ ten layers of the HEA, which acts on four qubits. Similar to the fidelity-estimation algorithms detailed above, we begin with a random set of variational parameters and update them using the gradient-descent algorithm. 

As the true normalized trace distance between $\rho$ and~$\sigma$ is assumed to be unknown, we use a stopping criterion as the number of iterations, which we take to be 300~iterations. For Algorithm~\ref{alg:trace-distance-states}, we run ten instances of it and pick the best run for generating Figure~\ref{fig:td-algs}.

In Figure~\ref{fig:td-algs}, we plot the results of Algorithm~\ref{alg:trace-distance-states}. The dashed-dotted line represents the true normalized trace distance between two random three-qubit quantum states $\rho$ and $\sigma$, as described above. The absolute error in trace-distance estimation is of order $10^{-4}$.

\begin{figure}[t]
   \centering
  \includegraphics[width=\columnwidth]{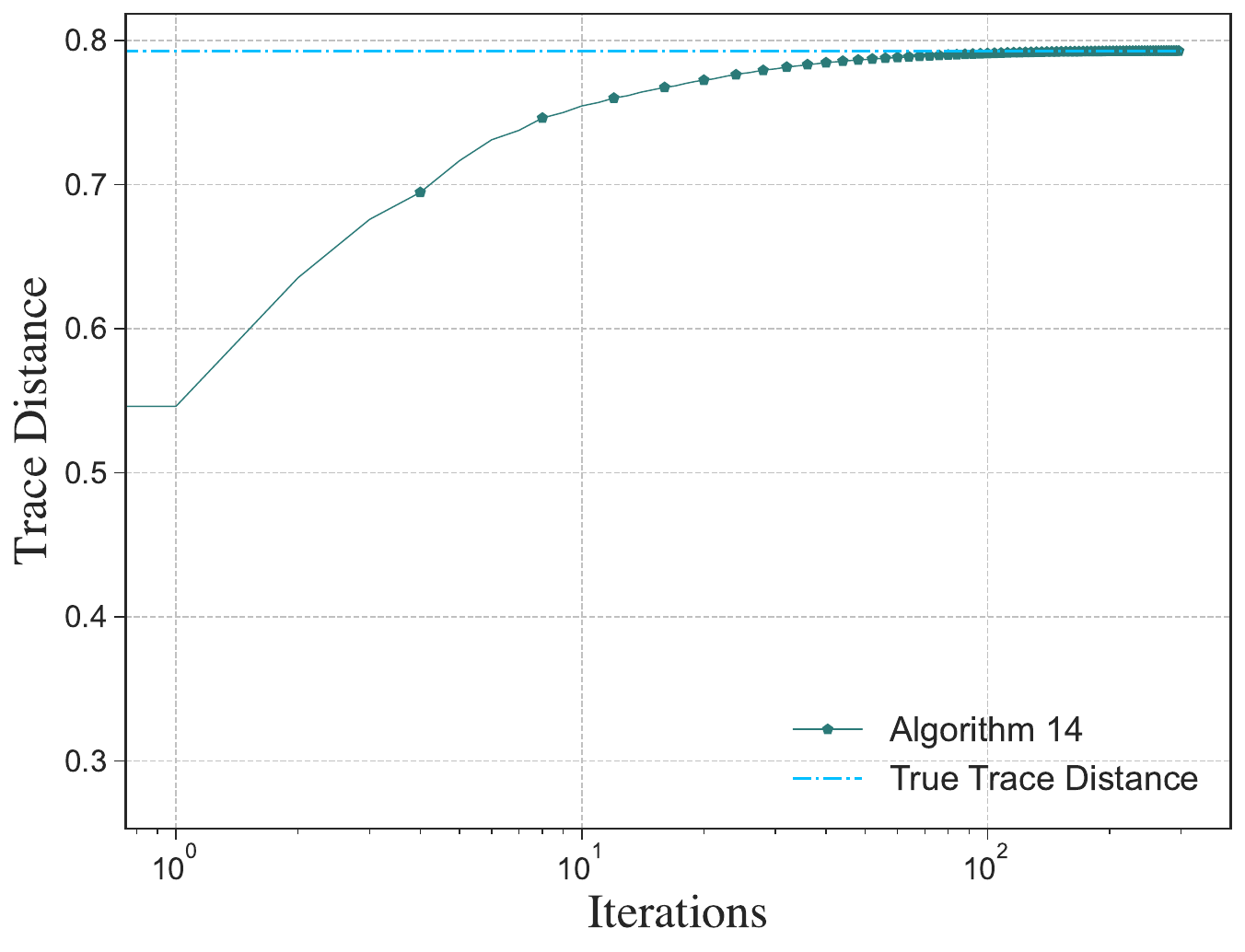}
  \caption{Estimation of the normalized trace distance between quantum states versus the number of iterations. We implement Algorithm~\ref{alg:trace-distance-states} on a noiseless simulator to estimate the normalized trace distance between three-qubit mixed states, each of rank four. Algorithm~\ref{alg:trace-distance-states} achieves an absolute error in trace distance estimation of order $10^{-4}$.}
  \label{fig:td-algs}
\end{figure}

\subsection{Fidelity of channels}

\label{sec:numerics-fid-chs}

In this section, we discuss the performance of Algorithm~\ref{alg:fid-channels} in both the noiseless and noisy scenarios. The channels in question are realized by using parameterized unitaries and tracing out ancilla qubits, as discussed in  Section~\ref{subsec:testStatesChannels}. The algorithm employs a min-max optimization and thus requires two parameterized unitaries representing the min- and max-provers, respectively. The controlled unitaries consist of one layer of the HEA, with each consisting of random rotations about the $x$-axis, on two qubits, thereby realizing the $\mathcal{N}^i_{A \rightarrow B}$ channels acting on one qubit, for $i \in \{0,1\}$. 

We now summarize the HEA employed in generating the min- and max-provers. The min-prover unitary is generated using two layers of the HEA, which acts on two qubits. The max-prover unitary is generated using two layers of the HEA, which acts on three qubits. The rotation angles for both provers around the $x$- and $y$-axes are chosen at random. The particular choices of the number of layers are made so that the cost function is minimized.

We begin the training phase with a random set of variational parameters for both parameterized unitaries. For the noiseless simulation, we evaluate the cost using a state vector simulator (noiseless simulator) \cite{Qiskit}. For the noisy simulation, we use the QASM-simulator with the noise model from IBM-Jakarta. Since the number of parameters is significantly higher than the previous algorithms, to speed up the convergence, we employ both the simultaneous perturbation stochastic approximation (SPSA) method \cite{Spall_SPSA} and the gradient-descent method to obtain a new set of parameters.

The optimization is carried out in a zig-zag fashion, explained as follows. The minimizing optimizer implements the SPSA algorithm and is allowed to run until convergence occurs. Then, the maximizing optimizer, implementing the gradient descent algorithm, runs for one iteration. We note that in general, the true fidelity between the channels $\mathcal{N}^0$ and $\mathcal{N}^1$ is not known. Thus, the stopping criterion for these algorithms is a maximum number of iterations. For our numerical experiments, we set the total number of iterations to be 6000, mostly used in the minimizing optimizer. The results of the numerical simulations are presented in Figure~\ref{fig:fid-channels-alg}.

Note that the graph presented in Figure~\ref{fig:fid-channels-alg} shows that the convergence is highly non-monotonic, unlike the convergence behavior presented in previous graphs. Each iteration consists of a decrease in the function value, followed by a single increasing iteration. This is clearly indicative of the min-max optimization nature of the algorithm. Furthermore, unlike other algorithms, the optimization value in this algorithm can overshoot the true solution, due to the min-max nature of the optimization. However, the noiseless plot indicates that, once it overshoots the solution, it oscillates with decreasing amplitude and converges. 

The noisy optimization converges as well, but it does not converge to the known value of the root fidelity of the two channels. However, the parameters found after convergence exhibit a noise resilience, as put forward in \cite{Sharma_2020}; i.e., using the parameters obtained from the noisy optimization in a noiseless simulator gives a value much closer to the true value, as indicated by the solid orange line in Figure~\ref{fig:fid-channels-alg}. 

\begin{figure}[t]
  \centering
  \includegraphics[width=\columnwidth]{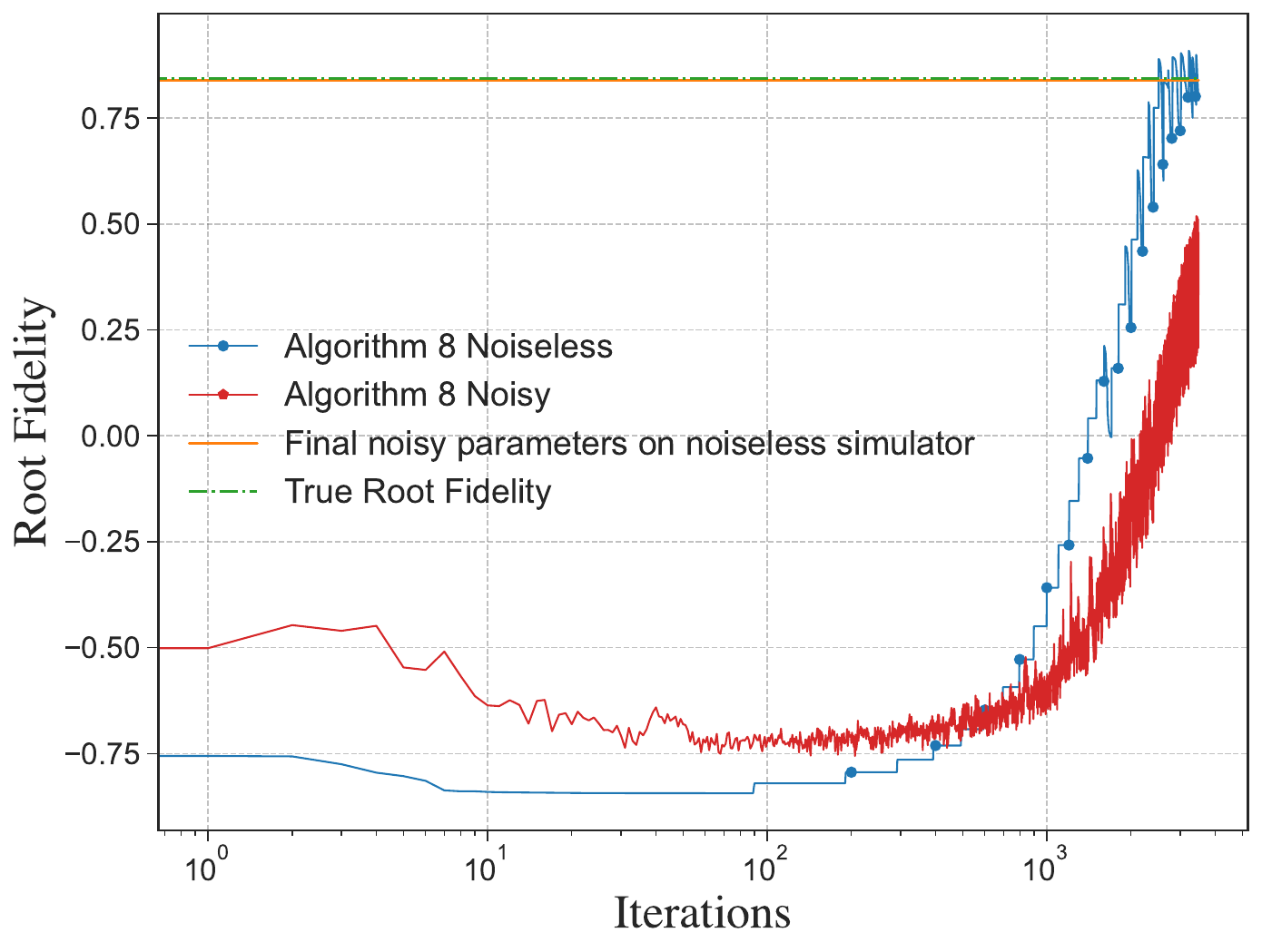}
  \caption{Estimation of the normalized fidelity between quantum channels versus the number of iterations. We implement Algorithm~\ref{alg:fid-channels} to estimate the normalized fidelity between two-qubit channels. The noiseless simulation achieves an absolute error in fidelity estimation of order $10^{-4}$. The parameters obtained from the noisy simulation, with the noise model from IBM-Jakarta, achieve an absolute error of $10^{-2}$ on a noiseless simulator.}
  \label{fig:fid-channels-alg}
\end{figure}

\subsection{Diamond distance of channels}

In this section, we discuss the performance of Algorithm~\ref{alg:diamond-distance-channels} in the noiseless and noisy scenarios. Algorithm~\ref{alg:diamond-distance-channels} requires eight qubits. Similar to the previous section, the channels in question are realized using the procedure from Section~\ref{subsec:testStatesChannels}. The algorithm utilizes a max-max optimization and thus requires two parameterized unitaries representing the two max-provers. Each unitary $U^i_{AE'\rightarrow BE}$, for $i\in\{0,1\}$, consists of one layer of the HEA with random rotations about the $x$- and $y$-axes, on two qubits, each thereby realizing the one-qubit channel $\mathcal{N}^i_{A \rightarrow B}$.

We now summarize the HEA employed in generating the two provers. The first prover, called the state-prover because its goal is to realize an optimal distinguishing state, is generated using two layers of the HEA, which acts on two qubits. The second prover, called the max-prover, is generated using two layers of the HEA, which acts on three qubits. The rotation angles for both provers around the $x$- and $y$-axes are chosen at random. The particular choices of the number of layers are made so that the cost function is minimized.  

We begin the training phase with a random set of variational parameters for both parameterized unitaries. In the noiseless simulation, we evaluate the cost using a state vector simulator (noiseless simulator). In the noisy setup, we use the QASM-simulator with the noise model from IBM-Jakarta. Similar to the previous section, we employ the SPSA optimization technique.

The optimization is carried out in two parts---the first part uses the COBYLA optimizer \cite{Powell1994, 2020SciPy-NMeth} (non-gradient based), and the second part uses the SPSA optimizer. In both stages, the optimization is carried out in a zig-zag fashion, explained as follows. The first stage allows for moving quickly into the neighbourhood of the actual solution, but then slows down dramatically. Once we approach the solution, we switch to a gradient-based method that converges to the solution more quickly. In both stages, we allow the state-prover and the max-prover to be optimized for a fixed number of iterations in a zig-zag manner. This is because, in general, the true diamond distance between channels $\mathcal{N}^0$ and $\mathcal{N}^1$ is not known. Thus the stopping criterion for these algorithms is a maximum number of iterations. For our numerical experiments, we set the total number of iterations to be 1600. The results of the numerical simulations are presented in Figure~\ref{fig:diamond-distance-channels-alg}.

Note that the noiseless graph presented in Figure~\ref{fig:diamond-distance-channels-alg} shows that the convergence is highly monotonic, unlike the fidelity of channels (see Figure~\ref{fig:fid-channels-alg}), because the optimization is a max-max one, as opposed to the min-max nature of Algorithm~\ref{alg:fid-channels}. The quick convergence, indicated by the lower number of iterations, is a consequence of this difference.

The noisy simulation converges as well, and similar to the previous section, the parameters exhibit a noise resilience. Once the COBYLA stage of the optimization is completed, the SPSA optimization is more noisy, due to the perturbative nature of the algorithm. Note that the COBYLA optimizer operates in batches of $30$, giving an impression of smoothness.

\begin{figure}[t]
  \centering
  \includegraphics[width=\columnwidth]{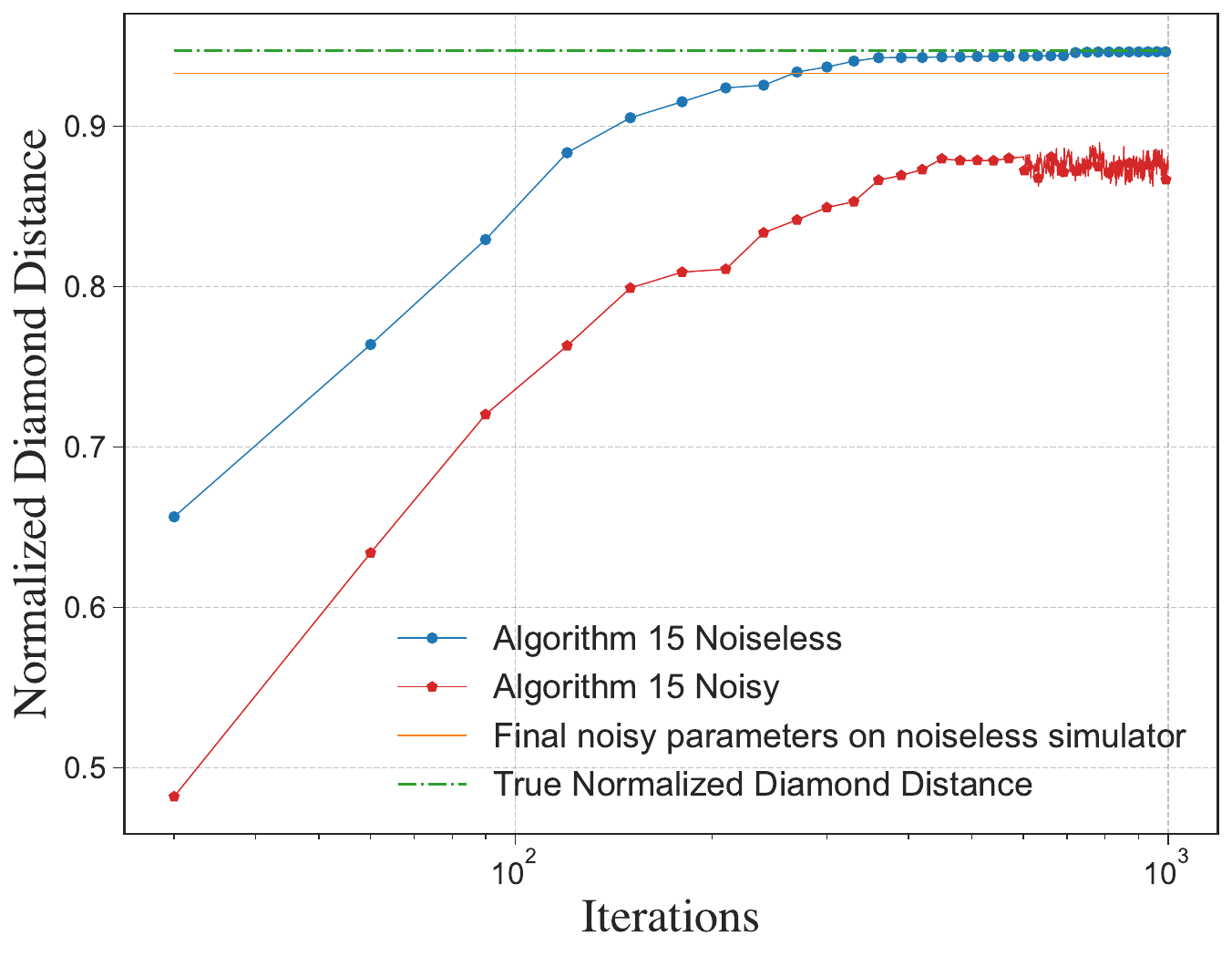}
  \caption{Estimation of the normalized diamond distance between quantum channels versus the number of iterations. We implement Algorithm~\ref{alg:diamond-distance-channels} to estimate the normalized diamond distance between one-qubit channels. Algorithm~\ref{alg:diamond-distance-channels} achieves an absolute error in diamond distance estimation of order~$10^{-4}$. The parameters obtained from the noisy simulation, with the noise model from IBM-Jakarta, achieve an absolute error of $10^{-2}$ on a noiseless simulator.}
  \label{fig:diamond-distance-channels-alg}
\end{figure}

\subsection{Multiple state discrimination}

In this section, we discuss the performance of Algorithm~\ref{alg:mult-state-disc-qubit} in the noisy and noiseless scenarios. We consider a specific scenario of distinguishing three one-qubit mixed states. Recall from Section~\ref{subsec:testStatesChannels} that the one-qubit states are generated by using two layers of the HEA on two qubits. We execute this on a qubit system, and hence we use  Algorithm~\ref{alg:mult-state-disc-qubit}. The algorithm requires twelve qubits in total and three two-qubit measurements. The measurement is realized using a parameterized unitary and ancilla qubits. By Naimark's extension theorem \cite{N40}, an arbitrary POVM can be realized using this procedure, so that there is no loss in expressiveness. The parameterized unitary required employs two layers of the HEA, which acts on three qubits. 

To speed up convergence, we use the SPSA algorithm for the optimization. As the true value of the optimal acceptance probability between the three states is assumed to be unknown, we set the stopping criterion to be a maximum number of iterations, which we take to be 250~iterations. 

In Figure~\ref{fig:mult-state-disc-alg}, we plot the results of simulating Algorithm~\ref{alg:mult-state-disc-qubit}. The dashed-dotted line represents the optimal acceptance probability of the three states, calculated using the semi-definite program corresponding to \eqref{eq:mult-state-discr}. The noiseless simulation converges to the known optimal acceptance probability. The noisy optimization converges as well, but it does not converge to the known optimal acceptance probability. However, similar to the previous sections, the parameters exhibit noise resilience, as indicated by the solid orange line in Figure~\ref{fig:mult-state-disc-alg}. 

\begin{figure}[t]
  \centering
  \includegraphics[width=\columnwidth]{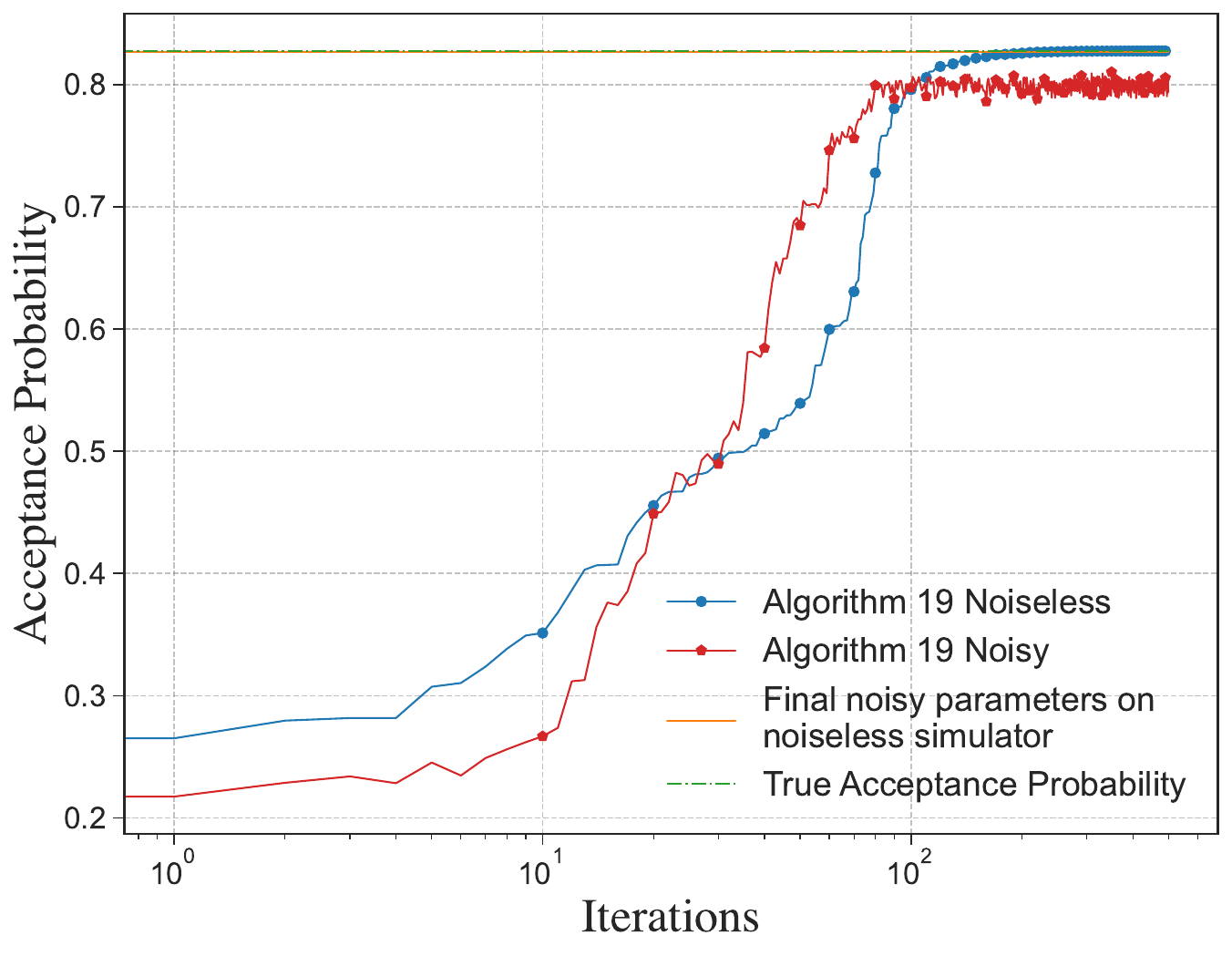}
  \caption{Estimation of the optimal acceptance probability for Algorithm~\ref{alg:mult-state-disc-qubit}. The noiseless simulation achieves an absolute error of order $10^{-4}$. The parameters obtained from the noisy simulation, with the noise model from IBM-Jakarta, achieve an absolute error of $10^{-3}$ on a noiseless simulator.} 
  \label{fig:mult-state-disc-alg}
\end{figure}

\section{Estimating distance measures as complexity classes}
\label{sec:compFidAlgs}

We now turn our attention to the intersection of our algorithms with quantum computational complexity theory. In this section, we prove that several basic quantum complexity classes can be reframed as distance and fidelity estimation problems. That is, we show that various distance and fidelity estimation problems are complete for various quantum complexity classes. Refs.~\cite{W09,VW15} provide reviews of basic concepts in quantum computational complexity theory for interested readers.

In particular, here we summarize existing results linking estimation problems to complexity classes, and furthermore, we prove that five new distance estimation algorithms that are complete for some complexity classes of interest. First, we prove that promise versions of the following estimation problems are BQP-complete:
\begin{enumerate}
    \item estimating the fidelity between two pure states,
    \item estimating the fidelity between a pure state and a mixed state,
    \item estimating the Hilbert--Schmidt distance of two arbitrary states.
\end{enumerate}
Fourth, we prove that the promise problem version of estimating the fidelity between a pure state and a channel with arbitrary input is QMA-complete. Finally, we show that the promise problem version of estimating the fidelity between a pure state and a channel with a separable input state is QMA(2)-complete. In Figure~\ref{tab:compFidAlgs}, we summarize the various quantum complexity classes and the representative fidelity and distance estimation algorithms. 

\renewcommand{\arraystretch}{1}
\begin{figure*}
\begin{tikzpicture}
[
    level 1/.style = {rect, level distance=2cm},
    level 2/.style = {rect, sibling distance = 6cm, level distance=2cm },
    every node/.style={rectangle,draw,minimum width=3cm}
]
    \node {\begin{tabular}{c} QIP-Complete \\ \rule{0pt}{5pt}\cite{KW00,W02QIP,RW05} \\ \hline \rule{0pt}{12pt}$\max\limits_{\rho, \sigma}F(\mathcal{N}(\rho), \mathcal{M}(\sigma))$ \end{tabular}}
    child {node (qip2) {\begin{tabular}{c} QIP(2)-Complete \\ \rule{0pt}{5pt}\cite{W02QIP,HMW13} \\ \hline \rule{0pt}{10pt}$\max\limits_{\sigma }F( \rho,\mathcal{N}(\sigma))$ \end{tabular}}
    child {node (qszk) {\begin{tabular}{c} QSZK-Complete \\ \rule{0pt}{5pt}\cite{W02,W09zkqa} \\ \hline \rule{0pt}{10pt}$F(\rho, \sigma)$ \end{tabular}}}
    child {node (qma) {\begin{tabular}{c} \textbf{QMA-Complete} \\ \hline \rule{0pt}{10pt}$\max\limits_\sigma F(\psi,\mathcal{N}(\sigma))$ \end{tabular}}}};
    
\node at (5.5,-2.1) (qma2) {\begin{tabular}{c} \textbf{QMA(2)-Complete} \\ \hline \rule{0pt}{10pt}$\underset{\sigma \in \operatorname{SEP}}{\max} F(\psi,\mathcal{N}(\sigma))$ \end{tabular}};
\draw (qma2) -- (qma);
\node at (0, -6.15) (bqp) {\begin{tabular}{c} \textbf{BQP-Complete} \\ \hline \rule{0pt}{10pt}$F(\psi, \phi)$ \\ \rule{0pt}{8pt}$F(\psi,\rho)$ \\ \rule{0pt}{8pt} $\left\Vert \rho - \sigma \right\Vert_2$ \end{tabular}};
\draw (bqp) -- (qszk);
\draw (bqp) -- (qma);
\end{tikzpicture}
\caption{List of distance estimation problems and the corresponding quantum complexity class. Entries in bold are the results of our paper. In this diagram, $\psi$ and $\phi$ are pure states, $\rho$ and $\sigma$ are mixed states, and $\mathcal{N}$ and $\mathcal{M}$ are channels. Note that $\rho$ and $\sigma$ may be of different dimensions, depending on the context. The cells are organized such that if a cell is connected to a cell above it, the complexity class for the lower cell is a subset of that for the the higher cell. For example, QMA is a subset of both QIP(2) and QMA(2).}
\label{tab:compFidAlgs}
\end{figure*}
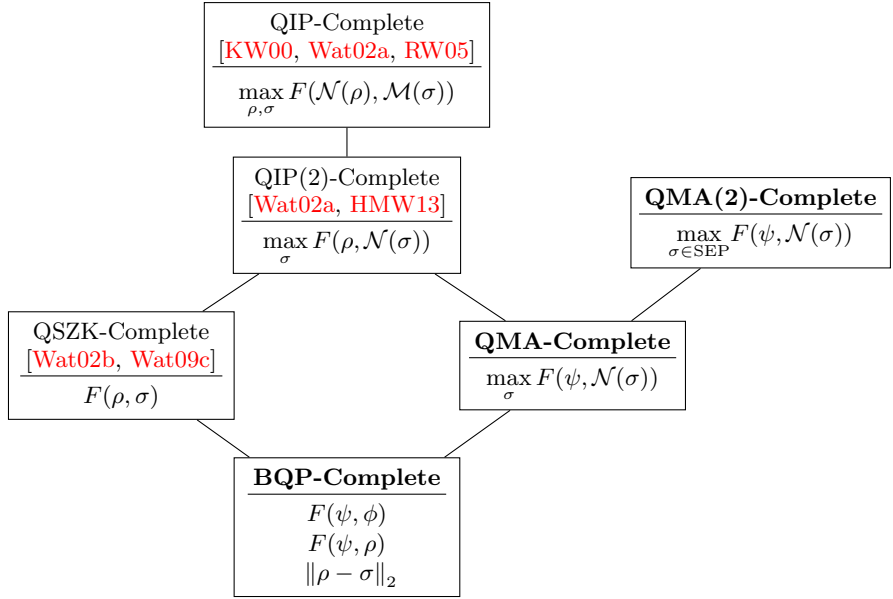

\subsection{BQP-complete problems}

First, we prove that promise versions of the problems of evaluating the fidelity between two pure states, evaluating the fidelity between a mixed state and a pure state, and evaluating the Hilbert--Schmidt distance of two arbitrary states are BQP-complete. Intuitively, this means that these problems can be solved efficiently on a quantum computer, and these problems furthermore capture the full power of polynomial-time quantum computation (in the sense that the ability to solve these problems implies the ability to solve an arbitrary BQP problem).

Here, we reproduce the definition of BQP for convenience. Note that our definition given here differs somewhat from the definition in \cite{W09}, in that we restrict the circuits considered to be unitary circuits; it is known that the two different definitions are equivalent, in the sense that the computational power of BQP does not change. Let $A = (A_{\text{yes}}, A_{\text{no}})$ be a promise problem and let $a, b : \mathbb{N} \rightarrow [0, 1]$ and $p$ be polynomial functions. Then $A \in \operatorname{BQP}(a, b)$ if there exists a polynomial-time generated family $Q  = \{Q_n : n \in \mathbb{N}\}$ of unitary circuits, where each circuit $Q_n$ takes $n + p(n)$ input qubits and produces one decision qubit $D$ and $n + p(n) - 1$ garbage qubits $G$, with the following properties (in what follows, we abbreviate each $Q_n$ as $Q_{SA \to DG}$, thereby suppressing the dependence on the input length $n = |x|$ and explicitly indicating the systems involved at the input and output of the unitary): 

\begin{enumerate}
    \item Completeness: For all $x \in A_{\text{yes}}$,  
        \begin{multline}
            \label{eq:BQP-completeness}
             \Pr[Q \text{ accepts }  x]  
             \\ = \left \| (\bra{1}_D \otimes I_G) Q_{SA \to DG}(\ket{x}_S \otimes \ket{0}_A)\right\|_2^2\\
             \geq a(\left\vert x \right\vert).
        \end{multline}
    
    \item Soundness: For all $x \in A_{\text{no}}$,  
        \begin{equation}
            \label{eq:BQP-soundness}
            \Pr[Q \text{ accepts } x ] \leq b(\left\vert x \right\vert),
        \end{equation}
\end{enumerate}
where, as clarified by the mathematical expression in \eqref{eq:BQP-completeness}, acceptance is defined as obtaining the outcome one upon measuring the decision qubit register $D$. $\operatorname{BQP}$ is then defined as $\operatorname{BQP}(2/3, 1/3)$.

\subsubsection{Fidelity between two pure states}

\label{sec:fid-pure-pure}

We now prove that the promise version of the problem of estimating the fidelity between two pure states is BQP-complete. In this problem and all that follows, the parameter $x$ is the description of the circuits involved, and the length $|x|$ is the number of bits needed to describe these circuits.

\begin{problem}
[$\left(  \alpha,\beta\right)  $-Fidelity-Pure-Pure]Let $\alpha$ and $\beta$
be such that $0\leq\alpha<\beta\leq1$. Given are descriptions of circuits
$U_{S}^{\psi}$ and $U_{S}^{\phi}$ that prepare the pure states $\psi_{S}$ and $\phi_S$, respectively. Decide which of
the following holds.%
\begin{align}
\text{Yes}  &  \text{:}\qquad F(\psi_{S},\phi_{S})\geq1-\alpha,\\
\text{No}  &  \text{:}\qquad F(\psi_{S},\phi_{S})\leq1-\beta.
\end{align}
\end{problem}

\begin{theorem}
\label{thm:BQP-comp-fid-pure} The promise problem Fidelity-Pure-Pure is BQP-complete.

\begin{enumerate}
\item $\left(  \alpha,\beta\right)  $-Fidelity-Pure-Pure is in BQP for all
$\alpha<\beta$. (It is implicit that the gap between $\alpha$ and $\beta$ is
larger than an inverse polynomial in the input length.)

\item $\left(  \varepsilon,1-\varepsilon\right)  $-Fidelity-Pure-Pure is
BQP-hard, even when $\varepsilon$ decays exponentially in the input length.
\end{enumerate}

\noindent Thus, $\left(  \alpha,\beta\right)  $-Fidelity-Pure-Pure is BQP-complete for
all $\left(  \alpha,\beta\right)  $ such that $0<\alpha<\beta<1$.
\end{theorem}

\begin{proof}
The containment of $\left(\alpha,\beta\right)  $-Fidelity-Pure-Pure in BQP\ is a direct consequence of Algorithm~\ref{alg:fid-pure-states-1}.

So we focus on proving the hardness result. Consider an arbitrary problem $L$ in BQP. Thus, there exists a family $Q$ of circuits such that \eqref{eq:BQP-completeness} and \eqref{eq:BQP-soundness} hold. Given an instance $x$, the acceptance probability of the BQP algorithm is 
\begin{align}
    \label{eq:BQP-accept-prob}
    p_{\text{acc}} &= \left\Vert (\bra{1}_D \otimes I_G) Q \ket{x}_S \ket{0}_A \right\Vert^2_2 \nonumber \\
    &= \bra{x}_S \bra{0}_A Q^\dagger (\outerproj{1}_D \otimes I_G) Q \ket{x}_S \ket{0}_A.
\end{align}
To prove the hardness result (i.e., to see that this is an instance of Fidelity-Pure-Pure), we use the BQP-subroutine theorem \cite{Bennett_1997}. Intuitively, we act with the circuit $Q_{SA \to DG}$ on the input $\ket{x}_S \ket{0}_A$, apply a $\operatorname{CNOT}$ gate from the decision qubit to an ancillary qubit initialized to $\ket{0}_C$, apply the inverse unitary $Q^\dagger$, measure the output qubits, and accept if we get the state $\ket{x}_S \ket{0}_A \ket{1}_C$. The acceptance probability of this procedure is equal to
\begin{equation}
    \label{eq:BQP-hard-accep-initial}
    \tilde{p}_{\text{acc}} = \left\vert (\bra{x}_S \bra{0}_A \bra{1}_C) Q^\dagger \operatorname{CNOT}_{DC} Q (\ket{x}_S \ket{0}_A \ket{0}_C)\right\vert^2.
\end{equation}
Expanding $\operatorname{CNOT}_{DC}$ as 
\begin{equation}
    \operatorname{CNOT}_{DC} \coloneqq \outerproj{0}_D \otimes I_C + \outerproj{1}_D \otimes X_C, 
\end{equation}
where $X_C$ denotes the Pauli-$X$ operator, it follows that
\begin{equation}
    \label{eq:BQP-hard-accep}
    \tilde{p}_{\text{acc}} = \left\vert \bra{x}_S \bra{0}_A Q^\dagger (\outerprod{1}{1}_D \otimes I_G) Q \ket{x}_S \ket{0}_A \right \vert^2.
\end{equation}
Comparing this expression to~\eqref{eq:BQP-accept-prob}, we see that the modified circuit has an acceptance probability equal to the square of the acceptance probability of the original BQP problem. Thus, by repeating the modified algorithm sufficiently many times, we can estimate the acceptance probability $\tilde{p}_{\text{acc}}$, and by taking a square root, we can output an estimate of the acceptance probability $p_{\text{acc}}$ of the original problem. In Appendix~\ref{sec:samples-Fid-Pure-Pure}, we derive the number of samples required to estimate $p_{\text{acc}}$ with accuracy $\varepsilon$ and error probability $\delta$.

The last step to be shown is that the modified acceptance probability $\tilde{p}_{\text{acc}}$ can be rewritten as the fidelity between two pure states. From \eqref{eq:BQP-hard-accep-initial}, we see that 
\begin{align}
    \tilde{p}_{\text{acc}} &= \left\vert (\bra{x}_S \bra{0}_A \bra{1}_C) Q^\dagger \operatorname{CNOT}_{DC} Q (\ket{x}_S \ket{0}_A \ket{0}_C)\right\vert^2 \nonumber \\
    &= F(|\psi\rangle\!\langle\psi|, |\phi\rangle\!\langle\phi|),
\end{align}
where 
\begin{align}
    |\psi\rangle & \coloneqq \ket{x}_S \ket{0}_A \ket{1}_C ,\\
    |\phi\rangle & \coloneqq Q^\dagger \operatorname{CNOT}_{DC} Q \ket{x}_S \ket{0}_A \ket{0}_C.
\end{align}
Thus, an arbitrary instance of a BQP problem can be rewritten as an instance of the fidelity between two pure states, proving that Fidelity-Pure-Pure is indeed a BQP-hard problem. 
\end{proof}

\subsubsection{Fidelity between a pure state and a mixed state}

\label{sec:compl-fid-pure-mixed}

\begin{problem}
[$\left(  \alpha,\beta\right)  $-Fidelity-Pure-Mixed]Let $\alpha$ and $\beta$
be such that $0\leq\alpha<\beta\leq1$. Given are descriptions of circuits
$U_{RS}^{\rho}$ and $U_{S}^{\psi}$ that prepare a purification of a mixed
state $\rho_{S}$ and a pure state $\psi_{S}$, respectively. Decide which of
the following holds.%
\begin{align}
\text{Yes}  &  \text{:}\qquad F(\rho_{S},\psi_{S})\geq1-\alpha,\\
\text{No}  &  \text{:}\qquad F(\rho_{S},\psi_{S})\leq1-\beta.
\end{align}
\end{problem}

\begin{theorem}
\label{thm:BQP-comp-fid} The promise problem Fidelity-Pure-Mixed is BQP-complete.

\begin{enumerate}
\item $\left(  \alpha,\beta\right)  $-Fidelity-Pure-Mixed is in BQP for all
$\alpha<\beta$. (It is implicit that the gap between $\alpha$ and $\beta$ is
larger than an inverse polynomial in the input length.)

\item $\left(  \varepsilon,1-\varepsilon\right)  $-Fidelity-Pure-Mixed is
BQP-hard, even when $\varepsilon$ decays exponentially in the input length.
\end{enumerate}

\noindent Thus, $\left(  \alpha,\beta\right)  $-Fidelity-Pure-Mixed is BQP-complete for
all $\left(  \alpha,\beta\right)  $ such that $0<\alpha<\beta<1$.
\end{theorem}

\begin{proof}
The containment of $\left(  \alpha,\beta\right)  $-Fidelity-Pure-Mixed in
BQP\ is a direct consequence of Algorithm~\ref{alg:fid-pure-mixed}.

So we
focus on proving the hardness result. Let $L$ be an arbitrary promise problem
in BQP, and let $\left\{  \phi^{x}_{DG}\right\}  _{x}$ be a family of
efficiently preparable pure states witnessing membership of $L$ in BQP. System
$D$ is a decision qubit indicating acceptance or rejection of $x$, and system
$G$ is a garbage system that purifies $D$. Suppose that the family $\left\{
\phi^{x}_{DG}\right\}  _{x}$ has completeness $1-\delta$ and soundness
$\delta$. If $x$ is a yes-instance of $L$, then, by the definition of BQP, it
follows that $\left\Vert \langle1|_{D}|\phi^{x}\rangle_{DG}\right\Vert
_{2}^{2}\geq1-\delta$. On the other hand, if $x$ is a no-instance of $L$, then
$\left\Vert \langle1|_{D}|\phi^{x}\rangle_{DG}\right\Vert _{2}^{2}\leq\delta$.
Since%
\begin{align}
\left\Vert \langle1|_{D}|\phi^{x}\rangle_{DG}\right\Vert _{2}^{2}  &
=\langle1|_{D}\operatorname{Tr}_{G}[\phi_{DG}^{x}]|1\rangle_{D}\\
&  =F(|1\rangle\!\langle1|_{D},\operatorname{Tr}_{G}[\phi_{DG}^{x}]),
\end{align}
it follows directly that this is an instance of $\left(  1-\delta
,\delta\right)  $-Fidelity-Pure-Mixed, given that the reduced state
$\operatorname{Tr}_{G}[\phi_{DG}^{x}]$ can be prepared efficiently, as well as
the state $|1\rangle\!\langle 1|_{D}$. The desired hardness result then follows
because BQP$(c,s)\subseteq\ $BQP$(\delta,1-\delta)$, for every $\delta$
exponentially small in the input length.
\end{proof}

\subsubsection{Hilbert--Schmidt distance}

\label{sec:HS-distance}

The next result we prove is that the promise version of the problem of estimating the normalized Hilbert--Schmidt distance of two arbitrary states is BQP-complete. Recall that the normalized Hilbert--Schmidt distance of two states $\rho$ and $\sigma$ is given by
\begin{align}
    \frac{1}{\sqrt{2}} \left\Vert \rho - \sigma \right\Vert_2 
    & \coloneqq \frac{1}{\sqrt{2}} \sqrt{\operatorname{Tr}[(\rho - \sigma)^2]} \nonumber \\
    &= \frac{1}{\sqrt{2}} \sqrt{\operatorname{Tr}[\rho^2] + \operatorname{Tr}[\sigma^2] - 2\operatorname{Tr}[\rho\sigma]}.
\end{align}

If $\rho = \sigma$, then the Hilbert--Schmidt distance is equal to zero. The prefactor of $2^{-1/2}$ is the correct normalization by the following argument. Since $\operatorname{Tr}[\rho\sigma] \geq 0$, the maximum value of the normalized distance satisfies
\begin{align}
\label{eq:HSD-Expansion}
    \frac{1}{\sqrt{2}} &\sqrt{\operatorname{Tr}[\rho^2] + \operatorname{Tr}[\sigma^2] - 2\operatorname{Tr}[\rho\sigma]} \nonumber \\
    &\leq \frac{1}{\sqrt{2}} \sqrt{\operatorname{Tr}[\rho^2] + \operatorname{Tr}[\sigma^2]} \nonumber \\
    &\leq 1,
\end{align}
where the second inequality follows because the purity of an arbitrary state $\rho$ satisfies $\operatorname{Tr}[\rho^2] \leq 1$. The upper bound is achieved by pure orthogonal states.

\begin{problem}
[$\left(  \alpha,\beta\right)  $-Hilbert--Schmidt-Distance]
\label{prob:HS-dist}
Let $\alpha$ and $\beta$
be such that $0\leq\alpha<\beta\leq1$. Given are descriptions of circuits $U_{RS}^{\rho}$ and $U_{RS}^{\sigma}$ that prepare a purification of a mixed states $\rho_{S}$ and $\sigma_{S}$, respectively. Decide which of the following holds.%
\begin{align}
\text{Yes}  &  \text{:}\qquad \NHS{\rho_S - \sigma_S} \geq 1-\alpha,
\label{eq:YES-HS-dist}
\\
\text{No}  &  \text{:}\qquad \NHS{\rho_S - \sigma_S} \leq 1-\beta.
\label{eq:NO-HS-dist}
\end{align}
\end{problem}

\begin{theorem}
\label{thm:BQP-comp-HSD} The promise problem Hilbert--Schmidt-Distance is BQP-complete.

\begin{enumerate}
\item $\left(  \alpha,\beta\right)  $-Hilbert--Schmidt-Distance is in BQP for all $\alpha<\beta$. (It is implicit that the gap between $\alpha$ and $\beta$ is larger than an inverse polynomial in the input length.)

\item $\left(  \varepsilon,1-\varepsilon\right)  $-Hilbert--Schmidt-Distance is BQP-hard, even when $\varepsilon$ decays exponentially in the input length.
\end{enumerate}

\noindent Thus, $\left(  \alpha,\beta\right)  $-Hilbert--Schmidt-Distance is BQP-complete for all $\left(  \alpha,\beta\right)  $ such that $0<\alpha<\beta<1$.
\end{theorem}

\begin{proof}
To show that the problem is BQP-complete, we need to demonstrate two facts: first, that the problem is in BQP, and second, that it is BQP-hard. Let us begin by proving that the problem is in BQP. This part of the proof is well known and understood by now, and it has been used in many quantum algorithms. We discuss it here for completeness. The intuitive idea is to estimate each term in \eqref{eq:HSD-Expansion} separately using a swap test. A term of the form $\operatorname{Tr}[\rho\sigma]$, where $\rho$ and $ \sigma$ are $n$-qubit states, can be estimated by repeatedly performing a swap test sufficiently many times to get a good estimate. Since there are only three terms to estimate, it follows that the problem is in BQP. 

Next, we show that any problem in the BQP class can be reduced to this problem. A simpler way to show this is to map a known BQP-complete problem to our problem. We now show that the BQP-complete Fidelity-Pure-Pure problem can be reduced to this problem. A special case of the Hilbert--Schmidt-Distance problem is when both inputs are pure states. In this scenario, the normalized Hilbert--Schmidt distance is given by
\begin{equation}
\begin{aligned}
    \NHS{\outerprod{\psi}{\psi} - \outerprod{\phi}{\phi}} & = \sqrt{1 - \left\vert \langle \psi \vert \phi \rangle \right\vert^2} \\
    & = \sqrt{1 - F(\psi,\phi)}.
\end{aligned}
\label{eq:HS-reduction-pure}
\end{equation}
Then the YES instance condition in \eqref{eq:YES-HS-dist} and \eqref{eq:HS-reduction-pure} imply that
$
    F(\psi,\phi) \leq \alpha (2-\alpha),
$
in the case of a YES instance of Hilbert--Schmidt-Distance,
and the NO instance condition in \eqref{eq:NO-HS-dist} and 
\eqref{eq:HS-reduction-pure} imply that
$
    F(\psi,\phi) \geq \beta (2-\beta),
$
in the case of a NO instance of Hilbert--Schmidt-Distance.
Since the function $x \to x(2-x)$ is a bijection on the unit interval $[0,1]$, it follows that the ability to decide Hilbert--Schmidt-Distance for pure states implies the ability to decide Fidelity-Pure-Pure, which is a BQP-complete problem by Theorem~\ref{thm:BQP-comp-fid-pure}.
We thus conclude that Hilbert--Schmidt-Distance is BQP-Hard. This, along with the fact that the problem is in the BQP class, concludes the proof.
\end{proof}

\begin{remark}
    The normalized Schatten-$p$ distance between two states $\rho$ and $\sigma$ is defined as
    \begin{equation}
        \frac{1}{2^{1/p}} \left\| \rho - \sigma \right\|_{p} \coloneqq \frac{1}{2^{1/p}} (\operatorname{Tr}[|\rho - \sigma|^p ])^{1/p}. 
        \label{eq:schatten-p-dist-normalized}
    \end{equation}
    We can formulate promise problems from these quantities, generalizing Hilbert--Schmidt-Distance in Problem~\ref{prob:HS-dist}. Plugging pure states $\psi$ and $\phi$ into \eqref{eq:schatten-p-dist-normalized} and exploiting the fact that the eigenvalues of $\psi - \phi$ are equal to $|\sin\theta|$ and $-|\sin\theta|$ \cite[Proof of Theorem~9.3.1]{Wbook17}, where $\theta$  satisfies $F(\psi,\phi) = \cos^2 \theta$, it follows that
    \begin{equation}
        \frac{1}{2^{1/p}} \left\| \psi - \phi \right\|_{p} = \sqrt{1 - F(\psi,\phi)}
    \end{equation}
    for all $p \geq 1$. Thus, by the same reasoning given in the second part of the proof of Theorem~\ref{thm:BQP-comp-HSD}, we conclude that these promise problems are all BQP-hard.
    
    Now consider that estimating the Schatten-$2k$ distance between two states, where $k\in \mathbb{N}$, is in BQP. For constant~$k$, each term in the expansion of $\left\| \rho - \sigma \right\|_{2k}^{2k} = \operatorname{Tr}[(\rho - \sigma)^{2k} ] $ can be estimated in polynomial time \cite{EAOHHK02}, in fact in constant quantum depth \cite{QKW22} after the circuits that prepare multiples copies of $\rho$ and $\sigma$ are executed. Thus, combining with the above, we conclude that, for each constant $k\in \mathbb{N}$, the promise version of the problem of estimating $\frac{1}{2^{1/(2k)}} \left\| \rho - \sigma \right\|_{2k}$ is a BQP-complete problem.
    
\end{remark}

\subsection{Fidelity between a pure state and a channel (QMA-complete)}

\label{sec:fid-pure-channel}

Next, we provide a proof that the promise version of the problem of evaluating the fidelity between a channel and a pure state is QMA-complete. The definition of QMA can be found in \cite{W09}, reproduced here for convenience (but again slightly different in that we consider unitary circuits). Let $A = (A_{\text{yes}}, A_{\text{no}})$ be a promise problem, let $p, q$ be polynomially-bounded functions, and let $a, b : \mathbb{N} \rightarrow [0, 1]$ be functions. Then $A \in \operatorname{QMA}_p(a, b)$ if  there exists a polynomial-time generated family of unitary circuits $Q  = \{Q_n : n \in \mathbb{N}\}$, where each circuit $Q_n$ takes $n + p(n) + q(n)$ input qubits and produces one decision qubit $D$ and $n + p(n) + q(n) - 1$ garbage qubits $G$, with the following properties (as before, we abbreviate each $Q_n$ as $Q_{SAP \to DG}$, thereby suppressing the dependence on the input length $n = |x|$ and explicitly indicating the systems involved at the input and output of the unitary): 

\begin{enumerate}
    \item Completeness: For all $x \in A_{\text{yes}}$, there exists a $q(\vert x\vert )$-qubit quantum state $\sigma$ such that 
    \begin{align}
         \Pr[Q \text{ accepts }(x, \sigma)]  & = 
         \langle 1 \vert_D \operatorname{Tr}_G [\omega_{DG}] \vert 1 \rangle_D \label{eq:QMA_completeness} \\
        & \geq a(\vert x\vert ),
    \end{align}
    where
    \begin{multline}
    \omega_{DG} \coloneqq \\
    Q_{SAP \to DG} (\outerproj{x}_S \otimes \outerproj{0}_A \otimes \sigma_P) (Q_{SAP \to DG})^\dagger.
    \end{multline}

    \item Soundness: For all $x \in A_{\text{no}}$, and every $q(\vert x\vert )$-qubit quantum state $\sigma$, the following inequality holds: 
    \begin{equation}
        \label{eq:QMA_soundness}
        \Pr[Q \text{ accepts }(x, \sigma)] \leq b(\vert x\vert ).
    \end{equation}

\end{enumerate}

Then $\operatorname{QMA} = \bigcup_p \operatorname{QMA}_p(2/3, 1/3)$, where the union is over every polynomially bounded function $p$.

\begin{problem}
[$\left(\alpha,\beta\right)  $-Fidelity-Channel-Pure]Let $\alpha$ and $\beta$
be such that $0\leq\alpha<\beta\leq1$. Given are descriptions of circuits
$U_{SR\to BE}^{\mathcal{N}}$ and $U_{B}^{\psi}$ that prepare a unitary dilation of a channel
\begin{equation}
\mathcal{N}_{S \to B} (\cdot) \coloneqq  \operatorname{Tr}_E[U^\mathcal{N}_{SR\to BE}((\cdot)_S \otimes \vert 0 \rangle\!\langle 0 \vert_R) (U^\mathcal{N}_{SR\to BE})^\dagger]    
\end{equation}
 and a pure state $\psi_{B} \coloneqq U_{B}^{\psi} \vert 0 \rangle\!\langle 0 \vert_B (U_{B}^{\psi})^\dag$, respectively. Decide which of
the following holds:
\begin{align}
\text{Yes}  &  \text{:}\qquad \max\limits_{\rho_S}\  F(\mathcal{N}_{S\to B }(\rho_{S}),\psi_{B})\geq1-\alpha,
\label{eq:fid-ch-pure-state}
\\
\text{No}  &  \text{:}\qquad \max\limits_{\rho_S}\  F(\mathcal{N}_{S \to B}(\rho_{S}),\psi_{B})\leq1-\beta,
\end{align}
where the maximization is over every input density operator $\rho_S$.
\end{problem}

\begin{theorem}
\label{thm:QMA-comp-fid} The promise problem Fidelity-Channel-Pure is QMA-complete.

\begin{enumerate}
\item $\left(  \alpha,\beta\right)  $-Fidelity-Channel-Pure is in QMA for all
$\alpha<\beta$. (It is implicit that the gap between $\alpha$ and $\beta$ is
larger than an inverse polynomial in the input length.)

\item $\left(  \varepsilon,1-\varepsilon\right)  $-Fidelity-Channel-Pure is QMA-hard, even when $\varepsilon$ decays exponentially in the input length.
\end{enumerate}

\noindent Thus,  $\left(  \alpha,\beta\right)  $-Fidelity-Channel-Pure is QMA-complete for
all $\left(  \alpha,\beta\right)  $ such that $0<\alpha<\beta<1$.
\end{theorem}

\begin{proof}
To show that the problem is QMA-complete, we need to demonstrate two facts: first, that the problem is in QMA, and second, that it is QMA-hard.

Let us begin by proving that the problem is in QMA. The intuitive idea is that the prover sends an optimal state $\rho_S$ to the verifier, who then performs the channel $\mathcal{N}_{S \to B}$ on it, followed by the unitary $(U_{B}^{\psi})^\dag$. The verifier then performs a computational basis measurement on all registers of system $B$ and accepts if and only if the all-zeros measurement outcome occurs. Indeed, the acceptance probability of this scheme is precisely equal to the fidelity in \eqref{eq:fid-ch-pure-state}:
\begin{align}
    & \langle 0 \vert_B (U_B^{\psi})^\dagger \mathcal{N}_{S \to B}(\rho_S) U_B^\psi \vert 0 \rangle_B \nonumber
    \\&= \langle \psi \vert_S \mathcal{N}_{S \to B}(\rho_S) \vert \psi \rangle_S \nonumber \\
    &= F(\mathcal{N}_{S \to B}(\rho_S), \psi_B). 
\end{align}
To bring the original expression more closely to the form given in \eqref{eq:QMA_completeness}, observe that
\begin{multline}
        \langle 0 \vert_B (U_B^{\psi})^\dagger \mathcal{N}_{S \to B}(\rho_S) U_B^\psi \vert 0 \rangle_B = 
         \langle 1 \vert_B X_B (U_B^{\psi})^\dagger \times \\ \operatorname{Tr}_E[U^\mathcal{N}_{SR \to BE}(\outerproj{0}_R \otimes \rho_S) (U^\mathcal{N}_{SR \to BE})^\dagger] U_B^\psi X_B \vert 1 \rangle_B,
\end{multline}
where $X_B$ is understood to be the tensor power Pauli $X$ operator acting on all qubits of the $B$ register.
To bring the final expression exactly into the form in \eqref{eq:QMA_completeness}, we need a single decision qubit that we measure. We can use a multi-controlled Toffoli gate from the $B$ register to a single qubit decision qubit. Thus, if we identify $x$ with $0$, $\sigma$ with $\rho_S$, and $Q_n$ with $(X_B \otimes \mathbb{I}_E) \circ ((U_B^\psi)^\dag~\otimes~\mathbb{I}_E) \circ U^\mathcal{N}_{SR \to BE}$, it follows that the problem belongs to the QMA class.

Next, we show that any problem in the QMA class can be polynomially reduced to this problem. Let $P$ be an arbitrary problem in the QMA class. This implies that \eqref{eq:QMA_completeness} and \eqref{eq:QMA_soundness} must hold. This problem can then be thought of as a fidelity problem with a channel $\mathcal{M}_x$ defined as
\begin{equation}
    \mathcal{M}^x_{SAP \to D}(\cdot) \coloneqq \operatorname{Tr}_G[Q (\outerproj{x}_S \otimes \outerproj{0}_A \otimes (\cdot)) Q^\dag].
\end{equation}
Furthermore, we identify the state $\psi$ from the fidelity problem with  $\vert 1 \rangle\!\langle 1 \vert_D$, and then we find that
\begin{align}
    & \langle 1_D \vert \operatorname{Tr}_G [Q (\outerproj{x}_S \otimes \outerproj{0}_A \otimes \sigma_P) Q^\dagger)] \vert 1 \rangle_D \nonumber \\
    &= \bra{1}_G\mathcal{M}^x_{SAP \to D}(\sigma) \ket{1}_G\\
    &= F(\mathcal{M}^x(\sigma),  \vert 1 \rangle\!\langle 1 \vert).
\end{align}
It follows directly that this is an instance of $\left(  1-a(|x|)
,1-b(|x|)\right)  $-Fidelity-Channel-Pure, given that the channel
$\mathcal{M}_x$ can be prepared efficiently, as well as
the state $|1\rangle\!\langle 1|$. The desired hardness result then follows
because QMA$( 1-a(|x|)
,1-b(|x|))\subseteq\ $QMA$(\delta,1-\delta)$, for every $\delta$
exponentially small in the input length.
\end{proof}

\subsection{Fidelity between a pure state and a channel with separable input (QMA(2)-complete)}

\label{sec:fid-pure-ch-sep}

Lastly, we provide a proof for the result that the promise version of the problem of evaluating the fidelity between a pure state and a channel with a separable state as input is QMA(2)-complete. A state is separable if and only if is it not entangled. A separable state $\sigma_{SR}$ can be expanded as follows:
\begin{equation}
\sigma_{SR} = \sum\limits_k p(k) \outerproj{\varphi^k}_S \otimes \outerproj{\phi^k}_R,
\end{equation}
where $\{p(k)\}_k$ is a probability distribution and $\{\outerproj{\varphi^k}_S\}_k$ and $\{\outerproj{\phi^k}_R\}_k$ are sets of pure states.
$\operatorname{SEP}$ is defined as the set of all separable states. QMA(2) is a generalization of QMA with proofs that consist of two systems guaranteed to be unentangled \cite{KMY01,HM10}.

We reproduce the definition of QMA(2) for convenience. Let $A = (A_{\text{yes}}, A_{\text{no}})$ be a promise problem, let $p,q,r$ be polynomially-bounded functions, and let $a, b : \mathbb{N} \rightarrow [0, 1]$ be functions. Then $A \in \operatorname{QMA}(2)_{p, q}(a, b)$ if there exists a polynomial-time generated family of circuits $Q= \{Q_n : n \in \mathbb{N}\}$, where each circuit $Q_n$ takes $n + p(n) + q(n) + r(n)$ input qubits and produces one decision qubit $D$ and $n + p(n) + q(n) + r(n) - 1$ garbage qubits $G$, with the following properties (again, we employ the notation $ Q_{SAP_1P_2 \to DG}$ in what follows): 

\begin{enumerate}
    \item Completeness: For all $x \in A_{\text{yes}}$, there exists a $q(\vert x\vert )$-qubit  state $\rho$ and an $r(\vert x\vert )$-qubit  state $\sigma$ such that 
    \begin{align}
        \label{eq:QMA(2)_completeness}
         \Pr[Q \text{ accepts }(x, \rho, \sigma)] \nonumber 
        &= \langle 1 \vert_D \operatorname{Tr}_G [\omega_{DG}] \vert 1 \rangle_D \nonumber \\
        &\geq a(\vert x\vert ),
    \end{align}
    where
    \begin{multline}
        \omega_{DG} \coloneqq
        Q_{SAP_1P_2 \to DG} (\outerproj{x}_S \otimes \outerproj{0}_A 
        \otimes \\ \rho_{P_1} \otimes \sigma_{P_2}) (Q_{SAP_1P_2 \to DG})^\dagger.
    \end{multline}
    
    \item Soundness: For all $x \in A_{\text{no}}$, and every $q(\vert x\vert )$-qubit  state $\rho$ and $r(\vert x\vert )$-qubit state $\sigma$, the following inequality holds:
    \begin{equation}
        \label{eq:QMA(2)_soundness}
        \Pr[Q \text{ accepts }(x, \rho, \sigma)] \leq b(\vert x\vert ).
    \end{equation}
\end{enumerate}

Then $\operatorname{QMA(2)} = \bigcup_p \operatorname{QMA(2)}_{p, q}(2/3, 1/3)$, where the union is over all polynomially bounded functions $p$ and~$q$.

\begin{problem}
[$\left(\alpha,\beta\right)  $-Fidelity-Pure-Channel-Sep-Inp]Let $\alpha$ and $\beta$
be such that $0\leq\alpha<\beta\leq1$. Given are descriptions of circuits
$U_{SRE \to AE'}^{\mathcal{N}}$ and $U_{A}^{\psi}$ that prepare a unitary dilation of a channel
\begin{multline}
\mathcal{N}_{SR \to A}(\cdot) \coloneqq \\ \operatorname{Tr}_{E'}[U^\mathcal{N}_{SRE \to AE'}((\cdot)_{SR} \otimes \outerproj{0}_E) (U^\mathcal{N}_{SRE \to AE'})^\dagger],    
\end{multline}
and a pure state $\psi_{A}$, respectively. Decide which of
the following holds:
\begin{align}
\text{Yes}  &  \text{:}\qquad \underset{\sigma_{SR} \in \operatorname{SEP}}{\max} F(\mathcal{N}_{SR \to A}(\sigma_{SR}), \psi_{A}) \geq1-\alpha,
\label{eq:fid-ch-sep-state}
\\
\text{No}  &  \text{:}\qquad \underset{\sigma_{SR} \in \operatorname{SEP}}{\max} F(\mathcal{N}_{SR \to A}(\sigma_{SR}), \psi_{A}) \leq1-\beta.
\end{align}
\end{problem}

\begin{theorem}
\label{thm:QMA(2)-comp-fid} The promise problem Fidelity-Pure-Channel-Sep-Inp is QMA$(2)$-complete.

\begin{enumerate}
\item $\left(  \alpha,\beta\right)  $-Fidelity-Pure-Channel-Sep-Inp is in QMA$(2)$ for all
$\alpha<\beta$. (It is implicit that the gap between $\alpha$ and $\beta$ is
larger than an inverse polynomial in the input length.)

\item $\left(  \varepsilon,1-\varepsilon\right)  $-Fidelity-Pure-Channel-Sep-Inp is QMA$(2)$-hard, even when $\varepsilon$ decays exponentially in the input length.
\end{enumerate}

\noindent Thus, $\left(  \alpha,\beta\right)  $-Fidelity-Pure-Channel-Sep-Inp is QMA$(2)$-complete for
all $\left(  \alpha,\beta\right)  $ such that $0<\alpha<\beta<1$.
\end{theorem}

\begin{proof}
To show that the problem is QMA(2)-complete, we need to demonstrate two facts: first, that the problem is in QMA(2), and second, that it is QMA(2)-hard. Let us begin by proving that the problem is in QMA(2). The intuitive idea is that the two provers, using shared randomness, send an optimal separable state $\sigma_{SR}$ to the verifier, who then performs the channel $\mathcal{N}_{SR \to A}$ on it, followed by the unitary $(U_{A}^{\psi})^\dag$. (Note that QMA(2) remains unchanged if the provers have access to shared randomness \cite{HM10}.) The verifier then performs a computational basis measurement on all registers of system~$A$ and accepts if and only if the all-zeros measurement outcome occurs. 

Consider that a separable state can be decomposed as
\begin{equation}
\sigma_{SR} = \sum\limits_k p(k) \outerproj{\varphi^k}_S \otimes \outerproj{\phi^k}_R.
\end{equation}
Indeed, the acceptance probability of this scheme is precisely equal to the fidelity in \eqref{eq:fid-ch-sep-state}:
\begin{align*}
    & F(\mathcal{N}_{SR \to A}(\sigma_{SR}), \psi_{A}) \nonumber \\
    &= \langle \psi \vert_{A} \mathcal{N}_{SR \to A}(\sigma_{SR}) \vert \psi \rangle_{A} \\
    &= \sum\limits_k p(k) \langle \psi \vert_{A} \mathcal{N}_{SR \to A}(\outerproj{\varphi^k}_S \otimes \outerproj{\phi^k}_R) \vert \psi \rangle_{A}.
\end{align*}
The final expression is an average of individual elements. Thus, taking a maximization over all separable states and noting that the maximum is always greater than the average, we conclude that
\begin{align}
    &\underset{\sigma_{SR} \in \operatorname{SEP}}{\max} F(\mathcal{N}_{SR \to A}(\sigma_{SR}), \psi_{A}) \nonumber \\
    &= \max\limits_{\ket{\varphi}_S, \ket{\phi}_R} \bra{\psi}_{A} \mathcal{N}_{SR \to A}(\varphi_S \otimes \phi_R) \ket{\psi}_{A} \nonumber \\
    &= \max\limits_{\ket{\varphi}_S, \ket{\phi}_R} \bra{0}_{A} (U^\psi_{A})^\dagger \mathcal{N}_{SR \to A}(\varphi_S \otimes \phi_R) U^\psi_{A} \ket{0}_{A}.
\end{align}
Thus, we see that 
\begin{multline}
    \underset{\sigma_{SR} \in \operatorname{SEP}}{\max} F(\mathcal{N}_{SR \to A}(\sigma_{SR}), \psi_{A}) = \max\limits_{\ket{\varphi}_S, \ket{\phi}_R} \bra{1}_{A} X_{A} \times \\
    (U^\psi_{A})^\dagger \operatorname{Tr}_{E'} [U^\mathcal{N}_{SRE \to AE'}(\outerproj{0}_E \otimes \varphi_S \otimes \phi_R) \times\\
    (U^\mathcal{N}_{SRE \to AE'})^\dagger] U^\psi_{A} X_{A} \ket{1}_{A},
\end{multline}
where $X_{A}$ is understood to be the tensor-power Pauli $X$ operator acting on all qubits of the $A$ register.
To bring the final expression into the precise form in \eqref{eq:QMA(2)_completeness}, we need a single decision qubit that we measure. We can use a multi-controlled Toffoli gate from the $A$ register to a single qubit decision qubit.  Thus, if we identify $x$ with $0$, $\rho$ with $\varphi_S$, $\sigma$ with $\phi_R$ and $Q_n$ with $(X_{A} \otimes \mathbb{I}_R) \circ ((U_{A}^\psi)^\dag \otimes \mathbb{I}_R) \circ U^\mathcal{N}_{SRE \to AE'}$, it follows that the problem belongs to the QMA(2) class.

Next, we show that any problem in the QMA(2) class can be polynomially reduced to this problem. Let $P$ be an arbitrary problem in the QMA(2) class. This implies that \eqref{eq:QMA(2)_completeness} and \eqref{eq:QMA(2)_soundness} must hold. This problem can then be thought of as a fidelity problem with a channel $\mathcal{M}_x$ defined as
\begin{equation}
    \mathcal{M}^x_{SAP_1P_2 \to D}(\cdot) \coloneqq \operatorname{Tr}_G[Q_n (\outerproj{x}_S \otimes \outerproj{0}_A \otimes (\cdot)_{P_1P_2}) Q_n^\dag].
\end{equation}
Furthermore, by identifying the state $\psi$ from the fidelity problem with $\outerproj{1}$, then we find that
\begin{align}
    &\langle 1 \vert \operatorname{Tr}_G [Q (\outerproj{x}_S \otimes \outerproj{0}_A \otimes \psi_1 \otimes \psi_2) Q^\dagger ] \vert 1 \rangle \\
    &= \langle 1 \vert \mathcal{M}^x(\psi_1 \otimes \psi_2) \vert 1 \rangle \\
    &= F(\mathcal{M}_x(\psi_1 \otimes \psi_2), \outerproj{1}).
\end{align}
It follows directly that this is an instance of $\left(  1-a(|x|)
,1-b(|x|)\right)  $-Fidelity-Channel-Pure, given that the channel
$\mathcal{M}_x$ can be prepared efficiently, as well as
the state $|1\rangle\!\langle 1|$. The desired hardness result then follows
because QMA$(1-a(|x|)
,1-b(|x|))\subseteq\ $QMA$(\delta,1-\delta)$, for every $\delta$
exponentially small in the input length (see \cite[Theorem~9]{HM10}).
\end{proof}

\section{Generating fixed points of quantum channels}
\label{sec:fixed-points-of-channels}

In this section, we discuss how
Algorithm~\ref{alg:fid-channels-single-prover} can generate a fixed-point
state or an approximate fixed-point state of a quantum channel. There are
various associated subtleties in such a scenario that we consider.

As a special case of Algorithm~\ref{alg:fid-channels-single-prover}, we can select $\mathcal{N}_{A\rightarrow B}^{0}$ to be a channel $\mathcal{N}$\ with
its output and input systems having the same dimension (i.e., $\left\vert
A\right\vert =\left\vert B\right\vert $), and we can select the second channel
$\mathcal{N}_{A\rightarrow B}^{1}$\ to be the identity channel. In this case,
the quantity in \eqref{eq:max-fid-channels-fixed-point} is always equal to
one. This follows from the well known fact that every quantum channel with
matching input and output systems has a fixed point state \cite{EHK78} (see
also \cite{Deutsch91,Wolf12}) and because the prover's goal is to maximize the acceptance
probability. That is, for every such channel $\mathcal{N}$,
there exists a state $\rho$ such that%
\begin{equation}
\mathcal{N}(\rho)=\rho, \label{eq:fixed-point-channel}%
\end{equation}
and so the prover can simply send this state.
Related to this, there is a faithfulness property that holds. If the
acceptance probability is equal to one, then it follows that%
\begin{equation}
\sup_{\rho}F(\mathcal{N}(\rho),\rho)=1,
\end{equation}
and we conclude that there exists a state $\rho$ satisfying
\eqref{eq:fixed-point-channel} because the fidelity is continuous and the set
of density operators is convex and compact.

What is interesting in this case is that
Algorithm~\ref{alg:fid-channels-single-prover}\ outputs a fixed point of the
channel $\mathcal{N}$. Fixed points of quantum channels are important not only
for understanding thermalization in a physical process
\cite{bardet2021modified}\ (a fixed point can be understood as an equilibrium
state of the channel)\ but also in the Deutschian theory of closed timelike
curves \cite{Deutsch91}.

We can also modify this approach slightly and employ
Algorithm~\ref{alg:fid-multiple-channels-single-prover}. In this case, the
verifier can employ the following ensemble of channels%
\begin{equation}
\left\{  \frac{1}{L},\mathcal{N}^{\ell}\right\}  _{\ell=0}^{L-1},
\end{equation}
where $\mathcal{N}^{\ell}$ here is defined as%
\begin{equation}
\mathcal{N}^{\ell}=\underbrace{\mathcal{N}\circ\cdots\circ\mathcal{N}}%
_{\ell\text{ times}}.
\end{equation}
In this case, the acceptance probability of
Algorithm~\ref{alg:fid-multiple-channels-single-prover} is given by%
\begin{equation}
\left[  \frac{1}{L}\sup_{\rho,\sigma}\sum_{\ell=0}^{L-1}\sqrt{F(\mathcal{N}^{\ell}(\rho),\sigma)}\right]  ^{2}. \label{eq:fixed-point-multiple-ch}%
\end{equation}
This is again equal to one because the prover can transmit a fixed point to
the verifier, which satisfies%
\begin{equation}
\mathcal{N}^{\ell}(\rho)=\rho\quad\forall\ell\in\left\{  0,\ldots,L-1\right\}
. \label{eq:fixed-point-multiple}%
\end{equation}
Similarly, in this case, a faithfulness property holds as well. If the
expression in \eqref{eq:fixed-point-multiple-ch} is equal to one, then there
exists a state $\rho$ satisfying \eqref{eq:fixed-point-multiple}. Furthermore,
Algorithm~\ref{alg:fid-multiple-channels-single-prover} outputs a fixed point
satisfying \eqref{eq:fixed-point-multiple-ch}.

The cases outlined above are simple. The situation becomes more subtle when
the verifier tries to use the state sent by the prover to solve a
computational problem, as is the case in quantum computation in the presence
of Deutschian closed timelike curves \cite{AW09}. In this case, there are different goals,
which are 1)\ to pass the test of the verifier in
Algorithm~\ref{alg:fid-multiple-channels-single-prover}, as well as 2)\ to
have the decision qubit be as close as possible to the $|1\rangle\!\langle1|$
state. In this case, the prover need not send an exact fixed point, but only
send an approximate fixed point, satisfying%
\begin{equation}
F(\rho,\mathcal{N}(\rho))\geq1-\varepsilon,
\label{eq:approximate-fixed-point-def}%
\end{equation}
or%
\begin{equation}
\left[  \frac{1}{L}\sup_{\sigma}\sum_{\ell=0}^{L-1}\sqrt{F(\mathcal{N}^{\ell}(\rho),\sigma)}\right]  ^{2}\geq1-\varepsilon,
\label{eq:approx-fixed-point-def-2}%
\end{equation}
where $\varepsilon\in (0,1)$.
The prover can do this to optimize the overall acceptance probability of the
QIP algorithm. Somewhat counter-intuitively, approximate fixed points need not
be close to exact fixed points, as illustrated by the following example.
Suppose that $\mathcal{N}$ is a classical channel that takes $1\rightarrow1$
deterministically, but then takes $0\rightarrow0$ with probability
$1-\varepsilon$ and $0\rightarrow1$ with probability $\varepsilon$. In this
case, $1$ is the exact fixed point of this stochastic process, but $0$ is an
approximate fixed point satisfying \eqref{eq:approximate-fixed-point-def}.
However, $0$ is completely distinguishable from $1$ (the fidelity of these two
classical states is equal to zero).

In Appendix~\ref{app:CTCs}, we discuss various issues related to fixed points and
approximate fixed points of quantum channels when attempting to understand
quantum interactive proofs and the computational complexity of Deutschian
closed timelike curves.

\section{Conclusion}

\label{sec:conclusion}

In this paper, we have delineated several algorithms for estimating distinguishability measures on quantum computers. All of the measures are based on trace distance or fidelity, and we have considered them for quantum states, channels, and strategies. Many of the algorithms rely on interaction with a quantum prover, and in these cases, we have replaced the prover with a parameterized quantum circuit.  As such, these methods are not guaranteed to converge for all possible states, channels, and strategies. It is an interesting open question to determine conditions under which the algorithms are guaranteed to converge and run efficiently.

We have also simulated several of the algorithms in both the noiseless and noisy scenarios. We found that the simulations converge well for all states and channels considered, and for all algorithms simulated. As more advanced quantum computers become available (with more qubits and greater reliability), it would be interesting to simulate our algorithms for states and channels involving larger numbers of qubits. All of our Python code is written in a modular way, such that it will be straightforward to explore this direction. Lastly, we proved several complexity-theoretic results about various distance estimation algorithms; in particular, we showed and, in some cases, recalled that there is a fidelity or distance estimation problem that is complete for the commonly studied complexity classes BQP, QMA, QMA(2), QSZK, QIP(2), and QIP.

Going forward from here, it remains open to determine methods for estimating other distinguishability measures such as the Petz--R\'{e}nyi relative entropy \cite{P85,P86}\ and the sandwiched R\'{e}nyi relative entropy \cite{MDSFT13,WWY14}\ of channels \cite{LKDW18}\ and strategies \cite{Wang2019a}. More generally, one could consider distinguishability measures beyond these. One desirable aspect of the algorithms appearing in this paper is that they provide a one-shot interpretation for the various distinguishability measures as the maximum acceptance probability in a quantum interactive proof (with the trace-distance based algorithms and interpretations being already known from \cite{W02,RW05,GW06,G09,G12}). However, it is unclear to us whether one could construct a quantum interactive proof for which the maximum acceptance probability is related to the Petz-- or sandwiched R\'{e}nyi relative entropy of a channel or a strategy.

\textit{Note added}: While finalizing the results of our initial arXiv post \cite{ARSW21}, we noticed the arXiv post \cite{BBC21}, which is related to the contents of Section~\ref{sec:TD-based-measures}. Ref.~\cite{BBC21} is now published as \cite{PhysRevApplied.17.024002}.

\begin{acknowledgments}
We acknowledge insightful discussions with Todd Brun, Patrick Coles, Zoe Holmes, Margarite LaBorde, Dhrumil Patel, Yihui Quek, and Aliza Siddiqui. We thank Robert Salzmann and John Watrous for discussions related to fixed points and thank John Watrous for reminding us of the example after~\eqref{eq:approx-fixed-point-def-2}. We also thank him and Scott Aaronson for discussions related to Deutschian CTCs. We thank Yupan Liu for pointing out a typo. SR and MMW\ acknowledge support from the National Science Foundation under Grant No.~1907615. KS acknowledges support from the Department of Defense.
\end{acknowledgments}

\bibliographystyle{alpha}
\bibliography{Ref}

\appendix

\section{Proofs from main text}

\subsection{Proof of Theorem~\ref{thm:acc-prob-fid-states}}

\label{app:ProofAlg4}

\begin{proof}[Proof of Theorem~\ref{thm:acc-prob-fid-states}]
After Step~1 of Algorithm~\ref{alg:fid-states}, the global state is%
\begin{equation}
|\Phi\rangle_{T^{\prime}T}|0\rangle_{RS}.
\end{equation}
After Step~2 of Algorithm~\ref{alg:fid-states}, it is%
\begin{equation}
\frac{1}{\sqrt{2}}\sum_{i\in\left\{  0,1\right\}  }|i\rangle_{T^{\prime}%
}|i\rangle_{T}|\psi^{i}\rangle_{RS}.
\end{equation}
After Step~4 of Algorithm~\ref{alg:fid-states}, it is%
\begin{equation}
P_{T^{\prime}RF\rightarrow T^{\prime\prime}F^{\prime}}\left(  \frac{1}%
{\sqrt{2}}\sum_{i\in\left\{  0,1\right\}  }|i\rangle_{T^{\prime}}|i\rangle
_{T}|\psi^{i}\rangle_{RS}|0\rangle_{F}\right)  .
\end{equation}
For a fixed unitary $P\equiv P_{T^{\prime}RF\rightarrow T^{\prime\prime
}F^{\prime}}$ of the prover, the acceptance probability is then%
\begin{multline}
\left\Vert \langle\Phi|_{T^{\prime\prime}T}P\left(  \frac{1}{\sqrt{2}}%
\sum_{i\in\left\{  0,1\right\}  }|i\rangle_{T^{\prime}}|i\rangle_{T}|\psi
^{i}\rangle_{RS}|0\rangle_{F}\right)  \right\Vert _{2}^{2}\\
=\frac{1}{2}\left\Vert \langle\Phi|_{T^{\prime\prime}T}P\sum_{i\in\left\{
0,1\right\}  }|i\rangle_{T^{\prime}}|i\rangle_{T}|\psi^{i}\rangle
_{RS}|0\rangle_{F}\right\Vert _{2}^{2}.
\end{multline}
In a quantum interactive proof, the prover is trying to maximize the
probability that the verifier accepts. So the acceptance probability of
Algorithm~\ref{alg:fid-states} is given by%
\begin{equation}
\max_{P_{T^{\prime}RF\rightarrow T^{\prime\prime}F^{\prime}}}\frac{1}%
{2}\left\Vert \langle\Phi|_{T^{\prime\prime}T}P\sum_{i\in\left\{  0,1\right\}
}|i\rangle_{T^{\prime}}|i\rangle_{T}|\psi^{i}\rangle_{RS}|0\rangle
_{F}\right\Vert _{2}^{2}.
\end{equation}

Setting%
\begin{align}
P_{R\rightarrow F^{\prime}}^{0}  &  \coloneqq\langle0|_{T^{\prime\prime}%
}P_{T^{\prime}RF\rightarrow T^{\prime\prime}F^{\prime}}|0\rangle_{T^{\prime}%
}|0\rangle_{F},\\
P_{R\rightarrow F^{\prime}}^{1}  &  \coloneqq\langle1|_{T^{\prime\prime}%
}P_{T^{\prime}RF\rightarrow T^{\prime\prime}F^{\prime}}|1\rangle_{T^{\prime}%
}|0\rangle_{F},
\end{align}
we have that%
\begin{align}
&  \frac{1}{2}\left\Vert \langle\Phi|_{T^{\prime\prime}T}P\sum_{i\in\left\{
0,1\right\}  }|i\rangle_{T^{\prime}}|i\rangle_{T}|\psi^{i}\rangle
_{RS}|0\rangle_{F}\right\Vert _{2}^{2}\nonumber\\
&  =\frac{1}{4}\left\Vert \sum_{i\in\left\{  0,1\right\}  }P_{R\rightarrow
F^{\prime}}^{i}|\psi^{i}\rangle_{RS}\right\Vert _{2}^{2}\\
&  =\frac{1}{4}\sum_{i,j\in\left\{  0,1\right\}  }\langle\psi^{i}%
|_{RS}(P_{R\rightarrow F^{\prime}}^{i})^{\dag}P_{R\rightarrow F^{\prime}}%
^{j}|\psi^{j}\rangle_{RS}\\
&  \leq\frac{1}{2}\left(  1+\operatorname{Re}\left\{  \langle\psi^{0}%
|_{RS}(P_{R\rightarrow F^{\prime}}^{0})^{\dag}P_{R\rightarrow F^{\prime}}%
^{1}|\psi^{1}\rangle_{RS}\right\}  \right) \\
&  \leq\frac{1}{2}\left(  1+\left\vert \langle\psi^{0}|_{RS}(P_{R\rightarrow
F^{\prime}}^{0})^{\dag}P_{R\rightarrow F^{\prime}}^{1}|\psi^{1}\rangle
_{RS}\right\vert \right)  .
\end{align}
The first inequality follows because $P_{R\rightarrow F^{\prime}}^{i}$ is a
contraction for $i\in\left\{  0,1\right\}  $, so that $(P_{R\rightarrow
F^{\prime}}^{i})^{\dag}P_{R\rightarrow F^{\prime}}^{i}\leq I_{F^{\prime}}$.
Then consider that%
\begin{align}
&  \left\vert \langle\psi^{0}|_{RS}(P_{R\rightarrow F^{\prime}}^{0})^{\dag
}P_{R\rightarrow F^{\prime}}^{1}|\psi^{1}\rangle_{RS}\right\vert \nonumber\\
&  \leq\max_{P^{0},P^{1}}\left\{
\begin{array}
[c]{c}%
\left\vert \langle\psi^{0}|_{RS}(P_{R\rightarrow F^{\prime}}^{0})^{\dag
}P_{R\rightarrow F^{\prime}}^{1}|\psi^{1}\rangle_{RS}\right\vert \\
:\left\Vert P^{i}\right\Vert _{\infty}\leq1\ \forall i
\end{array}
\right\} \\
&  =\sqrt{F}(\rho_{S}^{0},\rho_{S}^{1}).
\end{align}
The last line is a consequence of the following reasoning (which is the same
as that employed in Section~III\ in \cite{PhysRevA.94.022310}). The inequality%
\begin{multline}
\max_{P^{0},P^{1}}\left\{
\begin{array}
[c]{c}%
\left\vert \langle\psi^{0}|_{RS}(P_{R\rightarrow F^{\prime}}^{0})^{\dag
}P_{R\rightarrow F^{\prime}}^{1}|\psi^{1}\rangle_{RS}\right\vert \\
:\left\Vert P^{i}\right\Vert _{\infty}\leq1\ \forall i
\end{array}
\right\} \\
\geq\sqrt{F}(\rho_{S}^{0},\rho_{S}^{1})
\end{multline}
holds because the isometries $P_{R\rightarrow F^{\prime}}^{0}$ and
$P_{R\rightarrow F^{\prime}}^{1}$ that achieve the maximum for the fidelity
are each contractions and the optimization is conducted over all contractions.
The opposite inequality%
\begin{multline}
\max_{P^{0},P^{1}}\left\{
\begin{array}
[c]{c}%
\left\vert \langle\psi^{0}|_{RS}(P_{R\rightarrow F^{\prime}}^{0})^{\dag
}P_{R\rightarrow F^{\prime}}^{1}|\psi^{1}\rangle_{RS}\right\vert \\
:\left\Vert P^{i}\right\Vert _{\infty}\leq1\ \forall i
\end{array}
\right\} \\
\leq\sqrt{F}(\rho_{S}^{0},\rho_{S}^{1})
\end{multline}
is a consequence of the fact that every contraction can be written as a convex
combination of isometries \cite[Theorem~5.10]{Z11}. Indeed, this means that,
for each $i\in\left\{  0,1\right\}  $,%
\begin{equation}
P_{R\rightarrow F^{\prime}}^{i}=\sum_{x}p_{i}(x)W_{R\rightarrow F^{\prime}%
}^{i,x},
\end{equation}
where $\{p_{i}(x)\}_{x}$ is a probability distribution and $W_{R\rightarrow
F^{\prime}}^{i,x}$ is an isometry, for each $i$ and $x$. Then we find that%
\begin{align}
&  \left\vert \langle\psi^{0}|_{RS}(P_{R\rightarrow F^{\prime}}^{0})^{\dag
}P_{R\rightarrow F^{\prime}}^{1}|\psi^{1}\rangle_{RS}\right\vert \nonumber\\
&  =\left\vert
\begin{array}
[c]{c}%
\langle\psi^{0}|_{RS}\left(  \sum_{x}p_{0}(x)W_{R\rightarrow F^{\prime}}%
^{0,x}\right)  ^{\dag}\times\\
\quad\left(  \sum_{x^{\prime}}p_{1}(x^{\prime})W_{R\rightarrow F^{\prime}%
}^{1,x^{\prime}}\right)  |\psi^{1}\rangle_{RS}%
\end{array}
\right\vert \\
&  =\left\vert \sum_{x,x^{\prime}}p_{0}(x)p_{1}(x^{\prime})\langle\psi
^{0}|_{RS}\left(  W_{R\rightarrow F^{\prime}}^{0,x}\right)  ^{\dag
}W_{R\rightarrow F^{\prime}}^{1,x^{\prime}}|\psi^{1}\rangle_{RS}\right\vert \\
&  \leq\sum_{x,x^{\prime}}p_{0}(x)p_{1}(x^{\prime})\left\vert \langle\psi
^{0}|_{RS}\left(  W_{R\rightarrow F^{\prime}}^{0,x}\right)  ^{\dag
}W_{R\rightarrow F^{\prime}}^{1,x^{\prime}}|\psi^{1}\rangle_{RS}\right\vert \\
&  \leq\max_{x,x^{\prime}}\left\vert \langle\psi^{0}|_{RS}\left(
W_{R\rightarrow F^{\prime}}^{0,x}\right)  ^{\dag}W_{R\rightarrow F^{\prime}%
}^{1,x^{\prime}}|\psi^{1}\rangle_{RS}\right\vert \\
&  \leq\sqrt{F}(\rho_{S}^{0},\rho_{S}^{1}).
\end{align}
Thus, an upper bound on the acceptance probability of
Algorithm~\ref{alg:fid-states} is as follows:%
\begin{equation}
\frac{1}{2}\left(  1+\sqrt{F}(\rho_{S}^{0},\rho_{S}^{1})\right)  .
\end{equation}
This upper bound can be achieved if the prover applies a unitary extension of
the following isometry:%
\begin{equation}
P_{T^{\prime}RF\rightarrow T^{\prime\prime}F^{\prime}}=\sum_{i\in\left\{
0,1\right\}  }|i\rangle_{T^{\prime\prime}}\langle i|_{T^{\prime}}\otimes
P_{R\rightarrow F^{\prime}}^{i}\otimes\langle0|_{F},
\end{equation}
where $P_{R\rightarrow F^{\prime}}^{0}$ and $P_{R\rightarrow F^{\prime}}^{1}$
are isometries achieving the maximum in the fidelity $F(\rho_{S}^{0},\rho
_{S}^{1})$.
\end{proof}

\subsection{Proof of Theorem~\ref{thm:acc-prob-mixed-state-swap-test}}

\label{app:ProofAlg5}

\begin{proof}[Proof of Theorem~\ref{thm:acc-prob-mixed-state-swap-test}]
After Step~1 of Algorithm~\ref{alg:mixed-state-swap-test}, the global state is%
\begin{equation}
|\Phi\rangle_{T^{\prime}T}|0\rangle_{R_{1}S_{1}R_{2}S_{2}}.
\end{equation}
After Step~2, the global state is%
\begin{equation}
|\Phi\rangle_{T^{\prime}T}|\psi^{\rho^{0}}\rangle_{R_{1}S_{1}}|\psi^{\rho^{1}%
}\rangle_{R_{2}S_{2}}.
\end{equation}
After Step~3, it becomes%
\begin{multline}
\frac{1}{\sqrt{2}}|0\rangle_{T}|0\rangle_{T^{\prime}}|\psi^{\rho^{0}}%
\rangle_{R_{1}S_{1}}|\psi^{\rho^{1}}\rangle_{R_{2}S_{2}}\\
+\frac{1}{\sqrt{2}}|1\rangle_{T}|1\rangle_{T^{\prime}}|\psi^{\rho^{1}}%
\rangle_{R_{2}S_{1}}|\psi^{\rho^{0}}\rangle_{R_{1}S_{2}}.
\end{multline}
The verifier then sends systems $T^{\prime}$, $R_{1}$, and $R_{2}$ to the
prover, who appends the state $|0\rangle_{F}$ and acts with a unitary
$P_{T^{\prime}R_{1}R_{2}F\rightarrow T^{\prime\prime}F^{\prime}}$. Without
loss of generality, and for simplicity of the ensuing analysis, we can imagine
that before applying the unitary $P_{T^{\prime}R_{1}R_{2}F\rightarrow
T^{\prime\prime}F^{\prime}}$, the prover applies a controlled SWAP to systems
$T^{\prime}$, $R_{1}$, and $R_{2}$, so that the state before applying
$P_{T^{\prime}R_{1}R_{2}F\rightarrow T^{\prime\prime}F^{\prime}}$ is as
follows:%
\begin{multline}
\frac{1}{\sqrt{2}}|0\rangle_{T}|0\rangle_{T^{\prime}}|\psi^{\rho^{0}}%
\rangle_{R_{1}S_{1}}|\psi^{\rho^{1}}\rangle_{R_{2}S_{2}}%
\label{eq:gen-swap-global-state-before-prover-1}\\
+\frac{1}{\sqrt{2}}|1\rangle_{T}|1\rangle_{T^{\prime}}|\psi^{\rho^{1}}%
\rangle_{R_{1}S_{1}}|\psi^{\rho^{0}}\rangle_{R_{2}S_{2}}.
\end{multline}
This follows because the prover can apply arbitrary unitaries to his received
systems, and one such possible unitary is to apply this controlled SWAP, undo
it, and then apply $P_{T^{\prime}R_{1}R_{2}F\rightarrow T^{\prime\prime
}F^{\prime}}$. However, the latter two unitaries are a particular example of a
unitary $P_{T^{\prime}R_{1}R_{2}F\rightarrow T^{\prime\prime}F^{\prime}}$. So
we proceed with the ensuing analysis assuming that the global state, before the prover
applies $P_{T^{\prime}R_{1}R_{2}F\rightarrow T^{\prime\prime}F^{\prime}}$, is
given by \eqref{eq:gen-swap-global-state-before-prover-1}. Note that the
actions of tensoring in the state $|0\rangle_{F}$ and applying $P_{T^{\prime
}R_{1}R_{2}F\rightarrow T^{\prime\prime}F^{\prime}}$ together constitute an
isometry%
\begin{equation}
P_{T^{\prime}R_{1}R_{2}\rightarrow T^{\prime\prime}F^{\prime}}\coloneqq
P_{T^{\prime}R_{1}R_{2}F\rightarrow T^{\prime\prime}F^{\prime}}|0\rangle_{F},
\end{equation}
resulting in the state%
\begin{multline}
\frac{1}{\sqrt{2}}P_{T^{\prime}R_{1}R_{2}\rightarrow T^{\prime\prime}%
F^{\prime}}|0\rangle_{T}|0\rangle_{T^{\prime}}|\psi^{\rho^{0}}\rangle
_{R_{1}S_{1}}|\psi^{\rho^{1}}\rangle_{R_{2}S_{2}}\\
+\frac{1}{\sqrt{2}}P_{T^{\prime}R_{1}R_{2}\rightarrow T^{\prime\prime
}F^{\prime}}|1\rangle_{T}|1\rangle_{T^{\prime}}|\psi^{\rho^{1}}\rangle
_{R_{1}S_{1}}|\psi^{\rho^{0}}\rangle_{R_{2}S_{2}}.
\end{multline}
Let us set%
\begin{align}
P_{R_{1}R_{2}\rightarrow F^{\prime}}^{00} &  \coloneqq\langle0|_{T^{\prime
\prime}}P_{T^{\prime}R_{1}R_{2}\rightarrow T^{\prime\prime}F^{\prime}%
}|0\rangle_{T^{\prime}},\\
P_{R_{1}R_{2}\rightarrow F^{\prime}}^{11} &  \coloneqq\langle1|_{T^{\prime
\prime}}P_{T^{\prime}R_{1}R_{2}\rightarrow T^{\prime\prime}F^{\prime}%
}|1\rangle_{T^{\prime}}.
\end{align}
The verifier finally performs a Bell measurement and accepts if and only if
the outcome $\Phi_{T^{\prime\prime}T}$ occurs. The acceptance probability is
then\begin{widetext}%
\begin{align}
&  \left\Vert \langle\Phi|_{TT^{\prime\prime}}\frac{1}{\sqrt{2}}\left(
|0\rangle_{T}P_{T^{\prime}R_{1}R_{2}\rightarrow T^{\prime\prime}F^{\prime}%
}|0\rangle_{T^{\prime}}|\psi^{\rho_{0}}\rangle_{R_{1}S_{1}}|\psi^{\rho_{1}%
}\rangle_{R_{2}S_{2}}+|1\rangle_{T}P_{T^{\prime}R_{1}R_{2}\rightarrow
T^{\prime\prime}F^{\prime}}|1\rangle_{T^{\prime}}|\psi^{\rho_{1}}%
\rangle_{R_{1}S_{1}}|\psi^{\rho_{0}}\rangle_{R_{2}S_{2}}\right)  \right\Vert
_{2}^{2}\nonumber\\
&  =\frac{1}{4}\left\Vert \langle0|_{T^{\prime\prime}}P_{T^{\prime}R_{1}%
R_{2}\rightarrow T^{\prime\prime}F^{\prime}}|0\rangle_{T^{\prime}}|\psi
^{\rho_{0}}\rangle_{R_{1}S_{1}}|\psi^{\rho_{1}}\rangle_{R_{2}S_{2}}%
+\langle1|_{T^{\prime\prime}}P_{T^{\prime}R_{1}R_{2}\rightarrow T^{\prime
\prime}F^{\prime}}|1\rangle_{T^{\prime}}|\psi^{\rho_{1}}\rangle_{R_{1}S_{1}%
}|\psi^{\rho_{0}}\rangle_{R_{2}S_{2}}\right\Vert _{2}^{2}\\
&  =\frac{1}{4}\left\Vert P_{R_{1}R_{2}\rightarrow F^{\prime}}^{00}|\psi
^{\rho_{0}}\rangle_{R_{1}S_{1}}|\psi^{\rho_{1}}\rangle_{R_{2}S_{2}}%
+P_{R_{1}R_{2}\rightarrow F^{\prime}}^{11}|\psi^{\rho_{1}}\rangle_{R_{1}S_{1}%
}|\psi^{\rho_{0}}\rangle_{R_{2}S_{2}}\right\Vert _{2}^{2}\\
&  =\frac{1}{4}\left(
\begin{array}
[c]{c}%
\langle\psi^{\rho_{0}}|_{R_{1}S_{1}}\langle\psi^{\rho_{1}}|_{R_{2}S_{2}%
}\left(  P_{R_{1}R_{2}\rightarrow F^{\prime}}^{00}\right)  ^{\dag}%
P_{R_{1}R_{2}\rightarrow F^{\prime}}^{00}|\psi^{\rho_{0}}\rangle_{R_{1}S_{1}%
}|\psi^{\rho_{1}}\rangle_{R_{2}S_{2}}\\
+\langle\psi^{\rho_{1}}|_{R_{1}S_{2}}\langle\psi^{\rho_{0}}|_{R_{2}S_{1}%
}\left(  P_{R_{1}R_{2}\rightarrow F^{\prime}}^{11}\right)  ^{\dag}%
P_{R_{1}R_{2}\rightarrow F^{\prime}}^{11}|\psi^{\rho_{1}}\rangle_{R_{1}S_{1}%
}|\psi^{\rho_{0}}\rangle_{R_{2}S_{2}}\\
+\langle\psi^{\rho_{0}}|_{R_{1}S_{1}}\langle\psi^{\rho_{1}}|_{R_{2}S_{2}%
}\left(  P_{R_{1}R_{2}\rightarrow F^{\prime}}^{00}\right)  ^{\dag}%
P_{R_{1}R_{2}\rightarrow F^{\prime}}^{11}|\psi^{\rho_{1}}\rangle_{R_{1}S_{1}%
}|\psi^{\rho_{0}}\rangle_{R_{2}S_{2}}\\
+\langle\psi^{\rho_{1}}|_{R_{1}S_{1}}\langle\psi^{\rho_{0}}|_{R_{2}S_{2}%
}\left(  P_{R_{1}R_{2}\rightarrow F^{\prime}}^{11}\right)  ^{\dag}%
P_{R_{1}R_{2}\rightarrow F^{\prime}}^{00}|\psi^{\rho_{0}}\rangle_{R_{1}S_{1}%
}|\psi^{\rho_{1}}\rangle_{R_{2}S_{2}}%
\end{array}
\right)  \\
&  \leq\frac{1}{4}\left(  2+2\operatorname{Re}\left\{  \langle\psi^{\rho_{0}%
}|_{R_{1}S_{1}}\langle\psi^{\rho_{1}}|_{R_{2}S_{2}}\left(  P_{R_{1}%
R_{2}\rightarrow F^{\prime}}^{00}\right)  ^{\dag}P_{R_{1}R_{2}\rightarrow
F^{\prime}}^{11}|\psi^{\rho_{1}}\rangle_{R_{1}S_{1}}|\psi^{\rho_{0}}%
\rangle_{R_{2}S_{2}}\right\}  \right)  \\
&  \leq\frac{1}{4}\left(  2+2\left\vert \langle\psi^{\rho_{0}}|_{R_{1}S_{1}%
}\langle\psi^{\rho_{1}}|_{R_{2}S_{2}}\left(  P_{R_{1}R_{2}\rightarrow
F^{\prime}}^{00}\right)  ^{\dag}P_{R_{1}R_{2}\rightarrow F^{\prime}}^{11}%
|\psi^{\rho_{1}}\rangle_{R_{1}S_{1}}|\psi^{\rho_{0}}\rangle_{R_{2}S_{2}%
}\right\vert \right)  \\
&  =\frac{1}{2}\left(  1+\left\vert \langle\psi^{\rho_{0}}|_{R_{1}S_{1}%
}\langle\psi^{\rho_{1}}|_{R_{2}S_{2}}\left(  P_{R_{1}R_{2}\rightarrow
F^{\prime}}^{00}\right)  ^{\dag}P_{R_{1}R_{2}\rightarrow F^{\prime}}^{11}%
|\psi^{\rho_{1}}\rangle_{R_{1}S_{1}}|\psi^{\rho_{0}}\rangle_{R_{2}S_{2}%
}\right\vert \right)  \\
&  \leq\frac{1}{2}\left(  1+\max_{U_{R_{1}R_{2}}}\left\vert \langle\psi
^{\rho_{0}}|_{R_{1}S_{1}}\langle\psi^{\rho_{1}}|_{R_{2}S_{2}}U_{R_{1}R_{2}%
}|\psi^{\rho_{1}}\rangle_{R_{1}S_{1}}|\psi^{\rho_{0}}\rangle_{R_{2}S_{2}%
}\right\vert \right)  .\label{eq:proof-swap-mixed-states}%
\end{align}
\end{widetext}The steps given above follow for reasons very similar to those
given in the proof of Theorem~\ref{thm:acc-prob-fid-states}. Continuing, we
find that%
\begin{align}
\text{Eq.~\eqref{eq:proof-swap-mixed-states}} &  =\frac{1}{2}\left(
1+\sqrt{F(\rho^{0}\otimes\rho^{1},\rho^{1}\otimes\rho^{0})}\right)  \\
&  =\frac{1}{2}\left(  1+\sqrt{F(\rho^{0},\rho^{1})F(\rho^{1},\rho^{0}%
)}\right)  \label{eq:mult-fid-proof-step}\\
&  =\frac{1}{2}\left(  1+F(\rho^{0},\rho^{1})\right)
,\label{eq:sym-fid-proof-step}%
\end{align}
where we used the multiplicativity of the fidelity for tensor-product states
to get \eqref{eq:mult-fid-proof-step} and the symmetric property of fidelity
to arrive at \eqref{eq:sym-fid-proof-step}. Thus, we have established
\eqref{eq:accept-prob-mixed-swap-test} as an upper bound on the acceptance
probability. This upper bound can be achieved by setting $F^{\prime}\simeq
R_{1}R_{2}$ and
\begin{align}
P_{T^{\prime}R_{1}R_{2}F\rightarrow T^{\prime\prime}F^{\prime}} &
=|0\rangle_{T^{\prime\prime}}\langle0|_{T^{\prime}}\otimes I_{R_{1}%
R_{2}\rightarrow F^{\prime}}\otimes\langle0|_{F}\\
&  +|1\rangle_{T^{\prime\prime}}\langle1|_{T^{\prime}}\otimes U_{R_{1}}\otimes
U_{R_{2}}^{\dag}\otimes\langle0|_{F},
\end{align}
where $U_{R_{1}}$ is a unitary that achieves the fidelity for $F(\rho^{0}%
,\rho^{1})$, so that%
\begin{equation}
\sqrt{F}(\rho^{0},\rho^{1})=\langle\psi^{\rho^{0}}|_{R_{1}S_{1}}U_{R_{1}}%
|\psi^{\rho^{1}}\rangle_{R_{1}S_{1}}.
\end{equation}
This concludes the proof.
\end{proof}

\subsection{Proof of Theorem~\ref{thm:acc-prob-fid-channels}}

\label{app:ProofAlg8}

\begin{proof}[Proof of Theorem~\ref{thm:acc-prob-fid-channels}]
After Step~1 of Algorithm~\ref{alg:fid-channels}, the global state is%
\begin{equation}
|\Phi\rangle_{T^{\prime}T}|\psi\rangle_{RA}|0\rangle_{E^{\prime}}.
\end{equation}
After Step~2 of Algorithm~\ref{alg:fid-channels}, it is%
\begin{equation}
\frac{1}{\sqrt{2}}\sum_{i\in\left\{  0,1\right\}  }|i\rangle_{T^{\prime}%
}|i\rangle_{T}U^{i}|\psi\rangle_{RA}|0\rangle_{E^{\prime}},
\end{equation}
where $U^{i}\equiv U_{AE^{\prime}\rightarrow BE}^{i}$ for $i\in\left\{
0,1\right\}  $. After Step~4 of Algorithm~\ref{alg:fid-channels}, it is%
\begin{equation}
P\left(  \frac{1}{\sqrt{2}}\sum_{i\in\left\{  0,1\right\}  }|i\rangle
_{T^{\prime}}|i\rangle_{T}U^{i}|\psi\rangle_{RA}|00\rangle_{E^{\prime}%
F}\right)  ,
\end{equation}
where $P\equiv P_{T^{\prime}EF\rightarrow T^{\prime\prime}F^{\prime}}$. For a
fixed unitary $P_{T^{\prime}EF\rightarrow T^{\prime\prime}F^{\prime}}$ of the
max-prover and fixed state $|\psi\rangle_{RA}$ of the min-prover, the
acceptance probability is then%
\begin{multline}
\left\Vert \langle\Phi|_{T^{\prime\prime}T}P\left(  \frac{1}{\sqrt{2}}%
\sum_{i\in\left\{  0,1\right\}  }|i\rangle_{T^{\prime}}|i\rangle_{T}U^{i}%
|\psi\rangle_{RA}|00\rangle_{E^{\prime}F}\right)  \right\Vert _{2}^{2}\\
=\frac{1}{2}\left\Vert \langle\Phi|_{T^{\prime\prime}T}P\sum_{i\in\left\{
0,1\right\}  }|i\rangle_{T^{\prime}}|i\rangle_{T}U^{i}|\psi\rangle
_{RA}|00\rangle_{E^{\prime}F}\right\Vert _{2}^{2},
\end{multline}
In a competing-provers quantum interactive proof, the max-prover is trying to
maximize the probability that the verifier accepts, while the min-prover is
trying to minimize the acceptance probability. Since the max-prover plays
second in this game, the acceptance probability of
Algorithm~\ref{alg:fid-channels} is given by%
\begin{equation}
\min_{|\psi\rangle_{RA}}\max_{P}\frac{1}{2}\left\Vert \langle\Phi
|_{T^{\prime\prime}T}P\sum_{i\in\left\{  0,1\right\}  }|ii\rangle_{T^{\prime
}T}U^{i}|\psi\rangle_{RA}|00\rangle_{E^{\prime}F}\right\Vert _{2}^{2}.
\end{equation}
Applying the analysis of Theorem~\ref{thm:acc-prob-fid-states}, it follows
that%
\begin{multline}
\max_{P}\frac{1}{2}\left\Vert \langle\Phi|_{T^{\prime\prime}T}P\sum
_{i\in\left\{  0,1\right\}  }|ii\rangle_{T^{\prime}T}U^{i}|\psi\rangle
_{RA}|00\rangle_{E^{\prime}F}\right\Vert _{2}^{2}\\
=\frac{1}{2}\left(  1+\sqrt{F}(\mathcal{N}_{A\rightarrow B}^{0}(\psi
_{RA}),\mathcal{N}_{A\rightarrow B}^{1}(\psi_{RA}))\right)  .
\end{multline}
Thus, after applying the minimization over every input state $\psi_{RA}$, the
claim in \eqref{eq:acc-prob-fid-channels} follows.
\end{proof}

\subsection{Proof of Theorem~\ref{thm:acc-prob-fid-strategies}}

\label{app:ProofAlg9}

\begin{proof}[Proof of Theorem~\ref{thm:acc-prob-fid-strategies}]
After Step~1 of Algorithm~\ref{alg:fid-strategies}, the global state is%
\begin{equation}
|\Phi\rangle_{T^{\prime}T}|\psi\rangle_{RA}|0\rangle_{E^{\prime n}},
\end{equation}
where we have employed the shorthand $E^{\prime n}\equiv E_{1}^{\prime}\cdots
E_{n}^{\prime}$. After Step~6 of Algorithm~\ref{alg:fid-strategies}, the
global state is%
\begin{multline}
|\varphi(\psi,\{S^{j}\}_{j=1}^{n-1})\rangle\equiv\\
\frac{1}{\sqrt{2}}\sum_{i\in\left\{  0,1\right\}  }|i\rangle_{T^{\prime}%
}|i\rangle_{T}U^{i,n-1}\prod\limits_{j=1}^{n-1}\left(  S^{j}U^{i,j}\right)
|\psi\rangle_{RA}|0\rangle_{E^{\prime n}},
\end{multline}
where we have omitted many of the system labels for simplicity. After Step~8 of
Algorithm~\ref{alg:fid-strategies}, the global state is%
\begin{equation}
P|\varphi(\psi,\{S^{j}\}_{j=1}^{n-1})\rangle.
\end{equation}
For a fixed unitary $P$ of the max-prover and a fixed pure co-strategy
$(\psi,\{S^{j}\}_{j=1}^{n-1})$ of the min-prover, the acceptance probability
is thus%
\begin{equation}
\left\Vert \langle\Phi|_{T^{\prime\prime}T}P|\varphi(\psi,\{S^{j}%
\}_{j=1}^{n-1})\rangle\right\Vert _{2}^{2}.
\end{equation}
In a double-prover quantum interactive proof, the max-prover is trying to
maximize the probability that the verifier accepts, while the min-prover is
trying to minimize the acceptance probability. Since the max-prover plays
second in this game, the acceptance probability of
Algorithm~\ref{alg:fid-strategies} is given by%
\begin{equation}
\min_{(\psi,\{S^{j}\}_{j=1}^{n-1})}\max_{P}\left\Vert \langle\Phi
|_{T^{\prime\prime}T}P|\varphi(\psi,\{S^{j}\}_{j=1}^{n-1})\rangle\right\Vert
_{2}^{2}.
\end{equation}
Applying the analysis of Theorem~\ref{thm:acc-prob-fid-states}, it follows
that%
\begin{multline}
\max_{P}\left\Vert \langle\Phi|_{T^{\prime\prime}T}P|\varphi(\psi
,\{S^{j}\}_{j=1}^{n-1})\rangle\right\Vert _{2}^{2}\\
=\frac{1}{2}\left(  1+\sqrt{F}(\mathcal{N}^{0,(n)}\circ\mathcal{S}%
^{(n-1)},\mathcal{N}^{1,(n)}\circ\mathcal{S}^{(n-1)})\right)  .
\end{multline}
Thus, after applying a minimization over every pure co-strategy $\mathcal{S}%
^{(n-1)}$, the claim in \eqref{eq:acc-prob-fid-strategies} follows.
\end{proof}

\subsection{Proof of Theorem~\ref{thm:max-fid-channels}}

\label{app:ProofAlg10}

\begin{proof}[Proof of Theorem~\ref{thm:max-fid-channels}]
After Step~2 of Algorithm~\ref{alg:fid-channels-single-prover}, the global
state is%
\begin{equation}
|\Phi\rangle_{T^{\prime}T}|\psi\rangle_{RA}|0\rangle_{E^{\prime}}.
\end{equation}
After Step~3, the global state is%
\begin{equation}
\frac{1}{\sqrt{2}}\sum_{i\in\left\{  0,1\right\}  }|i\rangle_{T^{\prime}%
}|i\rangle_{T}U_{AE^{\prime}\rightarrow BE}^{i}|\psi\rangle_{RA}%
|0\rangle_{E^{\prime}}.
\end{equation}
After Step~5, it is%
\begin{equation}
\frac{1}{\sqrt{2}}P\sum_{i\in\left\{  0,1\right\}  }|i\rangle_{T^{\prime}%
}|i\rangle_{T}U_{AE^{\prime}\rightarrow BE}^{i}|\psi\rangle_{RA}%
|0\rangle_{E^{\prime}},
\end{equation}
where $P\equiv P_{T^{\prime}EF\rightarrow T^{\prime\prime}F^{\prime}}$. For a
fixed state $|\psi\rangle_{RA}$ and unitary $P_{T^{\prime}EF\rightarrow
T^{\prime\prime}F^{\prime}}$ of the prover, the acceptance probability is
\begin{equation}
\frac{1}{2}\left\Vert \langle\Phi|_{T^{\prime\prime}T}P\sum_{i\in\left\{
0,1\right\}  }|i\rangle_{T^{\prime}}|i\rangle_{T}U_{AE^{\prime}\rightarrow
BE}^{i}|\psi\rangle_{RA}|0\rangle_{E^{\prime}}\right\Vert _{2}^{2}.
\end{equation}
In a QIP algorithm, the prover chooses his actions in order to maximize the
acceptance probability, so that the acceptance probability is%
\begin{equation}
\frac{1}{2}\sup_{\substack{|\psi\rangle_{RA},\\P}}\left\Vert \langle
\Phi|_{T^{\prime\prime}T}P\sum_{i\in\left\{  0,1\right\}  }|i\rangle
_{T^{\prime}}|i\rangle_{T}U_{AE^{\prime}\rightarrow BE}^{i}|\psi\rangle
_{RA}|0\rangle_{E^{\prime}}\right\Vert _{2}^{2}.
\end{equation}
By the same reasoning employed in the proof of
Theorem~\ref{thm:acc-prob-fid-states}, we conclude that%
\begin{multline}
\frac{1}{2}\sup_{P}\left\Vert \langle\Phi|_{T^{\prime\prime}T}P\sum
_{i\in\left\{  0,1\right\}  }|i\rangle_{T^{\prime}}|i\rangle_{T}U_{AE^{\prime
}\rightarrow BE}^{i}|\psi\rangle_{RA}|0\rangle_{E^{\prime}}\right\Vert
_{2}^{2}\\
=\frac{1}{2}\left(  1+\sqrt{F}(\mathcal{N}_{A\rightarrow B}^{0}(\rho
_{A}),\mathcal{N}_{A\rightarrow B}^{1}(\rho_{A}))\right)  ,
\end{multline}
where $\rho_{A}$ is the reduced state of $\psi_{RA}$ (i.e., $\operatorname{Tr}%
_{R}[\psi_{RA}]=\rho_{A}$). Now including the optimization over every pure
state $\psi_{RA}$, we conclude the claim in \eqref{eq:max-fid-channels}.
\end{proof}

\subsection{Proof of Theorem~\ref{thm:fid-multiple-states}}

\label{app:ProofAlg11}

\begin{proof}[Proof of Theorem~\ref{thm:fid-multiple-states}]
After Step~2 of Algorithm~\ref{alg:fid-multiple-states}, the global state is%
\begin{equation}
\sum_{x\in\mathcal{X}}\sqrt{p(x)}|xx\rangle_{T^{\prime}T}|\psi^{x}\rangle
_{RS}.
\end{equation}
After Step~4, it is%
\begin{equation}
P\sum_{x\in\mathcal{X}}\sqrt{p(x)}|xx\rangle_{T^{\prime}T}|\psi^{x}%
\rangle_{RS}|0\rangle_{F},
\end{equation}
where $P\equiv P_{T^{\prime}RF\rightarrow T^{\prime\prime}F^{\prime}}$. Then,
for a fixed unitary $P_{T^{\prime}RF\rightarrow T^{\prime\prime}F^{\prime}}$,
the acceptance probability is%
\begin{multline}
\left\Vert \langle\Phi|_{T^{\prime\prime}T}P\sum_{x\in\mathcal{X}}\sqrt
{p(x)}|xx\rangle_{T^{\prime}T}|\psi^{x}\rangle_{RS}|0\rangle_{F}\right\Vert
_{2}^{2}=\label{eq:uhlmann-steps-1}\\
\sup_{|\varphi\rangle_{F^{\prime}S}}\left\vert \langle\Phi|_{T^{\prime\prime
}T}\langle\varphi|_{F^{\prime}S}P\sum_{x\in\mathcal{X}}\sqrt{p(x)}%
|xx\rangle_{T^{\prime}T}|\psi^{x}\rangle_{RS}|0\rangle_{F}\right\vert ^{2},
\end{multline}
where the optimization is over every pure state $|\varphi\rangle_{F^{\prime}%
S}$ and we have used the fact that $\left\Vert |\phi\rangle\right\Vert
_{2}^{2}=\sup_{|\psi\rangle:\left\Vert |\psi\rangle\right\Vert _{2}%
=1}\left\vert \langle\psi|\phi\rangle\right\vert ^{2}$. This implies that the
acceptance probability is given by%
\begin{equation}
\sup_{|\varphi\rangle_{F^{\prime}S},P}\left\vert \langle\Phi|_{T^{\prime
\prime}T}\langle\varphi|_{F^{\prime}S}P\sum_{x\in\mathcal{X}}\sqrt
{p(x)}|xx\rangle_{T^{\prime}T}|\psi^{x}\rangle_{RS}|0\rangle_{F}\right\vert
^{2}.
\end{equation}
Recall Uhlmann's theorem \cite{U76}, which is the statement that%
\begin{equation}
F(\omega_{C},\tau_{C})=\sup_{V_{B}}\left\vert \langle\varphi^{\tau}|_{BC}%
V_{B}\otimes I_{C}|\varphi^{\omega}\rangle_{BC}\right\vert ^{2},
\end{equation}
where $\omega_{C}$ and $\tau_{C}$ are density operators with respective
purifications $|\varphi^{\omega}\rangle_{BC}$ and $|\varphi^{\tau}\rangle
_{BC}$ and the optimization is over every unitary $V_{B}$. Observing that the
unitary $P_{T^{\prime}RF\rightarrow T^{\prime\prime}F^{\prime}}$ acts on
systems $T^{\prime}RF$ of $\sum_{x\in\mathcal{X}}\sqrt{p(x)}|xx\rangle
_{T^{\prime}T}|\psi^{x}\rangle_{RS}|0\rangle_{F}$ and systems $T^{\prime
\prime}F^{\prime}$ of $|\Phi\rangle_{T^{\prime\prime}T}|\varphi\rangle
_{F^{\prime}S}$, that their respective reduced states on systems $TS$ are%
\begin{align}
&  \sum_{x\in\mathcal{X}}p(x)|x\rangle\!\langle x|_{T}\otimes\rho_{S}^{x},\\
&  \pi_{T}\otimes\sigma_{S},
\end{align}
where $\pi_{T}$ is the maximally mixed state and $\sigma_{S}%
\coloneqq\operatorname{Tr}_{F^{\prime}}[\varphi_{F^{\prime}S}]$, and applying
Uhlmann's theorem, we conclude that the acceptance probability is given by%
\begin{align}
&  \sup_{\sigma_{S}}F\!\left(  \sum_{x\in\mathcal{X}}p(x)|x\rangle\!\langle
x|_{T}\otimes\rho_{S}^{x},\pi_{T}\otimes\sigma_{S}\right)
\label{eq:sym-dist-meas}\\
&  =\left[  \sup_{\sigma_{S}}\sqrt{F}\!\left(  \sum_{x\in\mathcal{X}%
}p(x)|x\rangle\!\langle x|_{T}\otimes\rho_{S}^{x},\pi_{T}\otimes\sigma
_{S}\right)  \right]  ^{2}\\
&  =\frac{1}{d}\left[  \sup_{\sigma_{S}}\sum_{x\in\mathcal{X}}\sqrt{p(x)}%
\sqrt{F}\!\left(  \rho_{S}^{x},\sigma_{S}\right)  \right]  ^{2}.
\label{eq:uhlmann-steps-last}%
\end{align}
In the second equality, we made use of the direct-sum property of the root
fidelity \cite[Proposition~4.29]{KW20book}. We note here that the analysis
employed is the same as that used to show that the
CLOSE-IMAGE problem is QIP(2)-complete \cite{HMW13,HMW14}.

We can also write the acceptance probability as%
\begin{multline}
\left\Vert \langle\Phi|_{T^{\prime\prime}T}P\sum_{x\in\mathcal{X}}\sqrt
{p(x)}|xx\rangle_{T^{\prime}T}|\psi^{x}\rangle_{RS}|0\rangle_{F}\right\Vert
_{2}^{2}\\
=\frac{1}{d}\left\Vert \sum_{x\in\mathcal{X}}\sqrt{p(x)}P_{R\rightarrow
F^{\prime}}^{x}|\psi^{x}\rangle_{RS}\right\Vert _{2}^{2}%
\end{multline}
where we have defined%
\begin{equation}
P_{R\rightarrow F^{\prime}}^{x}\coloneqq\langle x|_{T^{\prime\prime}%
}P_{T^{\prime}RF\rightarrow T^{\prime\prime}F^{\prime}}|x\rangle_{T^{\prime}%
}|0\rangle_{F}.
\end{equation}
The upper bound in \eqref{eq:upper-bound-fid-multiple-states} follows because%
\begin{align}
&  \frac{1}{d}\left\Vert \sum_{x\in\mathcal{X}}\sqrt{p(x)}P_{R\rightarrow
F^{\prime}}^{x}|\psi^{x}\rangle_{RS}\right\Vert _{2}^{2}\nonumber\\
&  =\frac{1}{d}\sum_{x,y\in\mathcal{X}}\sqrt{p(x)p(y)}\langle\psi^{x}%
|_{RS}(P_{R\rightarrow F^{\prime}}^{x})^{\dag}P_{R\rightarrow F^{\prime}}%
^{y}|\psi^{y}\rangle_{RS}\\
&  =\frac{1}{d}\sum_{x\in\mathcal{X}}p(x)\langle\psi^{x}|_{RS}(P_{R\rightarrow
F^{\prime}}^{x})^{\dag}P_{R\rightarrow F^{\prime}}^{x}|\psi^{x}\rangle
_{RS}\nonumber\\
&  +\frac{2}{d}\sum_{\substack{x,y\in\mathcal{X}\\:x<y}}\sqrt{p(x)p(y)}%
\operatorname{Re}[\langle\psi^{x}|_{RS}(P^{x})^{\dag}P^{y}|\psi^{y}%
\rangle_{RS}]\\
&  \leq\frac{1}{d}+\frac{2}{d} \sum_{x,y\in\mathcal{X}:x<y}\sqrt{p(x)p(y)}\sqrt{F}%
(\rho_{S}^{x},\rho_{S}^{y})
\end{align}
where the first equality follows by expanding the norm, the second by
splitting the terms into those for which $x=y$ and $x < y$, and the
inequality follows because $(P_{R\rightarrow F^{\prime}}^{x})^{\dag
}P_{R\rightarrow F^{\prime}}^{x}\leq I_{R}$ and from reasoning similar to that
in the proof of Theorem~\ref{thm:acc-prob-fid-states}.

The final statement about tightness of the upper bound for the case $d=2$ follows by picking
$P^{x}$ and $P^{y}$ for $x < y$ to be isometries from Uhlmann's theorem, as
was done at the end of the proof of Theorem~\ref{thm:acc-prob-fid-states}.
\end{proof}

\subsection{Proof of Theorem~\ref{thm:fid-multiple-channels}}

\label{app:ProofAlg12}

\begin{proof}[Proof of Theorem~\ref{thm:fid-multiple-channels}]
We can employ the result of Theorem~\ref{thm:fid-multiple-states}. For a fixed
state $\psi_{RA}$ of the min-prover, the acceptance probability is equal to%
\begin{equation}
\frac{1}{d}\left[  \sup_{\sigma_{RB}}\sum_{x\in\mathcal{X}}\sqrt{p(x)}\sqrt
{F}(\mathcal{N}_{A\rightarrow B}^{x}(\psi_{RA}),\sigma_{RB})\right]  ^{2},
\label{eq:fixed-state-accept-prob-mult-ch}%
\end{equation}
as a consequence of Theorem~\ref{thm:fid-multiple-states}. Thus, we arrive at
the claim in \eqref{eq:accept-prob-mult-channels} by minimizing over every
state $\psi_{RA}$ of the min-prover.

The upper bound in \eqref{eq:upper-bound-accept-prob-mult-ch} follows from the
upper bound in \eqref{eq:upper-bound-fid-multiple-states}. Indeed, for a fixed
state $\psi_{RA}$ of the min-prover, the acceptance probability in
\eqref{eq:fixed-state-accept-prob-mult-ch}\ is bounded from above by%
\begin{equation}
\frac{1}{d}+\frac{2}{d}\sum_{\substack{x,y\in\mathcal{X}:\\x<y}}\sqrt
{p(x)p(y)}\sqrt{F}(\mathcal{N}_{A\rightarrow B}^{x}(\psi_{RA}),\mathcal{N}%
_{A\rightarrow B}^{y}(\psi_{RA})).
\end{equation}
After taking infima, we arrive at \eqref{eq:upper-bound-accept-prob-mult-ch}.

The final statement follows from the same reasoning employed at the end of the
proof of Theorem~\ref{thm:fid-multiple-states}.
\end{proof}

\section{Number of samples for Fidelity-Pure-Pure}
\label{sec:samples-Fid-Pure-Pure}

In Theorem~\ref{thm:BQP-comp-fid-pure}, we argued that the problem Fidelity-Pure-Pure is BQP-complete; i.e., every problem in BQP can be reduced to this problem in polynomial time. In this section, we discuss the number of samples required to obtain a desired accuracy and confidence. Let us first recall Hoeffding's bound.

\begin{lemma}[Hoeffding's Bound]
\label{lem:hoeffding}
Suppose that we are given $n$ independent samples $Y_1, \ldots, Y_n$ of a bounded random variable $Y$ taking values in the interval $[a,b]$ and having mean $\mu$. Set 
\begin{equation}
    \overline{Y_n} \coloneqq \frac{1}{n} (Y_1 + \ldots +Y_n)
\end{equation}
to be the sample mean. Let $\varepsilon \in (0,1)$ be the desired accuracy, and let $1-\delta$ be the desired success probability, where $\delta \in (0,1)$. Then
\begin{equation}
\label{eq:Hoeff-bound}
\Pr[\vert \overline{Y_n} - \mu \vert \leq \varepsilon] \geq 1-\delta,
\end{equation}
as long as 
\begin{equation}
    n \geq \frac{M^2}{2\varepsilon^2} \ln \!\left( \frac{2}{\delta}\right),
\end{equation}
where $M \coloneqq b -  a$.
\end{lemma}

In the main text, we mapped a general BQP algorithm to Fidelity-Pure-Pure. In a general BQP algorithm, we measure a single qubit called the decision qubit, leading to a random variable $Y$ taking the value $0$ with probability $1-p$ and the value $1$ with probability $p$, where $p$ is the acceptance probability of the algorithm. We repeat this procedure $n$ times and label the outcomes $Y_1, \ldots, Y_n$. We output the mean
\begin{equation}
    \overline{Y_n} = \frac{1}{n}\left(Y_1 + \ldots + Y_n\right)
\end{equation}
as an estimate for the true value $p$ (as seen in \eqref{eq:BQP-accept-prob})
\begin{equation}
    p = \bra{x}_S \bra{0}_A Q^\dagger (\outerproj{1}_D \otimes I_G) Q \ket{x}_S \ket{0}_A.
\end{equation}
By plugging into Lemma~\ref{lem:hoeffding}, setting
\begin{equation}
\mu=p
\end{equation}
therein, and taking $n$ to satisfy the condition $n \geq \frac{1}{2\varepsilon^2} \ln\! \left(\frac{2}{\delta}\right)$, we can achieve an error $\varepsilon$ and confidence $\delta$ (as defined in \eqref{eq:Hoeff-bound}).

Now, we see from \eqref{eq:BQP-hard-accep} that the modified algorithm has an acceptance probability $p^2$, i.e.,  equal to the square of the original BQP problem's acceptance probability. In the modified algorithm, we measure the decision qubit, leading to a random variable $Z$ taking value $0$ with probability $1-p^2$ and the value $1$ with probability $p^2$. We repeat the procedure $m$ times and label the outcomes $Z_1, \ldots, Z_m$. We output the mean 
\begin{equation}
    \overline{Z}_m = \frac{1}{m} \left(Z_1 + \ldots + Z_m \right)
\end{equation}
as an estimate for the true value $p^2$ (as seen in \eqref{eq:BQP-hard-accep}). Setting $\tilde{\mu} = p^2$, and plugging into Lemma~\ref{lem:hoeffding}, it follows that
\begin{equation}
\Pr[\vert \overline{Z}_m - \tilde{\mu} \vert \leq \varepsilon^2] \geq 1-\delta,
\end{equation}
if
\begin{equation}
    m \geq \frac{1}{2\varepsilon^4} \ln \!\left( \frac{2}{\delta}\right).
\end{equation}
Consider the following inequalities:
\begin{align}
    \varepsilon^2 &\geq \left\vert \overline{Z}_m - \tilde{\mu} \right\vert \nonumber \\
    &= \left\vert \overline{Z}_m - \mu^2 \right\vert \nonumber \\
    &= \left\vert \sqrt{\overline{Z}_m} - \mu \right\vert \left\vert \sqrt{\overline{Z}_m} + \mu \right\vert \nonumber \\
    &\geq \left\vert \sqrt{\overline{Z}_m} - \mu \right\vert^2,
\end{align}
where the second inequality is derived from the fact that $\overline{Z}_m, \mu \in [0,1]$, so that $\left\vert\overline{Z}_m + \mu \right\vert \geq \left\vert \overline{Z}_m - \mu \right\vert$. Thus,
\begin{equation}
    \left\vert \sqrt{\overline{Z}_m} - \mu \right\vert \leq \varepsilon.
\end{equation}

In other words, 
\begin{equation}
    \varepsilon^2 \geq \left\vert \overline{Z}_m - \mu^2 \right\vert \implies \varepsilon \geq \left\vert \sqrt{\overline{Z}_m} - \mu \right\vert
    \end{equation}
    so that
    \begin{align}
    \Pr\left[\left\vert \sqrt{\overline{Z}_m} - \mu \right\vert \leq \varepsilon\right] &\geq \Pr[\left\vert \overline{Z}_m - \mu^2 \right\vert \leq \varepsilon^2] \nonumber \\
    &\geq 1-\delta.
\end{align}

Thus, $\sqrt{\overline{Z}_m}$ is an estimator for $p$ and taking 
\begin{equation}
    m \geq \frac{1}{2\varepsilon^4} \ln \!\left( \frac{2}{\delta}\right)
\end{equation}
suffices to achieve an error $\varepsilon$ and confidence $\delta$ in estimating $p$.

\section{Approximate fixed points and Deutschian closed timelike curves}

\label{app:CTCs}

The computational complexity of computation assisted by Deutschian closed
timelike curves (CTCs) was solved in \cite{AW09}, in which these authors
showed that the power of classical and quantum computing are equivalent and
equal to PSPACE, which is the class of decision problems solvable with
polynomial memory. Let us briefly review these results. In the Deutschian
model of CTCs \cite{Deutsch91}, we suppose that chronology-respecting qubits
in a state~$\rho$ can interact with chronology-violating qubits in a
state~$\sigma$ according to a unitary transformation $U$. Let $S$ denote the
quantum system for the chronology-respecting qubits, and let $C$ denote the
quantum system for the chronology-violating qubits. Then the output state of
the transformation is as follows:%
\begin{equation}
\operatorname{Tr}_{C}[U_{SC}(\rho_{S}\otimes\sigma_{C})U_{SC}^{\dag}].
\end{equation}
In an effort to avoid grandfather and unproved theorem paradoxes, Deutsch
postulates that nature imposes the following self-consistency condition on the
state of the CTC\ qubits:%
\begin{equation}
\sigma_{C}=\operatorname{Tr}_{S}[U_{SC}(\rho_{S}\otimes\sigma_{C})U_{SC}%
^{\dag}]. \label{eq:deutsch-self-cons}%
\end{equation}
At a first glance, this condition might seem innocuous, but its implications
for quantum information processing are dramatic, essentially due to the fact
that \eqref{eq:deutsch-self-cons} allows for non-linear evolutions, which are
disallowed in standard quantum mechanics. Indeed, quantum processors assisted
by Deutschian CTCs can violate the uncertainty principle
\cite{BHW09,PhysRevLett.110.060501}, can break the no-cloning theorem
\cite{BWW13,Yuan2015}, and can solve computational problems believed to be
difficult \cite{AW09}.

The connection of Deutschian CTCs (D-CTCs) with fixed points of channels is
that the condition in \eqref{eq:deutsch-self-cons} demands that the state of
the CTC system be a fixed point of the quantum channel $\mathcal{N}_{U,\rho}$:%
\begin{equation}
\omega_{C}\rightarrow\mathcal{N}_{U,\rho}(\omega_{C}%
)\coloneqq \operatorname{Tr}_{S}[U_{SC}(\rho_{S}\otimes\omega_{C})U_{SC}%
^{\dag}].
\end{equation}
Thus, this is how Section~\ref{sec:fixed-points-of-channels} connects with Deutschian CTCs.

The class of computational problems efficiently decidable by a quantum
computer assisted by D-CTCs is called BQP$_{\text{CTC}}$, and it is formally
defined as follows. Set $\delta\in(0,1/2)$. Let $G$ be a universal set of
quantum gates. A quantum D-CTC\ algorithm is a deterministic polynomial time
algorithm that takes as input a string $x\in\left\{  0,1\right\}  ^{n}$ and
produces an encoding of a unitary quantum circuit $U$ using gates from $G$.
This unitary acts on two systems of qubits, called $S$ and $C$ as discussed
above, which consist of $p(n)$ and $q(n)$ qubits, respectively, where $p(n)$
and $q(n)$ are polynomials. The system $S$ is initialized to the all-zeros
state $|0\rangle\!\langle0|_{S}^{\otimes p(n)}$, and the system $C$ is
initialized to a state $\sigma_{C}$ satisfying the causal self-consistency
condition in \eqref{eq:deutsch-self-cons} with $\rho_{S}=|0\rangle
\!\langle0|_{S}^{\otimes p(n)}$. That is, $\sigma_{C}$ is such that%
\begin{equation}
\mathcal{N}_{U,|0\rangle\!\langle0|_{S}^{\otimes p(n)}}(\sigma_{C})=\sigma
_{C}.
\end{equation}
Let $\mathcal{M}$ be a measurement of the last qubit of $S$ in the
computational basis. The algorithm accepts the input $x$ if%
\begin{equation}
\mathcal{M}(\operatorname{Tr}_{C}[U_{SC}(|0\rangle\!\langle0|_{S}^{\otimes
p(n)}\otimes\sigma_{C})U_{SC}^{\dag}]) \label{eq:CTC-state-after-meas}%
\end{equation}
results in the output 1 with probability at least $1-\delta$ for every state
$\sigma_{C}$ satisfying \eqref{eq:deutsch-self-cons}. The algorithm rejects if
\eqref{eq:CTC-state-after-meas} results in output 1 with probability no larger
than $\delta$ for every state $\sigma_{C}$ satisfying
\eqref{eq:deutsch-self-cons}.\ The algorithm decides the promise problem
$A=A_{\text{yes}}\cup A_{\text{no}}\subseteq\left\{  0,1\right\}  ^{\ast}$
(where $A_{\text{yes}}\cap A_{\text{no}}=\emptyset$) if the algorithm accepts
every input $x\in A_{\text{yes}}$ and rejects every input $x\in A_{\text{no}}%
$. BQP$_{\text{CTC}}$ is the class of all promise problems that are decided by some
quantum D-CTC\ algorithm.

It is already known from \cite{AW09} that BQP$_{\text{CTC}}=\ $PSPACE, and it
is also known that QIP $=$ PSPACE \cite{JJUW11}. Thus, it immediately follows
from these results that BQP$_{\text{CTC}}=\ $QIP. Here, we discuss an attempt
at a direct proof that BQP$_{\text{CTC}}\subseteq\ $QIP, which ideally
would be arguably simpler to see than by examining the proofs of the
equalities BQP$_{\text{CTC}}=\ $PSPACE and QIP $=$ PSPACE individually.
However, there are some difficulties in establishing this direct proof. We
note here that this is related to an open question posed in \cite[Section~8]%
{A05}, the spirit of which is to find a direct proof of the containment
BQP$_{\text{CTC}}\subseteq\ $QIP.

Consider the following purported algorithm for simulating $\operatorname{BQP}%
_{\operatorname{CTC}}$ in $\operatorname{QIP}$:

\begin{algorithm}
\label{alg:CTC-alg-complicated}The algorithm proceeds as follows:

\begin{enumerate}
\item The verifier prepares a state%
\begin{equation}
|\Phi\rangle_{T^{\prime}T}\coloneqq\sum_{\ell=0}^{L-1}\sqrt{\frac{1}{L}}%
|\ell\ell\rangle_{T^{\prime}T}%
\end{equation}
on registers $T^{\prime}$ and $T$ and prepares system $S^{L}$ in the all-zeros
state $|0\rangle_{S^{L}}$.

\item The prover transmits the system $C$ of the state $|\psi\rangle_{RC}$ to
the verifier.

\item Using the circuit $U_{SC}$, the verifier performs the following
controlled unitary:%
\begin{equation}
\sum_{\ell=0}^{L-1}|\ell\rangle\!\langle\ell|_{T}\otimes U_{S_{1}^{\ell}%
C}^{\ell},
\end{equation}
where%
\begin{equation}
U_{S_{1}^{\ell}C}^{\ell}\coloneqq \underbrace{\left(  U_{S_{\ell}C\rightarrow S_{\ell
}C}\circ\cdots\circ U_{S_{1}C\rightarrow S_{1}C}\right)  }_{\ell\text{ times}}%
\end{equation}

\item The verifier transmits systems $T^{\prime}$ and $S^{L}$ to the max-prover.

\item The prover prepares a system $F$ in the $|0\rangle_{F}$ state and acts
on systems $T^{\prime}$, $S^{L}$, and $F$ with a unitary $P_{T^{\prime}%
S^{L}F\rightarrow T^{\prime\prime}F^{\prime}}$ to produce the output systems
$T^{\prime\prime}$ and $F^{\prime}$, where $T^{\prime\prime}$ is a qudit system.

\item The prover sends system $T^{\prime\prime}$ to the verifier, who then
performs a qudit Bell measurement%
\begin{equation}
\{\Phi_{T^{\prime\prime}T},I_{T^{\prime\prime}T}-\Phi_{T^{\prime\prime}T}\}
\end{equation}
on systems $T^{\prime\prime}$ and $T$, where $\Phi_{T^{\prime\prime}T}$ is
defined in \eqref{eq:qudit-bell-state}. The verifier then initializes a system
$S_{L+1}$ to the all-zeros state $|0\rangle_{S_{L+1}}$, performs the unitary
$U_{S_{L+1}C}$, and measures the decision qubit of system $S_{L+1}$. The
verifier accepts if and only if the outcome $\Phi_{T^{\prime\prime}T}$ occurs
and the decision qubit is measured to be in the $|1\rangle$ state.
\end{enumerate}
\end{algorithm}

\begin{proposition}
The acceptance probability of Algorithm~\ref{alg:CTC-alg-complicated} is equal
to%
\begin{equation}
\left[  \sup_{\rho_{C},\sigma_{G}}\frac{1}{L}\sum_{\ell=0}^{L-1}\sqrt
{F}\!\left(
\begin{array}
[c]{c}%
|1\rangle\!\langle1|_{D}\otimes\sigma_{G},\\
U_{SC}(|0\rangle\!\langle0|_{S}\otimes\mathcal{N}^{\ell}(\rho_{C}%
))U_{SC}^{\dag}%
\end{array}
\right)  \right]  ^{2},\label{eq:acc-prob-BQP-CTC-sim}%
\end{equation}
where%
\begin{equation}
\mathcal{N}(\omega_{C})\coloneqq \operatorname{Tr}_{S}[U_{SC}(|0\rangle\!\langle
0|_{S}\otimes\omega_{C})U_{SC}^{\dag}].
\end{equation}

\end{proposition}

\begin{proof}
This follows by employing reasoning similar to that for \cite[Lemma~4.2]{R09}
(see also \cite{KW00}). This reasoning is also very similar to the reasoning used
around \eqref{eq:uhlmann-steps-1}--\eqref{eq:uhlmann-steps-last}. For a fixed
state $|\psi\rangle_{RC}$ of the prover, the global state after Step~6 of
Algorithm~\ref{alg:CTC-alg-complicated}, but before the measurements,\ is%
\begin{equation}
PU_{S_{L+1}C}\sum_{\ell=0}^{L-1}\sqrt{\frac{1}{L}}|\ell\ell\rangle_{T^{\prime
}T}U_{S_{1}^{\ell}C}^{\ell}|0\rangle_{S^{L+1}}|\psi\rangle_{RC},
\end{equation}
where $P\equiv P_{T^{\prime}S^{L}F\rightarrow T^{\prime\prime}F^{\prime}}$.
Then, by splitting the systems $S_{L+1}C$ into the decision qubit $D$ and
denoting all other qubits by $G$, the acceptance probability is given by%
\begin{multline}
\left\Vert
\begin{array}
[c]{c}%
\langle\Phi|_{T^{\prime\prime}T}\langle1|_{D}PU_{S_{L+1}C}\times\\
\sum_{\ell=0}^{L-1}\sqrt{\frac{1}{L}}|\ell\ell\rangle_{T^{\prime}T}%
U_{S_{1}^{\ell}C}^{\ell}|0\rangle_{S^{L+1}}|\psi\rangle_{RC}%
\end{array}
\right\Vert _{2}^{2}\label{eq:CTC-accept-prob-final}\\
=\sup_{|\varphi\rangle_{F^{\prime}G}}\left\vert
\begin{array}
[c]{c}%
\langle\Phi|_{T^{\prime\prime}T}\langle1|_{D}\langle\varphi|_{F^{\prime}%
G}PU_{S_{L+1}C}\times\\
\sum_{\ell=0}^{L-1}\sqrt{\frac{1}{L}}|\ell\ell\rangle_{T^{\prime}T}%
U_{S_{1}^{\ell}C}^{\ell}|0\rangle_{S^{L+1}}|\psi\rangle_{RC}%
\end{array}
\right\vert ^{2}.
\end{multline}
Considering that the reduced state of $\sum_{\ell=0}^{L-1}\sqrt{\frac{1}{L}%
}|\ell\ell\rangle_{T^{\prime}T}U_{S_{1}^{\ell}C}^{\ell}|0\rangle_{S^{L+1}%
}|\psi\rangle_{RC}$, after tracing over all systems sent to the prover, is%
\begin{equation}
\frac{1}{L}\sum_{\ell=0}^{L-1}|\ell\rangle\!\langle\ell|_{T}\otimes
\mathcal{N}^{\ell}(\rho_{C}),
\end{equation}
where $\rho_{C}\coloneqq \operatorname{Tr}_{R}[\psi_{RC}]$, and the reduced state of
$|\Phi\rangle_{T^{\prime\prime}T}|1\rangle_{D}|\varphi\rangle_{F^{\prime}G}$,
after tracing over all systems not transmitted by the prover, is%
\begin{equation}
\frac{1}{L}\sum_{\ell=0}^{L-1}|\ell\rangle\!\langle\ell|_{T}\otimes
|1\rangle\!\langle1|_{D}\otimes\sigma_{G},
\end{equation}
where $\sigma_{G}\coloneqq \operatorname{Tr}_{F^{\prime}}[\varphi_{F^{\prime}G}]$, we
conclude by Uhlmann's theorem that \eqref{eq:CTC-accept-prob-final} is equal
to%
\begin{multline}
\sup_{\sigma_{G}}F\!\left(
\begin{array}
[c]{c}%
\frac{1}{L}\sum_{\ell=0}^{L-1}|\ell\rangle\!\langle\ell|_{T}\otimes
\mathcal{N}^{\ell}(\rho_{C}),\\
\frac{1}{L}\sum_{\ell=0}^{L-1}|\ell\rangle\!\langle\ell|_{T}\otimes
|1\rangle\!\langle1|_{D}\otimes\sigma_{G}%
\end{array}
\right) \\
=\left[  \sup_{\sigma_{G}}\frac{1}{L}\sum_{\ell=0}^{L-1}\sqrt{F}\!\left(
\begin{array}
[c]{c}%
|1\rangle\!\langle1|_{D}\otimes\sigma_{G},\\
U_{SC}(|0\rangle\!\langle0|_{S}\otimes\mathcal{N}^{\ell}(\rho_{C}%
))U_{SC}^{\dag}%
\end{array}
\right)  \right]  ^{2}.
\end{multline}
We conclude the expression in the statement of the theorem after optimizing
over all input states of the prover.
\end{proof}

\bigskip

In order to establish that $\operatorname{BQP}_{\operatorname{CTC}}$ is
contained in $\operatorname{QIP}$, it is necessary to map yes-instances of the
former to yes-instances of the latter, and the same for the no-instances.
Accomplishing the first part of the task is straightforward. A yes-instance of
$\operatorname{BQP}_{\operatorname{CTC}}$ implies that there exists a
fixed-point state $\rho_{C}$ such that
\begin{equation}
\operatorname{Tr}\left[  (|1\rangle\!\langle1|_{D}\otimes I_{G})U_{SC}%
(|0\rangle\!\langle0|_{S}\otimes\rho_{C})U_{SC}^{\dag}\right]  \geq1-\delta.
\end{equation}
Thus, the prover transmits such a fixed-point state $\rho_{C}$ to the
verifier, and we find that the acceptance probability is not smaller than%
\begin{align}
&  \left[  \sup_{\sigma_{G}}\frac{1}{L}\sum_{\ell=0}^{L-1}\sqrt{F}\!\left(
\begin{array}
[c]{c}%
|1\rangle\!\langle1|_{D}\otimes\sigma_{G},\\
U_{SC}(|0\rangle\!\langle0|_{S}\otimes\mathcal{N}^{\ell}(\rho_{C}%
))U_{SC}^{\dag}%
\end{array}
\right)  \right]  ^{2}\nonumber\\
&  \geq\sup_{\sigma_{G}}F\!\left(  |1\rangle\!\langle1|_{D}\otimes\sigma
_{G},U_{SC}(|0\rangle\!\langle0|_{S}\otimes\rho_{C})U_{SC}^{\dag}\right)  \\
&  \geq\sup_{|\varphi\rangle_{F^{\prime}G}}\left\vert \langle1|_{D}%
\langle\varphi|_{F^{\prime}G}U_{SC}|0\rangle_{S}|\psi\rangle_{RC}\right\vert
^{2}\\
&  =\left\Vert \langle1|_{D}U_{SC}|0\rangle_{S}|\psi\rangle_{RC}\right\Vert
_{2}^{2}\\
&  =\operatorname{Tr}\left[  (|1\rangle\!\langle1|_{D}\otimes I_{G}%
)U_{SC}(|0\rangle\!\langle0|_{S}\otimes\rho_{C})U_{SC}^{\dag}\right]  \\
&  \geq1-\delta.
\end{align}
The first inequality follows because $\mathcal{N}^{\ell}(\rho_{C})=\rho_{C}$
for all$~\ell$.

It is less clear how to handle the case of a no-instance of
$\operatorname{BQP}_{\operatorname{CTC}}$, because the definition of this
complexity class only specifies the behavior of the circuit when $\rho_{C}$ is
an exact fixed point of $\mathcal{N}$. Algorithm~\ref{alg:CTC-alg-complicated}
attempts to verify whether the prover sends a fixed point, but it only
actually verifies whether the prover sends a state that is an approximate
fixed point. The acceptance probability of
Algorithm~\ref{alg:CTC-alg-complicated}\ is given by
\eqref{eq:acc-prob-BQP-CTC-sim} and is bounded from above by%
\begin{multline}
\sup_{\rho_{C},\sigma_{G}}F\!\left(  |1\rangle\!\langle1|_{D}\otimes\sigma
_{G},U_{SC}(|0\rangle\!\langle0|_{S}\otimes\overline{\mathcal{N}_{L}}(\rho
_{C}))U_{SC}^{\dag}\right)  \\
\leq\sup_{\rho_{C}}\langle1|_{D}\operatorname{Tr}_{G}[U_{SC}(|0\rangle
\!\langle0|_{S}\otimes\overline{\mathcal{N}_{L}}(\rho_{C}))U_{SC}^{\dag
}]|1\rangle_{D},
\end{multline}
where the bounds follow from concavity of root fidelity and the
data-processing inequality for fidelity. In the above, $\overline
{\mathcal{N}_{L}}$ is the Cesaro mean channel:%
\begin{equation}
\overline{\mathcal{N}_{L}}(\omega_{C})\coloneqq \frac{1}{L}\sum_{\ell=0}%
^{L-1}\mathcal{N}^{\ell}(\omega_{C}).
\end{equation}
This channel has the property that the sequence $\{\overline{\mathcal{N}_{L}%
}\}_{L}$ converges to the fixed-point projection channel $\mathcal{P}%
\coloneqq \lim_{L\rightarrow\infty}\overline{\mathcal{N}_{L}}$ of $\mathcal{N}$, so
that $\mathcal{P}(\omega_{C})$ is guaranteed to be a fixed point of
$\mathcal{N}$ for every input state $\omega_{C}$ \cite{Wolf12}. It is not
clear how to obtain a channel independent bound that relates the convergence
of $\overline{\mathcal{N}_{L}}$\ to $\mathcal{P}$, as a function of $L$ alone.
Furthermore, it is likely not possible that the closeness of $\overline
{\mathcal{N}_{L}}$\ to $\mathcal{P}$ could generally be inverse polynomial in
$L$; for if it were, then one could simulate $\operatorname{BQP}%
_{\operatorname{CTC}}$ in BQP, because the verifier could apply the map
$\overline{\mathcal{N}_{L}}$ and generate a fixed point of the channel without
the help of the prover. However, we now know that $\operatorname{BQP}%
_{\operatorname{CTC}}=\operatorname{PSPACE}$ \cite{AW09}, and it is widely
believed that $\operatorname{PSPACE}\neq\operatorname{BQP}$. In the case of a
no-instance of $\operatorname{BQP}_{\operatorname{CTC}}$, it is thus not clear
how to relate the acceptance probability of
Algorithm~\ref{alg:CTC-alg-complicated} to the acceptance probability of the
no-instance of $\operatorname{BQP}_{\operatorname{CTC}}$. We leave this as a
curious open question.

\end{document}